\documentclass[draft,journal,a4paper,oneside,onecolumn]{IEEEtran}
\usepackage{amsmath,amsfonts, amssymb}

\newtheorem{thm}{Theorem}

\newtheorem{lem}{Lemma}
\newtheorem{df}{Definition}
\newtheorem{rem}{Remark}

\newcommand{\lrsb}[1]{\left({#1}\right)}
\newcommand{\lrb}[1]{\left\{{#1}\right\}}
\newcommand{\lrB}[1]{\left[{#1}\right]}
\newcommand{\lrceil}[1]{\left\lceil{#1}\right\rceil}
\newcommand{\lrfloor}[1]{\left\lfloor{#1}\right\rfloor}
\newcommand{\A}{\mathcal{A}}
\newcommand{\B}{\mathcal{B}}
\newcommand{\C}{\mathcal{C}}
\newcommand{\E}{\mathcal{E}}
\newcommand{\G}{\mathcal{G}}
\newcommand{\cH}{\mathcal{H}}
\newcommand{\cS}{\mathcal{S}}
\newcommand{\T}{\mathcal{T}}
\newcommand{\U}{\mathcal{U}}
\newcommand{\bU}{\overline{\mathcal{U}}}
\newcommand{\V}{\mathcal{V}}
\newcommand{\X}{\mathcal{X}}
\newcommand{\Y}{\mathcal{Y}}
\newcommand{\Z}{\mathcal{Z}}
\newcommand{\W}{\mathcal{W}}
\newcommand{\ba}{\boldsymbol{a}}
\newcommand{\bb}{\boldsymbol{b}}
\newcommand{\cc}{\boldsymbol{c}}
\newcommand{\hcc}{\widehat{\boldsymbol{c}}}
\newcommand{\mm}{\boldsymbol{m}}
\newcommand{\bt}{\boldsymbol{t}}
\newcommand{\uu}{\boldsymbol{u}}
\newcommand{\vv}{\boldsymbol{v}}
\newcommand{\ww}{\boldsymbol{w}}
\newcommand{\xx}{\boldsymbol{x}}
\newcommand{\yy}{\boldsymbol{y}}
\newcommand{\zz}{\boldsymbol{z}}
\newcommand{\aalpha}{\boldsymbol{\alpha}}
\newcommand{\bbeta}{\boldsymbol{\beta}}
\newcommand{\kkappa}{\boldsymbol{\kappa}}
\newcommand{\zero}{\boldsymbol{0}}
\newcommand{\one}{\boldsymbol{1}}
\newcommand{\hg}{\widehat{g}}
\newcommand{\hA}{\widehat{A}}
\newcommand{\hB}{\widehat{B}}
\newcommand{\hC}{\widehat{C}}
\newcommand{\hcA}{\widehat{\mathcal{A}}}
\newcommand{\hcB}{\widehat{\mathcal{B}}}
\newcommand{\hcC}{\widehat{\mathcal{C}}}
\newcommand{\hcH}{\widehat{\mathcal{H}}}
\newcommand{\e}{\varepsilon}
\newcommand{\eA}{\e_{\A}}
\newcommand{\eB}{\e_{\B}}
\newcommand{\ehA}{\e_{\hcA}}
\newcommand{\ehB}{\e_{\hcB}}
\newcommand{\limn}{\lim_{n\to\infty}}
\newcommand{\GF}{\mathrm{GF}}
\newcommand{\GFq}{\mathrm{GF}(q)}
\newcommand{\GFql}{\mathrm{GF}(q^l)}
\newcommand{\dft}{\mathfrak{F}}
\newcommand{\dfti}{\mathfrak{F}^{-1}}
\newcommand{\Encoder}{\varphi}
\newcommand{\Decoder}{\varphi^{-1}}
\newcommand{\Capacity}{\mathrm{Capacity}}
\newcommand{\Rate}{R}
\newcommand{\Error}{\mathrm{Error}}
\newcommand{\im}{\mathrm{Im}}

\allowdisplaybreaks

\title{
Hash Property
and Coding Theorems
for Sparse Matrices
and
Maximum-Likelihood Coding
}
\author{
Jun~Muramatsu
and~Shigeki Miyake
 \thanks{J.~Muramatsu is with
        NTT Communication Science Laboratories, NTT Corporation,
        2-4, Hikaridai, Seika-cho, Soraku-gun, Kyoto 619-0237, Japan
        (E-mail: pure@cslab.kecl.ntt.co.jp).
        S.~Miyake is with
        NTT Network Innovation Laboratories, NTT Corporation,
        1-1, Hikarinooka, Yokosuka-shi, Kanagawa 239-0847, Japan
        (E-mail: miyake.shigeki@lab.ntt.co.jp).
  }
	 \thanks{This paper is submitted to {\em IEEE Transactions on
			 Information Theory} and a part of this paper is submitted to
			 {\em IEEE International Symposium on Information Theory
			 (ISIT2008, ISIT2009)}.
			 }
}
  
\begin{document}
\maketitle

\begin{abstract}
The aim of this paper is to prove the
achievability of several coding problems
by using sparse matrices
(the maximum column weight grows logarithmically in the block length)
and maximal-likelihood (ML) coding.
These problems are the Slepian-Wolf problem,
the Gel'fand-Pinsker problem, the Wyner-Ziv problem, and
the One-helps-one
problem (source coding with partial side information at the decoder).
To this end, the notion of a hash property for an ensemble of
functions is introduced and it is proved that
an ensemble of $q$-ary sparse matrices satisfies the hash property.
Based on this property,
it is proved that the rate of codes using sparse matrices and
maximal-likelihood (ML) coding can achieve the optimal rate.
\end{abstract}
\begin{keywords}
Shannon theory, hash functions, linear codes,
sparse matrix, maximum-likelihood eoncoding/decoding,
the Slepian-Wolf problem,
the Gel'fand-Pinsker problem, the Wyner-Ziv problem,
the One-helps-one problem
\end{keywords}

\section{Introduction}
The aim of this paper is to prove the
achievability of several coding problems
by using sparse matrices
(the maximum column weight grows logarithmically in the block length)
and maximal-likelihood (ML) coding
\footnote{This operation is usually called an {\em ML decoding}.
We use the word `coding' because this operation is also used in the
construction of an encoder.},
namely the
Slepian-Wolf problem \cite{SW73} (Fig. \ref{fig:sw}),
the Gel'fand-Pinsker problem \cite{GP83} (Fig. \ref{fig:gp}),
the Wyner-Ziv problem \cite{WZ76} (Fig. \ref{fig:wz}), and
the One-helps-one
problem
(source coding with partial side information at the decoder)
\cite{W73}\cite{WZ73} (Fig. \ref{fig:psi}).
To prove these theorems,
we first introduce the notion of a hash property for an ensemble of
functions, where functions are not assumed to be linear.
This notion is a sufficient condition for the achievability
of coding theorems.
Next, we prove that
an ensemble of $q$-ary sparse matrices,
which is an extension of \cite{Mac99}, satisfies the hash property.
Finally, based on the hash property, we prove that
the rate of codes can achieve the optimal rate.
This implies that the rate of codes using sparse matrices and ML 
coding can achieve the optimal rate.
It should be noted here that
there is a practical approximation method
of ML coding by using sparse matrices and 
the linear programing technique
introduced by \cite{FWK05}.

\begin{figure}[b]
\begin{center}
\unitlength 0.55mm
\begin{picture}(153,40)(0,0)
\put(5,25){\makebox(0,0){$X$}}
\put(10,25){\vector(1,0){10}}
\put(20,18){\framebox(14,14){$\Encoder_{X}$}}
\put(34,25){\vector(1,0){80}}
\put(74,29){\makebox(0,0){$R_{X}> H(X|Y)$}}
\put(5,7){\makebox(0,0){$Y$}}
\put(10,7){\vector(1,0){10}}
\put(20,0){\framebox(14,14){$\Encoder_{Y}$}}
\put(34,7){\vector(1,0){80}}
\put(74,11){\makebox(0,0){$R_{Y}> H(Y|X)$}}
\put(114,0){\framebox(14,32){$\Decoder$}}
\put(128,16){\vector(1,0){10}}
\put(148,16){\makebox(0,0){$XY$}}
\put(77,-8){\makebox(0,0){$R_{X}+R_{Y}> H(X,Y)$}}
\end{picture}
\end{center}
\caption{Slepian-Wolf Problem}
\label{fig:sw}
\end{figure}
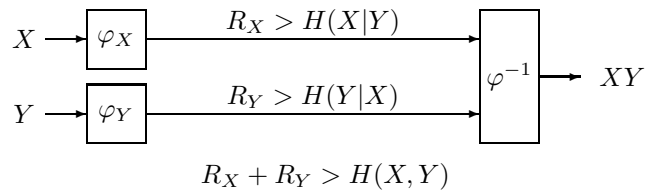

\begin{figure}[t]
\begin{center}
\unitlength 0.55mm
\begin{picture}(150,30)(0,-20)
\put(4,7){\makebox(0,0){$M$}}
\put(10,7){\vector(1,0){10}}
\put(20,0){\framebox(14,14){$\Encoder$}}
\put(34,7){\vector(1,0){10}}
\put(50,7){\makebox(0,0){$X$}}
\put(56,7){\vector(1,0){10}}
\put(66,0){\framebox(24,14){$\mu_{Y|XZ}$}}
\put(90,7){\vector(1,0){10}}
\put(106,7){\makebox(0,0){$Y$}}
\put(112,7){\vector(1,0){10}}
\put(122,0){\framebox(14,14){$\Decoder$}}
\put(136,7){\vector(1,0){10}}
\put(150,7){\makebox(0,0){$M$}}
\put(4,-10){\makebox(0,0){$Z$}}
\put(10,-10){\line(1,0){68}}
\put(27,-10){\vector(0,1){10}}
\put(78,-10){\vector(0,1){10}}
\put(75,-20){\makebox(0,0){$R<I(W;Y)-I(W;Z)$}}
\end{picture}
\end{center}
\caption{Gel'fand-Pinsker Problem}
\label{fig:gp}
\end{figure}
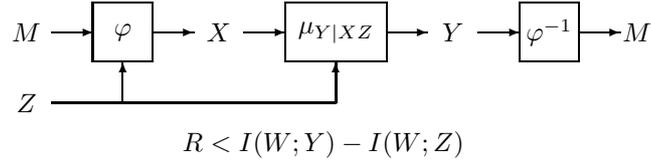

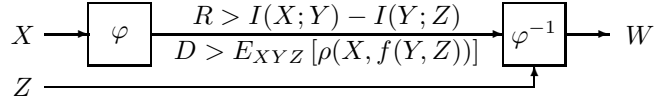
\begin{figure}[t]
\begin{center}
\unitlength 0.58mm
\begin{picture}(150,14)(0,0)
\put(5,7){\makebox(0,0){$X$}}
\put(10,7){\vector(1,0){10}}
\put(20,0){\framebox(14,14){$\Encoder$}}
\put(34,7){\vector(1,0){80}}
\put(74,11){\makebox(0,0){$R>I(X;Y)-I(Y;Z)$}}
\put(114,0){\framebox(14,14){$\Decoder$}}
\put(128,7){\vector(1,0){10}}
\put(145,7){\makebox(0,0){$W$}}
\put(74,3){\makebox(0,0){$D>E_{XYZ}\lrB{\rho(X,f(Y,Z))}$}}
\put(5,-5){\makebox(0,0){$Z$}}
\put(10,-5){\line(1,0){111}}
\put(121,-5){\vector(0,1){5}}
\end{picture}
\end{center}
\caption{Wyner-Ziv Problem}
\label{fig:wz}
\end{figure}

\begin{figure}[t]
\begin{center}
\unitlength 0.55mm
\begin{picture}(153,30)(0,0)
\put(5,25){\makebox(0,0){$X$}}
\put(10,25){\vector(1,0){10}}
\put(20,18){\framebox(14,14){$\Encoder_{X}$}}
\put(34,25){\vector(1,0){80}}
\put(74,29){\makebox(0,0){$R_X>H(X|Z)$}}
\put(5,7){\makebox(0,0){$Y$}}
\put(10,7){\vector(1,0){10}}
\put(20,0){\framebox(14,14){$\Encoder_{Y}$}}
\put(34,7){\vector(1,0){80}}
\put(74,11){\makebox(0,0){$R_Y>I(Y;Z)$}}
\put(114,0){\framebox(14,32){$\Decoder$}}
\put(128,16){\vector(1,0){10}}
\put(148,16){\makebox(0,0){$X$}}
\end{picture}
\end{center}
\caption{One-helps-one Problem}
\label{fig:psi}
\end{figure}
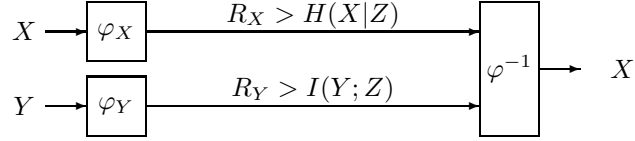

The contributions of this paper are summarized in the following.
\begin{itemize}
 \item The notion of a hash property is introduced.
			 It is an extension of the notion of
			 a universal class of
			 hash functions introduced in \cite{CW}.
			 The single source coding problem is studied in
			 \cite[Section 14.2]{MacKay}\cite{K04}
			 by using the hash function.
			 We prove that
			 an ensemble of $q$-ary sparse matrices has a hash property,
			 while a weak version of hash property
			 is proved in
			 \cite{MB01}\cite{BB}\cite{EM05}\cite{SWLDPC}\cite{CRYPTLDPC}
			 implicitly.
			 It should be noted that 
			 our definition of hash property is 
			 also an extension of the definition of random bin coding
			 introduced in \cite{C75},
			 where the set of all sequences is partitioned at random.
			 On the other hand,
			 the random codebook (a set of codewords/representations)
			 generation is introduced for
			 the proof of the original
			 channel coding theorem \cite{SHANNON}
			 and lossy source coding theorem \cite{SHA59}.
			 Here it is proved that random bin coding
			 and partitioning determined by randomly generated matrix
			 can be applied to
			 the original
			 channel coding theorem and lossy source coding theorem.
 \item 
			 The proof of the achievability of the Slepian-Wolf problem
			 is demonstrated based on the hash property.
			 It is the extension of \cite[Section 14.2]{MacKay}\cite{K04}
			 and provides a new proof of \cite{C75}\cite{CSI82}\cite{SWLDPC}.
			 By applying the theorem to 
			 the coding theorem of channel with additive (symmetric) noise,
			 it also provides a new proof of
			 \cite{MB01}\cite{BB}\cite{EM05}\cite{SWLDPC}.
 \item
			 The optimality of a code is proved by using sparse matrices and ML
			 coding for the Gel'fand-Pinsker problem.
			 We prove the $q$-ary and asymmetric version of the theorem,
			 while a binary and symmetric version is studied in \cite{MW06b}.
			 It should be noted here that
			 the column/row weight of matrices used in \cite{MW06b}
			 is constant with respect to the block length,
			 while it grows logarithmically in our construction.
			 The detailed difference from \cite{MW06b} is stated in
			 Section~\ref{sec:gp}.
			 As a corollary, we have the optimality
			 of codes using sparse matrices
			 for the coding problem of an arbitrary ($q$-ary and asymmetric) 
			 channel, while a symmetric channel is assumed in
			 many of the channel coding theorems by using sparse matrices.
			 The construction is based on the coset code presented in
			 \cite{SWLDPC}\cite{MM08},
			 which is different from
			 that presented in \cite{GA68}\cite{BB}.
			 When our theorem is applied to the ensemble
			 of sparse matrices,
			 our proof is simpler than that in \cite{MM08}.
 \item 
			 The optimality of a code is proved
			 by using sparse matrices and ML coding
			 for the Wyner-Ziv problem.
			 We prove the $q$-ary and biased source, and the non-additive side
			 information version of the theorem,
			 while a binary and unbiased source and additive side information
			 are assumed in \cite{MW06b}.
			 As a corollary, we have the optimality
			 of codes using sparse matrices
			 for the lossy coding problem of an arbitrary ($q$-ary
			 and biased) source and a distortion measure.
			 In \cite{MY03}\cite{M04}\cite{MW06a}\cite{M06}\cite{GV07},
			 a lossy code is proposed
			 by using sparse matrices called low density
			 generator matrices (LDGM)
			 by assuming an unbiased source and Hamming distortion.
			 The column/row weight of matrices used in \cite{MW06b}
			 is constant with respect to the block length,
			 while it grows logarithmically in our construction.
			 The lower bounds on the rate-distortion function
			 is discussed in \cite{DWR07}\cite{KU08}.
			 It should be noted that the construction of the codes
			 is different from those presented in
			 \cite{MY03}\cite{M04}\cite{MW06a}\cite{MW06b}
			 \cite{M06}\cite{GV07}.
			 The detailed difference is stated in Section~\ref{sec:wz}.
			 Our construction is based on the code presented in
			 \cite{MM07}\cite{SWLOSSY}
			 and similar to the code presented in
			 \cite{WAD03aew}\cite{WAD03ieice}\cite{ZSE02}.
			 When our theorem is applied to the ensemble
			 of sparse matrices,
			 our
			 proof is simpler than that in \cite{MM07}.
 \item The achievability of the One-helps-one problem
			 is proved by using sparse matrices and ML coding.
\end{itemize}

\section{Definitions and Notations}
Throughout this paper, we use the following definitions and notations.

Column vectors and sequences are denoted in boldface.
Let $A\uu$ denote a value taken by a function $A:\U^n\to\bU$ at $\uu\in\U^n$
where $\U^n$ is a domain of the function.
It should be noted that $A$ may be nonlinear.
When $A$ is a linear function
expressed by an $l\times n$ matrix,
we assume that $\U\equiv\GFq$
is a finite field and the range of functions is defined by $\bU\equiv\U^{l}$.
It should be noted that this assumption is not essential
for general (nonlinear) functions
because discussion is not changed if $l\log|\U|$ is replaced by $\log|\bU|$.
For a set $\A$ of functions,
let $\im \A$ be defined as
\begin{align*}
 \im\A &\equiv \bigcup_{A\in\A}\{A\uu: \uu\in\U^n\}.
\end{align*}
The cardinality of a set $\U$ is denoted by $|\U|$,
$\U^c$ denotes the compliment of $\U$,
and 
$\U\setminus\V\equiv\U\cap\V^c$ denotes the
set difference.
We define sets $\C_A(\cc)$ and $\C_{AB}(\cc,\bb)$ as
\begin{align*}
 \C_A(\cc)
 &\equiv\{\uu: A\uu = \cc\}
 \\
 \C_{AB}(\cc,\bb)
 &\equiv\{\uu: A\uu = \cc, B\uu = \bb\}.
\end{align*}
In the context of linear codes,
$\C_A(\cc)$ is called a coset determined by $\cc$.

Let $p$ and $p'$ be probability distributions
and let $q$ and $q'$ be conditional probability distributions.
Then entropy $H(p)$, conditional entropy $H(q|p)$,
divergence $D(p\|p')$, and conditional divergence $D(q\|q'|p)$
are defined as
\begin{align*}
 H(p)
 &\equiv\sum_{u}p(u)\log\frac 1{p(u)}
 \\
 H(q|p)
 &\equiv\sum_{u,v}q(u|v)p(v)\log\frac 1{q(u|v)}
 \\
 D(p\parallel p')
 &\equiv
 \sum_{u}p(u)
 \log\frac{p(u)}{p'(u)}
 \\
 D(q\parallel q' | p)
 &\equiv
 \sum_{v} p(v)\sum_{u}q(u|v)
 \log\frac{q(u|v)}{q'(u|v)},
\end{align*}
where we assume the base $2$ of the logarithm when the subscript of $\log$
is omitted.

Let $\mu_{UV}$ be the joint probability distribution of random variables
$U$ and $V$.
Let  $\mu_{U}$ and $\mu_{V}$ be the respective marginal distributions
and $\mu_{U|V}$ be the conditional probability distribution.
Then the entropy $H(U)$, the conditional entropy $H(U|V)$, and the mutual
information $I(U;V)$ of random variables are defined as
\begin{align*}
 H(U)&\equiv H(\mu_U)
 \\
 H(U|V)&\equiv H(\mu_{U|V}|\mu_{V})
 \\
 I(U;V)&\equiv H(U)-H(U|V).
\end{align*}

A set of typical sequences $\T_{U,\gamma}$
and a set of conditionally typical sequences $\T_{U|V,\gamma}(\vv)$
are defined as
\begin{align*}
 \T_{U,\gamma}
 &\equiv
 \lrb{\uu:
 D(\nu_{\uu}\|\mu_{U})<\gamma
 }
 \\
 \T_{U|V,\gamma}(\vv)
 &\equiv
 \lrb{\uu:
 D(\nu_{\uu|\vv}\|\mu_{U|V}|\nu_{\vv})<\gamma
 },
\end{align*}
respectively,
where $\nu_{\uu}$ and $\nu_{\uu|\vv}$ are defined as
\begin{align*}
 &\nu_{\uu}(u)
 \equiv
 \frac {|\{1\leq i\leq n : u_{i}=u\}|}n
 \\
 &\nu_{\uu|\vv}(u|v)
 \equiv \frac{\nu_{\uu\vv}(u,v)}{\nu_{\vv}(v)}.
\end{align*}

We define $\chi(\cdot)$ as
\begin{align*}
 \chi(a = b)
	&\equiv
 \begin{cases}
	1,&\text{if}\ a = b
	\\
	0,&\text{if}\ a\neq b
	\end{cases}
	\\
	\chi(a \neq b)
	&\equiv
 \begin{cases}
	1,&\text{if}\ a \neq b
	\\
	0,&\text{if}\ a = b.
	\end{cases}
\end{align*}

Finally, for $\gamma,\gamma'>0$, we define
\begin{align}
 \lambda_{\U}
 &\equiv \frac{|\U|\log[n+1]}n
 \label{eq:lambda}
 \\
 \zeta_{\U}(\gamma)
 &\equiv
 \gamma-\sqrt{2\gamma}\log\frac{\sqrt{2\gamma}}{|\U|}
 \label{eq:zeta}
 \\
 \zeta_{\U|\V}(\gamma'|\gamma)
 &\equiv
 \gamma'-\sqrt{2\gamma'}\log\frac{\sqrt{2\gamma'}}{|\U||\V|}
 +\sqrt{2\gamma}\log|\U|
 \label{eq:zetac}
 \\
 \eta_{\U}(\gamma)
 &\equiv
 -\sqrt{2\gamma}\log\frac{\sqrt{2\gamma}}{|\U|}
 +\frac{|\U|\log[n+1]}n
 \label{eq:def-eta}
 \\
 \eta_{\U|\V}(\gamma'|\gamma)
 &\equiv
 -\sqrt{2\gamma'}\log\frac{\sqrt{2\gamma'}}{|\U||\V|}
 +\sqrt{2\gamma}\log|\U|
 +\frac{|\U||\V|\log[n+1]}n.
 \label{eq:def-etac}
\end{align}
It should be noted here that
the product set $\U\times\V$ is denoted by $\U\V$
when it appears in the subscript of these functions.

\section{$(\aalpha,\bbeta)$-hash Property}

In this section, we introduce the notion of the
$(\aalpha,\bbeta)$-hash property which is a sufficient condition
for coding theorems, where the linearity of functions is not assumed.
The $(\aalpha,\bbeta)$-hash property of an ensemble of linear (sparse)
matrices will be discussed in Section \ref{sec:linear}.
In Section~\ref{sec:coding-theorems},
we provide coding theorems
for the Slepian-Wolf problem, the Gel'fand-Pinsker problem, the
Wyner-Ziv problem,
and the One-helps-one problem.

Before stating the formal definition,
we explain the random coding arguments
and two implications
which introduce the intuition of the hash property.

\subsection{Two types of random coding}

We review the random coding argument
introduced by \cite{SHANNON}.
Most of coding theorems
are proved by using the combination of the following two types of random
coding.

\noindent{\bf Random codebook generation:}
In the proof of the original channel coding theorem	\cite{SHANNON}
and lossy source coding theorem	\cite{SHA59},
a codebook (a set of codewords/representations)
is randomly generated and shared by the encoder and the decoder.
It should be noted that
the randomly generated  codebook represents
the list of typical sequences.
In the encoding step of channel coding,
a message is mapped to a member of randomly generated
codewords as a channel input.
In the decoding step,
we use the maximum-likelihood decoder,
which guess the most probable channel input	from the channel output.
In the encoding step of lossy source coding,
we find a member of randomly generated representations
to satisfy the fidelity criterion compared to a message,
and then we let the index of this member as a codeword.
In the decoding step,
the codeword (index) is mapped to the reproduction.
It should be noted that
the encoder and the decoder have to share
the large size (exponentially in the block length)
of table which indicates the correspondence
between a index and a member of randomly generated
codewords/representations.
The time complexity of 
the encoding and decoding step of channel coding
and the encoding step of lossy source coding
is exponentially in the block	length.
They are obstacles for the implementation.

\noindent{\bf Random partitioning (random bin coding):}
In the proof of Slepian-Wolf problem \cite{C75},
the set of all sequences is partitioned at random
and shared by the encoder and the decoder.
In the encoding step,
a pair of messages are mapped independently to
the index of bin which contains the message.
In the decoding step,
we use the maximum-likelihood decoder,
which guess the most probable pair of messages.
Random partitioning
by cosets determined by randomly generated matrix
can be considered as a kind of random bin coding,
where the syndrome corresponds the index of bin.
This approach is introduced in \cite{E55}
for the coding of symmetric channel
and applied to the ensemble of sparse matrices
in \cite{MB01}\cite{BB}\cite{EM05}.
This argument is also applied to
the coding theorem for Slepian-Wolf problem
in \cite{W74}\cite{CSI82}\cite{SWLDPC}.
It should be noted that
the time complexity of the decoding step
is exponentially in the block length,
but there are practical approximation methods
by using sparse matrices and 
the techniques introduced by
\cite{GDL}\cite{KFL01}\cite{FWK05}.
By using the randomly generated matrix,
the size of tables shared by the encoder and the decoder
has at most square order with respect to the block length.

One of the aim of introducing hash property
is to replace the random codebook generation by the random
partitioning.
In other words, it is a unification of these two random coding arguments.
It is expected that the space and time complexity
can be reduced compared to the random codebook generation.

\subsection{Two implications of hash property}

We introduce the following two implications of hash property,
which connect the number of bins and messages (items)
and is essential for the coding by using the random partitioning.
In Section~\ref{sec:basic}, these two property are derived from the
hash property by adjusting the number of bins
taking account of the number of sequences.

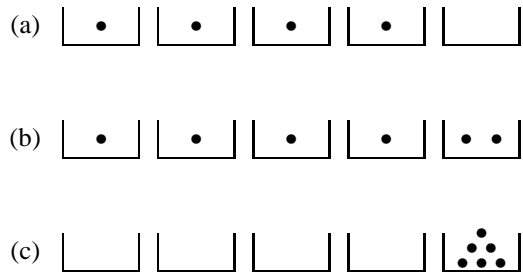
\begin{figure}[t]
\begin{center}
\unitlength = 0.5mm
\begin{picture}(140,30)(-20,5)
\put(-10,15){\makebox(0,0){(a)}}
\put(0,10){\line(1,0){20}}
\put(0,10){\line(0,1){10}}
\put(20,10){\line(0,1){10}}
\put(25,10){\line(1,0){20}}
\put(25,10){\line(0,1){10}}
\put(45,10){\line(0,1){10}}
\put(50,10){\line(1,0){20}}
\put(50,10){\line(0,1){10}}
\put(70,10){\line(0,1){10}}
\put(75,10){\line(1,0){20}}
\put(75,10){\line(0,1){10}}
\put(95,10){\line(0,1){10}}
\put(100,10){\line(1,0){20}}
\put(100,10){\line(0,1){10}}
\put(120,10){\line(0,1){10}}
\put(10,15){\makebox(0,0){{\small $\bullet$}}}
\put(35,15){\makebox(0,0){{\small $\bullet$}}}
\put(60,15){\makebox(0,0){{\small $\bullet$}}}
\put(85,15){\makebox(0,0){{\small $\bullet$}}}
\end{picture}
\\
\begin{picture}(140,30)(-20,5)
\put(-10,15){\makebox(0,0){(b)}}
\put(0,10){\line(1,0){20}}
\put(0,10){\line(0,1){10}}
\put(20,10){\line(0,1){10}}
\put(25,10){\line(1,0){20}}
\put(25,10){\line(0,1){10}}
\put(45,10){\line(0,1){10}}
\put(50,10){\line(1,0){20}}
\put(50,10){\line(0,1){10}}
\put(70,10){\line(0,1){10}}
\put(75,10){\line(1,0){20}}
\put(75,10){\line(0,1){10}}
\put(95,10){\line(0,1){10}}
\put(100,10){\line(1,0){20}}
\put(100,10){\line(0,1){10}}
\put(120,10){\line(0,1){10}}
\put(10,15){\makebox(0,0){{\small $\bullet$}}}
\put(35,15){\makebox(0,0){{\small $\bullet$}}}
\put(60,15){\makebox(0,0){{\small $\bullet$}}}
\put(85,15){\makebox(0,0){{\small $\bullet$}}}
\put(106,15){\makebox(0,0){{\small $\bullet$}}}
\put(114,15){\makebox(0,0){{\small $\bullet$}}}
\end{picture}
\\
\begin{picture}(140,30)(-20,5)
\put(-10,15){\makebox(0,0){(c)}}
\put(0,10){\line(1,0){20}}
\put(0,10){\line(0,1){10}}
\put(20,10){\line(0,1){10}}
\put(25,10){\line(1,0){20}}
\put(25,10){\line(0,1){10}}
\put(45,10){\line(0,1){10}}
\put(50,10){\line(1,0){20}}
\put(50,10){\line(0,1){10}}
\put(70,10){\line(0,1){10}}
\put(75,10){\line(1,0){20}}
\put(75,10){\line(0,1){10}}
\put(95,10){\line(0,1){10}}
\put(100,10){\line(1,0){20}}
\put(100,10){\line(0,1){10}}
\put(120,10){\line(0,1){10}}
\put(105,12){\makebox(0,0){{\small $\bullet$}}}
\put(110,12){\makebox(0,0){{\small $\bullet$}}}
\put(115,12){\makebox(0,0){{\small $\bullet$}}}
\put(112.5,16){\makebox(0,0){{\small $\bullet$}}}
\put(107.5,16){\makebox(0,0){{\small $\bullet$}}}
\put(110,20){\makebox(0,0){{\small $\bullet$}}}
\end{picture}
\end{center}
\caption{
Properties connecting the number of bins and items (black dots, messages).
(a) Collision-resistant property: every bin contains at most one item.
(b) Saturating property: every bin contains at least one item.
(c) Pigeonhole principle: there is at least one bin which contains two or more items.
}
\label{fig:principle}
\end{figure}

\noindent{\bf Collision-resistant property:}
The good code assigns a message to a codeword
which is different from the codewords of other messages,
where the loss (error probability) is as small as possible.
The collision-resistant property is the nature of the hash property.
Figure \ref{fig:principle} (a) represents the ideal situation
of this property, where
the black dots represent messages we want to distinguish.
When the number of bins is greater than the number of black dots,
we can find a good function
that allocates the black dots to the different bins.
It is because the hash property tends to avoid the collision.
It should be noted that it is enough for coding problems
to satisfy this property for `almost all (close to probability one)'
black dots by letting
the ratio
 $[\text{the number of black dots}]/[\text{the number of bins}]$
close to zero.
This property is used for the estimation of decoding error probability
of lossless source coding by using maximum-likelihood decoder.
In this situation, the black dots correspond to the typical sequences.

\noindent{\bf Saturating property:}
To replace the random codebook generation by the random partitioning,
we prepare the method finding a typical sequence for each bin.
The saturating property is the another nature of the hash property.
Figure \ref{fig:principle} (b) represents the ideal situation
of this property.
When the number of bins is smaller than the number of black dots,
we can find a good function
so that every bins has at least one black dot.
It is because the hash property tends to avoid the collision.
It should be noted that this property is different
from the pigeonhole principle: there is at least one
bin which includes two or more black dots.
Figure \ref{fig:principle} (c) represents an unusual situation,
which does not contradict by the pigeonhole principle
while the hash property tends to avoid this situation.
It should be noted that it is enough for coding problems
to satisfy this property for `almost all (close to probability one)' bins
by letting the ratio
 $[\text{the number of bins}]/[\text{the number of black dots}]$
close to zero.
To find a typical sequence from each bin,
we use the maximum-likelihood/minimum-divergence coding
introduced in Section~\ref{sec:basic}.
In this situation, the black dots correspond to the typical sequences.

\subsection{Formal definition of $(\aalpha,\bbeta)$-hash property}

In this section, we introduce the formal definition of the hash property.

In the proof of the fixed-rate source coding theorem given in
\cite{C75}\cite{CSI82}\cite{K04},
it is proved implicitly that there is a probability distribution $p_A$
on a set of functions $A:\U^n\to\U^{l}$ such that
\begin{equation}
 p_A\lrsb{\lrb{
 A: \exists\uu'\in\G\setminus\{\uu\}, A\uu'= A\uu
 }}
 \leq
 \frac{|\G|}{|\U|^{l}}
 \label{eq:bin}
\end{equation}
for any $\uu\in\U^n$, where
\begin{equation}
 \G\equiv\lrb{\uu': \mu(\uu')\geq \mu(\uu), \uu\neq\uu'}
	\label{eq:G}
\end{equation}
and $\mu$ is the probability distribution of a source
or the probability distribution of the additive noise of a channel.
In the proof of coding theorems for spare matrices
given in \cite{BB}\cite{EM05}\cite{MB01}\cite{SWLDPC}\cite{CRYPTLDPC},
it is proved implicitly that there are
a probability distribution on a set of $l\times n$ spare matrices
and
sequences
$\aalpha\equiv\{\alpha(n)\}_{n=1}^{\infty}$
and $\bbeta\equiv\{\beta(n)\}_{n=1}^{\infty}$
satisfying
\begin{gather*}
 \limn \frac{\log\alpha(n)}n=0
 \\
 \limn \beta(n)=0
\end{gather*}
such that
\begin{equation}
 p_A\lrsb{\lrb{
 A: \exists\uu'\in\G\setminus\{\uu\}, A\uu'= A\uu
 }}
 \leq
 \frac{|\G|\alpha(n)}{|\U|^{l}}+\beta(n)
 \label{eq:ldpc}
\end{equation}
for any $\uu\in\U^n$, where
$\alpha(n)$ measures how the ensemble of $l\times n$ sparse matrices
differs from the ensemble of all $l\times n$ matrices
and $\beta(n)$ measures the probability
that the code determined by an
$l\times n$ sparse matrix has low-weight codewords.
It should be noted that the collision-resistant property can be
derived from (\ref{eq:bin}) and (\ref{eq:ldpc}).
It is shown in Section~\ref{sec:basic}.

The aim of this paper is not only unifying the above results,
but also providing several coding theorems under the general settings
such as an asymmetric channel for channel coding and an unbiased source
for lossy source coding.
To this end, we define an $(\aalpha,\bbeta)$-hash property in the following.

\begin{df}
Let $\A$ be a set of functions $A:\U^n\to\bU_{\A}$
and we assume that
 \begin{equation}
	\limn \frac{\log\frac{|\bU_{\A}|}{|\im\A|}}n=0.
	 \tag{H1}
	 \label{eq:imA}
 \end{equation}
For 
a probability distribution
\footnote{
It should be noted that $p_A$ does not depend on a particular function $A$.
Strictly speaking, the subscript $A$ of $p$ represents
the random variable of a function.
We use this ambiguous notation when $A$ aperars in the subscript of
$p$ because random variables are always denoted in Roman letter.
}
$p_A$ on $\A$, we
call a pair $(\A,p_A)$ an {\em ensemble}
 \footnote{In the standard definition, an ensemble is defined as a set
 of functions and a uniform distribution is assumed for this set.
 It should be noted that an ensemble is defined as
 the probability distribution on a set of functions in this paper.
 }.
Then, an ensemble $(\A,p_A)$ has an
{\em $(\aalpha_A,\bbeta_A)$-hash property} if
there are two sequences
\footnote{
It should be noted that 
$\aalpha_A$ and $\bbeta_A$ do not depend on a particular function $A$
but may depend on an ensemble $(\A,p_A)$.
Strictly speaking, the subscript $A$ represents the random variable.
}
$\aalpha_A\equiv\{\alpha_A(n)\}_{n=1}^{\infty}$ and
$\bbeta_A\equiv\{\beta_A(n)\}_{n=1}^{\infty}$ such that
 \begin{align}
	&\limn\alpha_A(n)=1
	\tag{H2}
	\label{eq:alpha}
	\\
	&\limn\beta_A(n)=0,
	\tag{H3}
	\label{eq:beta}
 \end{align}
and
\begin{align}
 \sum_{\substack{
 \uu\in\T
 \\
 \uu'\in\T'
 }}
 p_A
 \lrsb{\lrb{A: A\uu = A\uu'}}
 &\leq
 |\T\cap\T'|
 +\frac{|\T||\T'|\alpha_A(n)}{|\im\A|}
 +\min\{|\T|,|\T'|\}\beta_A(n)
 \tag{H4}
 \label{eq:hash}
\end{align}
for any $\T,\T'\subset\U^n$.
Throughout this paper,
we omit dependence on $n$ of $\alpha_A$ and $\beta_A$
when $n$ is fixed.
\end{df}

It should be noted that
we have  (\ref{eq:ldpc}) from (\ref{eq:hash})
by letting $\T\equiv\{\uu\}$ and $\T'\equiv\G$, and
(\ref{eq:bin}) is the case when $\alpha_A(n)\equiv 1$ and
$\beta_A(n)\equiv 0$.
In the right hand side of the inequality (H4),
the first term corresponds to the sum of $p_A(\{A:A\uu=A\uu\})=1$
over all $\uu\in\T\cap\T'$,
the second term bounds the sum of
the probability $p_A(\{A:A\uu=A\uu'\})$
which are
approximately $1/|\im\A|$ for $\uu\neq\uu'$,
and the third term bounds the sum
of the probability $p_A(\{A:A\uu=A\uu'\})$
far greater than $1/|\im\A|$ for $\uu\neq\uu'$.
This intuition is explained in Section~\ref{sec:linear}
for the ensemble of matrices.

In the following, we present two examples of ensembles that have a hash
property.

\noindent{\bf Example 1:}
Our terminology `hash' is derived from
a universal class of 
hash functions introduced in \cite{CW}.
We call a set $\A$ of functions $A:\U^n\to\bU_{\A}$
a {\em universal class of hash functions} if 
\[
 |\lrb{A: A\uu=A\uu'}|\leq \frac{|\A|}{|\bU_{\A}|}
\]
for any $\uu\neq\uu'$.
For example, the set of all functions on $\U^n$
and the set of all linear functions $A:\U^n\to\U^{l_{\A}}$
are classes of universal hash functions (see \cite{CW}).
Furthermore, for $\U^n\equiv\GF(2^n)$, the set
\[
 \A\equiv\lrb{
 A:
 \begin{aligned}
	&A\uu\equiv\lrB{\text{the first $l_{\A}$ bits of}\ \ba\uu}
	\\
	&\ba\in\GF(2^n)
 \end{aligned}
 }
\]
is a universal class of hash functions, where
$\ba\uu$ is a multiplication of two elements $\ba,\uu\in\GF(2^n)$.
It should be noted that every example above
satisfies $\im\A=\bU_{\A}$.
When $\A$ is a class of universal hash functions
and  $p_A$ is the uniform distribution on $\A$,
we have
\begin{align*}
 \sum_{\substack{
 \uu\in\T
 \\
 \uu'\in\T'
 }}
 p_A
 \lrsb{\lrb{A: A\uu=A\uu'}}
 \leq
 |\T\cap\T'|+\frac{|\T||\T'|}{|\im\A|}.
\end{align*}
This implies that $(\A,p)$ has a $(\one,\zero)$-hash property,
where $\one(n)\equiv 1$ and $\zero(n)\equiv 0$ for every $n$.

\noindent{\bf Example 2:}
In this example, we consider a set of linear functions
$A:\U^n\to\U^{l_{\A}}$.
It was discussed in the above example that
the uniform distribution on
the set of all linear functions has a $(\one,\zero)$-hash property.
The hash property of an ensemble of $q$-ary sparse matrices
will be discussed in Section \ref{sec:linear}.
The binary version of this ensemble is
introduced in \cite{Mac99}.

\subsection{Basic lemmas of hash property}
\label{sec:basic}

In the following, basic lemmas of the $(\aalpha,\bbeta)$-hash property
are introduced.
All lemmas are proved in Section~\ref{sec:proof-hash}.
Let $\A$ (resp. $\B$) be a set of functions $A:\U^n\to\bU_{\A}$
(resp. $B:\U^n\to\bU_{\B}$).
We assume that $(\A,p_A)$ (resp. $(\B,p_B)$) has
an $(\aalpha_A,\bbeta_A)$-hash  (resp. $(\aalpha_B,\bbeta_B)$-hash)
property.
We also assume that
$p_C$ is the uniform distribution on $\im\A$,
and random variables $A$, $B$, and $C$ are mutually
independent, that is,
\begin{align*}
 p_C(\cc)&=
 \begin{cases}
	\frac 1{|\im\A|},&\text{if}\ \cc\in\im\A
	\\
	0,&\text{if}\ \cc\in\bU_{\A}\setminus\im\A
 \end{cases}
 \\
 p_{ABC}(A,B,\cc)&=p_A(A)p_B(B)p_C(\cc)
\end{align*}
for any $A$, $B$, and $\cc$.

First, we demonstrate
that the collision-resistant property
and saturating property are
derived from the $(\aalpha,\bbeta)$-hash property.

The first lemma introduce the collision-resistant property.
\begin{lem}
 \label{lem:Anotempty}
 For any $\G\subset\U^n$ and $\uu\in\U^n$,
 \[
 p_A\lrsb{\lrb{
 A: \lrB{\G\setminus\{\uu\}}\cap\C_A(A\uu)\neq \emptyset
 }}
 \leq 
 \frac{|\G|\alpha_A}{|\im\A|} + \beta_A.
\]
\end{lem}

We prove the collision-resistant property
from Lemma~\ref{lem:Anotempty}.
Let $\mu_U$ be the probability distribution on $\G\subset\U^n$.
We have
\begin{align*}
 E_{A}\lrB{
 \mu_U\lrsb{\lrb{\uu: \lrB{\G\setminus\{\uu\}}\cap\C_A(A\uu)\neq\emptyset}}
 }
 &
 \leq
 \sum_{\uu\in\G}\mu_U(\uu)
 p_A\lrsb{\lrb{A: \lrB{\G\setminus\{\uu\}}\cap\C_A(A\uu)\neq\emptyset}}
 \\
 &
 \leq
 \sum_{\uu\in\G}\mu_U(\uu)
 \lrB{\frac{|\G|\alpha_A}{|\im\A|} + \beta_A}
 \\
 &
 \leq
 \frac{|\G|\alpha_A}{|\im\A|} + \beta_A.
\end{align*}
By assuming that $|\G|/|\im\A|$ vanishes as $n\to\infty$,
we have the fact that
there is a function $A$ such that
\[
 \mu_U\lrsb{\lrb{\uu:
 \lrB{\G\setminus\{\uu\}}\cap\C_A(A\uu)\neq\emptyset}}
 <\delta
\]
for any $\delta>0$ and sufficiently large $n$.
Since the relation $\lrB{\G\setminus\{\uu\}}\cap\C_A(A\uu)\neq\emptyset$
corresponds to the event that there is $\uu'\in\G$, $\uu'\neq\uu$
such that $\uu$ and $\uu'$ are the members of the same bin
(have the same codeword determined by $A$),
we have the fact that
the members of $\G$ are located in the different bins
(the members of $\G$ can be decoded correctly)
with high probability.
In the proof of fixed-rate source coding, $\G$ is defined by
(\ref{eq:G}) for given probability distribution $\mu_U$ of a source $U$,
where $\mu_{U}(\G^c)$ is close to zero.
In the linear coding of a channel with additive noise,
additive noise $\uu$ can be specified by the syndrome $A\uu$
obtained by operating the parity check matrix $A$ to
the channel output.
It should be noted that, when
Lemma~\ref{lem:Anotempty} is
applied, it is sufficient to assume that
\begin{equation*}
 \limn \frac{\log\alpha_A(n)}n=0
\end{equation*}
instead of (\ref{eq:alpha})
because it is usually assumed that $|\G|/|\im\A|$
vanishes exponentially by letting $n\to\infty$.
It is implicitly proved in \cite{MB01}\cite{BB}\cite{EM05}\cite{SWLDPC}
that some ensembles of sparse linear matrices have this weak hash
property.
In fact, the condition (\ref{eq:alpha}) is required for the saturating
property.

The second lemma introduce the saturating property.
We use the folloing lemma in the proof of the coding theorems of 
the Gel'fand-Pinsker problem, the Wyner-Ziv problem, and the
One-helps-one problem.

\begin{lem}
\label{lem:noempty}
If $\T\neq\emptyset$, then
 \begin{align*}
	p_{AC}\lrsb{\lrb{(A,\cc):
	\T\cap\C_A(\cc)=\emptyset
	}}
	&\leq
 \alpha_A-1+\frac{|\im\A|\lrB{\beta_A+1}}{|\T|}.
 \end{align*}
\end{lem}

We prove the saturating property
form Lemma~\ref{lem:noempty}.
We have
\begin{align*}
 E_{A}\lrB{
 p_{C}\lrsb{\lrb{\cc: \T\cap\C_A(\cc)=\emptyset}}
 }
 &=
 p_{AC}\lrsb{\lrb{(A,\cc): \T\cap\C_A(\cc)=\emptyset}}
 \\
 &\leq
 \alpha_A-1+\frac{|\im\A|\lrB{\beta_A+1}}{|\T|}.
\end{align*}
By assuming that $|\im\A|/|\T|$ vanishes as $n\to\infty$,
we have the fact that
there is a function $A$ such that
\[
 p_{C}\lrsb{\lrb{\cc: \T\cap\C_A(\cc)=\emptyset}}
 <\delta
\]
for any $\delta>0$ and sufficiently large $n$.
Since the relation $\T\cap\C_A(\cc)=\emptyset$
corresponds to the event that there is no $\uu\in\T$
in the bin $\C_A(\cc)$,
we have the fact that we can find a member of $\T$
in the randomly selected bin with high probability.
It should be noted that, when Lemma~\ref{lem:noempty} is applied,
it is sufficient to assume that
\[
 \limn \frac{\log\beta_A(n)}n=0
\]
instead of (\ref{eq:beta})
because it is usually assumed that $|\im\A|/|\T|$
vanishes exponentially by letting $n\to\infty$.
In fact, the condition (\ref{eq:beta}) is required for the
collision-resistant property.

Next, we prepare the lemmas
used in the proof of coding theorems.
The following lemmas come from Lemma~\ref{lem:Anotempty}.

\begin{lem}
 \label{lem:ACnotempty}
 If $\G\subset\U^n$ and $\uu\notin\G$, then
\[
 p_{AC}\lrsb{\lrb{
 (A,\cc):
 \begin{aligned}
	&\G\cap\C_A(\cc)\neq \emptyset
	\\
	&\uu\in\C_A(\cc)
 \end{aligned}
 }}
 \leq 
 \frac{|\G|\alpha_A}{|\im\A|^2} + \frac{\beta_A}{|\im\A|}.
\]
\end{lem}

\begin{lem}
 \label{lem:ABCnoempty}
 Assume that $\uu_{A,\cc}\in\U^n$ depends on $A$ and $\cc$.
 Then 
\begin{align*}
 p_{ABC}\lrsb{\lrb{
 (A,B,\cc):
 \lrB{\G\setminus\{\uu_{A,\cc}\}}\cap\C_{AB}(\cc,B\uu_{A,\cc})\neq\emptyset
 }}
 &\leq 
	\frac{|\G|\alpha_B}
	{|\im\A||\im\B|}
	+\beta_B
\end{align*}
 for any $\G\subset\U^n$.
\end{lem}

Finally, we introduce the
method for finding a typical sequence in a bin.
The probability of an event where the function
finds a conditionally typical sequence
is evaluated by the following
lemmas. They are the key lemmas for the coding theorems
for the Gel'fand-Pinsker problem, the Wyner-Ziv problem, and 
the One-helps-one problem.
These lemmas are proved by using Lemma~\ref{lem:noempty}.
For $\e>0$, let
\begin{align*}
 l_{\A}
 &\equiv \frac{n[H(U|V)-\e]}{\log|\U|}
\end{align*}
and assume that $\A$ is a set of functions $A:\U^n\to\U^{l_{\A}}$
and $\vv\in\T_{V,\gamma}$.

\begin{lem}
\label{thm:joint-typical}
We define a {\em maximum-likelihood (ML) coding function}
$g_A$ under constraint $\uu\in\C_A(\cc)$ as
\begin{align*}
 g_A(\cc|\vv)
 &\equiv
 \arg\max_{\uu\in\C_A(\cc)}
 \mu_{U|V}(\uu|\vv)
 \\
 &=
 \arg\max_{\uu\in\C_A(\cc)}
 \mu_{UV}(\uu,\vv)
\end{align*}
and assume that a set $\T(\vv)\subset\T_{U|V,2\e}(\vv)$
satisfies
\begin{itemize}
 \item $\T(\vv)$ is not empty, and
 \item if $\uu\in\T(\vv)$ and $\uu'$ satisfies
 \[
 \mu_{U|V}(\uu|\vv)\leq \mu_{U|V}(\uu'|\vv)\leq 2^{-n[H(U|V)-2\e]}
 \]
 then $\uu'\in\T(\vv)$.
\end{itemize}
In fact, we can construct such $\T(\vv)$
by taking up $|\T(\vv)|$ elements from $\T_{U|V,2\e}(\vv)$
in the order of probability rank.
If an ensemble $(\A,p_A)$ of a set of functions $A:\U^n\to\U^{l_{\A}}$
has an $(\aalpha_A,\bbeta_A)$-hash property,
then 
\begin{align*}
 p_{AC}\lrsb{\lrb{(A,\cc):
 g_{A}(\cc|\vv)\notin\T(\vv)
 }}
 &\leq 
 \alpha_A-1
 +\frac{|\im\A|\lrB{\beta_A+1}}{|\T(\vv)|}
 +\frac{2^{-n\e}|\U|^{l_{\A}}}{|\im\A|}
\end{align*}
for any $\vv$ satisfying $\T_{U|V,2\e}(\vv)\neq\emptyset$.
\end{lem}

\begin{lem}
 \label{lem:md-typical}
 We define a {\em minimum-divergence (MD) coding function}
 $\hg_A$ under constraint $\uu\in\C_A(\cc)$ as
 \begin{align*}
	\hg_A(\cc|\vv)
	&\equiv
	\arg\min_{\uu\in\C_A(\cc)}
	D(\nu_{\uu|\vv}\|\mu_{U|V}|\nu_{\vv})
	\\
	&=
	\arg\min_{\uu\in\C_A(\cc)}
	D(\nu_{\uu\vv}\|\mu_{UV})
 \end{align*}
 and assume that, for $\gamma>0$,
 a set $\T\subset\T_{U|V,\gamma}(\vv)$
 satisfies
 that if $\uu\in\T$ and $\uu'$ satisfies
 \[
 D(\nu_{\uu'|\vv}\|\mu_{U|V}|\nu_{\vv})\leq
 D(\nu_{\uu|\vv}\|\mu_{U|V}|\nu_{\vv})
 \]
 then $\uu'\in\T$. In fact, we can construct such $\T$
 by picking up $|\T|$ elements from $\T_{U|V,\gamma}(\vv)$
 in the descending order of conditional divergence.
 Then 
 \begin{align*}
	p_{AC}\lrsb{\lrb{(A,\cc):
	\hg_{A}(\cc|\vv)\notin\T(\vv)
	}}
	&\leq 
	\alpha_A-1
	+\frac{|\im\A|\lrB{\beta_A+1}}{|\T(\vv)|}
 \end{align*}
 for any $\vv$ satisfying $\T_{U|V,\gamma}(\vv)\neq\emptyset$.
\end{lem}

In Section~\ref{sec:coding-theorems},
we construct codes
by using the maximum-likelihood coding function.
It should be noted that we can replace
the maximum-likelihood coding function
by the minimum-divergence coding function
and prove theorems simpler than that presented in this paper.

\section{Hash Property for Ensemble of Matrices}
\label{sec:linear}
In this section, we discuss the hash property of an ensemble
of (sparse) matrices.

First, we introduce the average spectrum of an ensemble of
matrices given in~\cite{BB}.
Let $\U$ be a finite field
and  $p_A$ be a probability distribution on a set of
$l_{\A}\times n$ matrices.
It should be noted that $A$ represents a corresponding
linear function $A:\U^n\to\U^{l_{\A}}$,
where we define $\bU_{\A}\equiv\U^{l_{\A}}$.

Let $\bt(\uu)$ be the type\footnote{
As in~\cite{CK}, the type is defined in terms
of the empirical probability distribution
$\nu_{\uu}$.
In our definition, the type is the number $n\nu_{\uu}$ of
occurences which is different from the empilical probablity distribution.
}
of $\uu\in\U^n$,
where a type is characterized by the number $n\nu_{\uu}$ of 
occurrences of each symbol in the sequence $\uu$.
Let $\cH$ be a set of all types of length $n$ except $\bt(\zero)$,
where $\zero$ is the zero vector.
For the probability distribution
$p_A$ on a set of $l_{\A}\times n$ matrices, let $S(p_A,\bt)$ be defined as
\begin{gather*}
 S(p_A,\bt)
 \equiv
 \sum_{A}p_A(A)|\{\uu\in\U^n: A\uu=\zero, \bt(\uu)=\bt\}|.
\end{gather*}
For $\hcH\subset\cH$, we define $\alpha_A(n)$ and $\beta_A(n)$ as
\begin{align}
 \alpha_A(n)
 &\equiv
 \frac{|\im\A|}{|\U|^{l_{\A}}}\cdot\max_{\bt\in \hcH}
 \frac {S(p_A,\bt)}{S(u_A,\bt)}
 \label{eq:alpha-linear}
 \\
 \beta_A(n)
 &\equiv
 \sum_{\bt\in \cH\setminus\hcH}S(p_A,\bt),
 \label{eq:beta-linear}
\end{align}
where $u_{A}$ denotes the uniform distribution
on the set of all $l_{\A}\times n$ matrices.
When $\U\equiv\GF(2)$ and $\hcH$ is a set of high-weight types,
$\alpha_A$ measures how the ensemble $(\A,p_A)$
differs from the ensemble of all $l_{\A}\times n$ matrices
with respect to high-weight part of average spectrum
and $\beta_A$ provides the upper bound of the probability
that the code $\{\uu\in\U^n: A\uu=\zero\}$ has low-weight codewords.
It should be noted that
\begin{align}
 \widetilde{\alpha}_A(n)
 &\equiv
 \max_{\bt\in \hcH}
 \frac {S(p_A,\bt)}{S(u_A,\bt)}
 \label{eq:alpha-linear-org}
\end{align}
is introduced
in~\cite{MB01}\cite{BB}\cite{EM05}\cite{SWLDPC}\cite{CRYPTLDPC}
instead of $\alpha_A(n)$.
We multiply the coefficient $|\im\A|/|\U|^{l_{\A}}$
so that $\aalpha_A$ satisfies (\ref{eq:alpha}).

We have the following theorem.
\begin{thm}
\label{thm:hash-linA}
Let $(\A,p_A)$ be an ensemble of matrices
and assume that
$p_A\lrsb{\lrb{A: A\uu=\zero}}$
depends on $\uu$ only through the type $\bt(\uu)$.
If $|\bU_{\A}|/|\im\A|$ satisfies (\ref{eq:imA})
and $(\alpha_A(n),\beta_A(n))$, defined by
(\ref{eq:alpha-linear}) and (\ref{eq:beta-linear}),
satisfies (\ref{eq:alpha}) and (\ref{eq:beta}),
then $(\A,p_A)$ has an $(\aalpha_A,\bbeta_A)$-hash property.
\end{thm}
The proof is given in Section \ref{sec:proof-linear}.

Next, we consider the independent combination of two ensembles $(\A,p_A)$ and
$(\B,p_B)$, of $l_{\A}\times n$ and $l_{\B}\times n$ matrices, respectively.
We assume that $(\A,p_A)$ has
an $(\aalpha_A,\bbeta_A)$-hash property,
where 
$(\alpha_A(n),\beta_A(n))$ is defined as (\ref{eq:alpha-linear})
and  (\ref{eq:beta-linear}).
Similarly we define $(\alpha_B(n),\beta_B(n))$ for an ensemble $(\B,p_B)$
and assume that $(\B,p_B)$ has an $(\aalpha_B,\bbeta_B)$-hash
property.
Let $p_{AB}$ be the joint distribution
defined as
\[
 p_{AB}(A,B)\equiv p_A(A)p_B(B).
\]
We have the following two lemmas.
The proof is given in Section \ref{sec:proof-linear}.
\begin{lem}
\label{thm:hash-linApB}
Let $(\alpha_{AB}(n),\beta_{AB}(n))$ defined as
\begin{align*}
 \alpha_{AB}(n)&\equiv \alpha_A(n)\alpha_B(n)
 \\
 \beta_{AB}(n)&\equiv \min\{\beta_A(n),\beta_B(n)\}.
\end{align*}
Then
the
ensemble $(\A\times\B,p_{AB})$ 
of functions $A\oplus B:\U^n\to\U^{l_{\A}+l_{\B}}$ defined as
\[
	A\oplus B(\uu)\equiv(A\uu,B\uu)
\]
has an $(\aalpha_{AB},\bbeta_{AB})$-hash property.
\end{lem}

\begin{lem}
\label{thm:hash-linAB}
Let $(\alpha_{AB}(n),\beta'_{AB}(n))$ be defined as
\begin{align*}
 \alpha_{AB}(n)&\equiv \alpha_A(n)\alpha_B(n)
 \\
 \beta'_{AB}(n)&\equiv 
 \frac{\alpha_A(n)\beta_B(n)}{|\im\A|}+
 \frac{\alpha_B(n)\beta_A(n)}{|\im\B|}
 +\beta_A(n)\beta_B(n).
\end{align*}
Then the ensemble $(\A\times\B,p_{AB})$ of functions
$A\otimes B:\U^n\times\V^n\to\U^{l_{\A}}\times\V^{l_{\B}}$
defined as
\[
 A\otimes B(\uu,\vv)\equiv (A\uu,B\vv)
\]
has an $(\aalpha_{AB},\bbeta'_{AB})$-hash property.
\end{lem}

Finally, we introduce
an ensemble of
$q$-ary sparse matrices, where the binary version of this ensemble is
proposed in \cite{Mac99}.
In the following, let $\U\equiv\GFq$ and $l_{\A}\equiv nR$.
We generate an $l_{\A}\times n$ matrix $A$ with the following
procedure,
where at most $\tau$ random nonzero elements are introduced in every row.
\begin{enumerate}
\item Start from an all-zero matrix.
\item For each $i\in\{1,\ldots,n\}$, repeat the following
			procedure $\tau$ times:
			\begin{enumerate}
			 \item Choose
						 $(j,a)\in\{1,\ldots,l_{\A}\}\times[\GFq\setminus\{0\}]$
						 uniformly at random.
			 \item Add  $a$ to the  $(j,i)$ component of $A$.
			\end{enumerate}
\end{enumerate}
Let $(\A,p_A)$ be an ensemble corresponding to the above procedure.
It is proved in 
Section \ref{sec:alphabeta}
that $p_A\lrsb{\lrb{A: A\uu=\zero}}$
depends on $\uu$ only through the type
$\bt(\uu)$.
Let $(\alpha_A(n),\beta_A(n))$ be defined by
(\ref{eq:alpha-linear}) and (\ref{eq:beta-linear})
for this ensemble.
We assume that column weight $\tau=O(\log n)$ is even.
Let $w(\bt)$ be the weight of type $\bt=(t(0),\ldots,t(q-1))$ defined as
\[
 w(\bt)\equiv\sum_{i=1}^{q-1}t(i)
\]
and $w(\uu)$ be defined as
\[
 w(\uu)\equiv w(\bt(\uu)).
\]
We define
\begin{equation}
 \hcH\equiv\{\bt: w(\bt)>\xi l_{\A}\}.
 \label{eq:hcH}
\end{equation}
Then it is also proved in
Section~\ref{sec:alphabeta}
that
\begin{align*}
 \im\A
 &=
 \begin{cases}
	\{\uu\in\U^{l_{\A}}: w(\uu)\ \text{is even}\},&\text{if}\ q=2
	\\
	\U^l,&\text{if}\ q>2
 \end{cases}
 \\
 \frac{|\im\A|}{|\U|^{l_{\A}}}
 &=
 \begin{cases}
	2,&\text{if}\ q=2
	\\
	1,&\text{if}\ q>2
 \end{cases}
\end{align*}
and
there is $\xi>0$ such that
$(\aalpha_A,\bbeta_A)$ satisfies (\ref{eq:alpha}) and (\ref{eq:beta}).
From
Theorem~\ref{thm:hash-linA}, we have the following theorem.
\begin{thm}
\label{thm:qMacKay}
For the above ensemble $(\A,p_A)$ of sparse matrices,
let $(\alpha_A(n),\beta_A(n))$ be defined by
(\ref{eq:alpha-linear}),
(\ref{eq:beta-linear}),
(\ref{eq:hcH}), and suitable $\{\tau(n)\}_{n=1}^{\infty}$ and $\xi>0$.
Then $(\A,p_A)$ has an $(\aalpha_A,\bbeta_A)$-hash property.
\end{thm}
We prove the theorem in Section~\ref{sec:alphabeta}.
It should be noted here that,
as we can see in the proof of the theorem,
the asymptotic behavior
(convergence speed) of $(\aalpha_A,\bbeta_A)$ 
depends on the weight $\tau$.

\begin{rem}
It is proved in \cite{MB01}\cite{SWLDPC} that
$(\widetilde{\alpha}_A(n),\beta_A(n))$,
defined by
(\ref{eq:alpha-linear-org}) and (\ref{eq:beta-linear}),
satisfies weaker properties
 \begin{equation}
	\limn \frac {\log\widetilde{\alpha}_A(n)}n =0
	 \label{eq:talpha}
 \end{equation} 
and (\ref{eq:beta}) when $q=2$.
It is proved in \cite[Section III, Eq.~(23),(82)]{BB} that
$(\widetilde{\alpha}_A(n),\beta_A(n))$
of another ensemble of Modulo-$q$ LDPC matrices satisfies
weaker properties (\ref{eq:talpha})
and (\ref{eq:beta}).
\end{rem}

\section{Coding Theorems}
\label{sec:coding-theorems}
In this section,
we present several coding theorems.
We prove these theorems in Section \ref{sec:proof} based on the hash property.

Throughout this section,
the encoder and decoder are denoted by $\Encoder$ and $\Decoder$,
respectively.
We assume that the dimension of vectors $\xx$, $\yy$, $\zz$, and $\ww$
is $n$.

\subsection{Slepian-Wolf Problem}
In this section, we consider the Slepian-Wolf problem
illustrated in Fig.\ \ref{fig:sw}.
The achievable rate region for this problem is given by
the set of encoding rate pair $(\Rate_X,\Rate_Y)$ satisfying
\begin{gather*}
 \Rate_X \geq H(X|Y)
 \\
 \Rate_Y \geq H(Y|X)
 \\
 \Rate_X+\Rate_Y \geq H(X,Y).
\end{gather*}

The achievability of the Slepian-Wolf problem
is proved in \cite{C75} and \cite{CSI82}
for the ensemble of bin-coding and all $q$-ary linear
matrices, respectively.
The constructions of encoders using sparse matrices is studied
in \cite{M02}\cite{SPR02}\cite{CRYPTLDPC}
and the achievability is proved
in \cite{SWLDPC} by using ML decoding.
The aim of this section is to demonstrate the proof of
the coding theorem based on the hash property.
The proof is given in Section \ref{sec:proof-sw}.

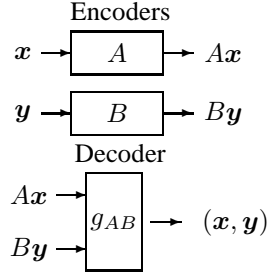
\begin{figure}[t]
\begin{center}
\unitlength 0.4mm
\begin{picture}(176,47)(0,0)
\put(82,41){\makebox(0,0){Encoders}}
\put(50,27){\makebox(0,0){$\xx$}}
\put(56,27){\vector(1,0){10}}
\put(66,20){\framebox(30,14){$A$}}
\put(96,27){\vector(1,0){10}}
\put(116,27){\makebox(0,0){$A\xx$}}
\put(50,7){\makebox(0,0){$\yy$}}
\put(56,7){\vector(1,0){10}}
\put(66,0){\framebox(30,14){$B$}}
\put(96,7){\vector(1,0){10}}
\put(116,7){\makebox(0,0){$B\yy$}}
\end{picture}
\\
\begin{picture}(176,45)(0,0)
\put(82,39){\makebox(0,0){Decoder}}
\put(52,25){\makebox(0,0){$A\xx$}}
\put(61,25){\vector(1,0){10}}
\put(52,7){\makebox(0,0){$B\yy$}}
\put(61,7){\vector(1,0){10}}
\put(71,0){\framebox(18,32){$g_{AB}$}}
\put(92,16){\vector(1,0){10}}
\put(120,16){\makebox(0,0){$(\xx,\yy)$}}
\end{picture}
\end{center}
\caption{Construction of Slepian-Wolf Source Code}
\label{fig:sw-code}
\end{figure}

We fix functions
\begin{align*}
 A&:\X^n\to\X^{l_{\A}}
 \\
 B&:\Y^n\to\Y^{l_{\B}}
\end{align*}
which are available to construct encoders and a decoder.
We define the encoders and the decoder (illustrated in Fig.\ \ref{fig:sw-code})
\begin{align*}
 \Encoder_X&:\X^n\to\X^{l_{\A}}
 \\
 \Encoder_Y&:\Y^n\to\Y^{l_{\B}}
 \\
 \Decoder&:\X^{l_{\A}}\times\Y^{l_{\B}}\to\X^n\times\Y^n
\end{align*}
as
\begin{align*}
 \Encoder_X(\xx)&\equiv A\xx
 \\
 \Encoder_Y(\yy)&\equiv B\yy
 \\
 \Decoder(\bb_X,\bb_Y)
 &\equiv g_{AB}(\bb_X,\bb_Y),
\end{align*}
where
\begin{align*}
 g_{AB}(\bb_X,\bb_Y)
 &\equiv
 \arg\max_{(\xx',\yy')\in\C_{A}(\bb_X)\times\C_B(\bb_Y)}\mu_{XY}(\xx',\yy').
\end{align*}

The encoding rate pair $(\Rate_X,\Rate_Y)$ is given by
\begin{align*}
 \Rate_X\equiv\frac{l_{\A}\log|\X|}n
 \\
 \Rate_Y\equiv\frac{l_{\B}\log|\Y|}n
\end{align*}
and the error probability $\Error_{XY}(A,B)$ is given by
\begin{align*}
 \Error_{XY}(A,B)
 &
 \equiv
 \mu_{XY}\lrsb{\lrb{
 (\xx,\yy): \Decoder(\Encoder_X(\xx),\Encoder_Y(\yy))\neq (\xx,\yy)
 }}.
\end{align*}

We have the following theorem.
It should be noted that
and alphabets $\X$ and $\Y$ may not be binary
and the correlation of the two sources may not be symmetric.
\begin{thm}
 \label{thm:sw}
 Assume that $(\A,p_A)$, $(\B,p_B)$, and $(\A\times\B,p_A\times p_B)$
 have hash property.
 Let  $(X,Y)$ be a pair of stationary memoryless sources.
 If $(\Rate_X,\Rate_Y)$ satisfies
 \begin{gather}
	\Rate_X > H(X|Y)
	\label{eq:swx}
	\\
	\Rate_Y > H(Y|X)
	\label{eq:swy}
	\\
	\Rate_X+\Rate_Y > H(X,Y),
	\label{eq:swxy}
 \end{gather}
 then for any $\delta>0$
 and all sufficiently large $n$
 there are functions (sparse matrices) $A\in\A$ and $B\in\B$ such that
 \begin{align*}
	&\Error_{XY}(A,B)
	\leq
	\delta.
 \end{align*}
\end{thm}

\begin{rem}
In \cite{C75}\cite{CSI82},
random (linear) bin-coding is used to prove the achievability
of the above theorem.
In fact, random bin-coding
is equivalent to a uniform ensemble on a set of all (linear) functions
and it has a $(\one,\zero)$-hash property.
\end{rem}

\begin{rem}
The above theorem
includes the fixed-rate coding of a single source $X$
as a special case of the Slepian-Wolf problem with $|\Y|\equiv 1$.
This implies that the encoding rate can achieve
the entropy of a source.
It should be noted that source coding using a class
of hash functions is studied in \cite[Section 14.2]{MacKay}\cite{K04}.
\end{rem}

\begin{rem}
Assuming $|\Y|\equiv 1$,
we can prove the coding theorem for 
a channel with additive noise $X$
by letting $A$ and $\{\xx: A\xx=\zero\}$ 
be a parity check matrix and a set of codewords
(channel inputs), respectively.
This implies that
the encoding rate of the channel can achieve the channel capacity.
The coding theorem for a channel with additive noise
is proved
by using a low density parity check (LDPC) matrix
in \cite{MB01}\cite{BB}\cite{EM05}.
\end{rem}

\subsection{Gel'fand-Pinsker Problem}
\label{sec:gp}

In this section we consider 
the Gel'fand-Pinsker problem illustrated in Fig.\ \ref{fig:gp}.

First, we construct a code for the
standard channel coding problem
illustrated in Fig.\ \ref{fig:channel},
which is a special case of Gel'fand-Pinsker problem.

A channel is given by
the conditional probability distribution $\mu_{Y|X}$,
where $X$ and $Y$ are random variables corresponding to
the channel input and channel output, respectively.
The capacity of a channel is given by
\[
 \Capacity\equiv\max_{\mu_{X}}I(X;Y),
\]
where the maximum is taken over all probability
distributions $\mu_{X}$
and the joint distribution of random variable $(X,Y)$
is given by
\[
 \mu_{XY}(x,y)
 \equiv \mu_{Y|X}(y|x)\mu_{X}(x).
\]

\begin{figure}[t]
\begin{center}
\unitlength 0.47mm
\begin{picture}(150,20)(0,0)
\put(4,7){\makebox(0,0){$M$}}
\put(10,7){\vector(1,0){10}}
\put(20,0){\framebox(14,14){$\Encoder$}}
\put(34,7){\vector(1,0){10}}
\put(50,7){\makebox(0,0){$X$}}
\put(56,7){\vector(1,0){10}}
\put(66,0){\framebox(24,14){$\mu_{Y|X}$}}
\put(90,7){\vector(1,0){10}}
\put(106,7){\makebox(0,0){$Y$}}
\put(112,7){\vector(1,0){10}}
\put(122,0){\framebox(14,14){$\Decoder$}}
\put(136,7){\vector(1,0){10}}
\put(150,7){\makebox(0,0){$M$}}
\put(75,-10){\makebox(0,0){$R<I(X;Y)$}}
\end{picture}
\end{center}
\caption{Channel Coding}
\label{fig:channel}
\end{figure}
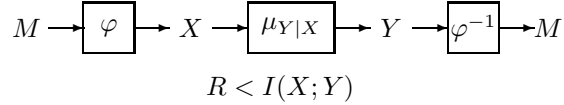

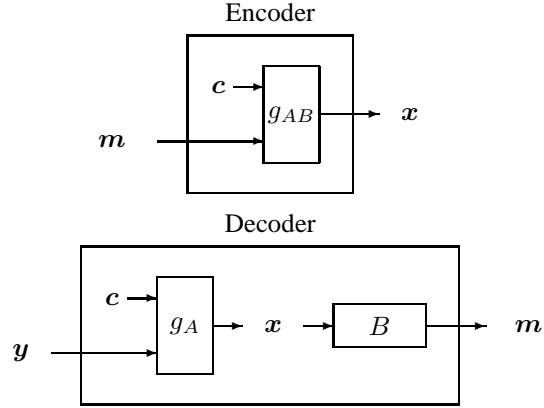
\begin{figure}[t]
\begin{center}
\unitlength 0.4mm
\begin{picture}(176,70)(0,0)
\put(82,60){\makebox(0,0){Encoder}}
\put(65,35){\makebox(0,0){$\cc$}}
\put(70,35){\vector(1,0){10}}
\put(30,17){\makebox(0,0){$\mm$}}
\put(45,17){\vector(1,0){35}}
\put(80,10){\framebox(18,32){$g_{AB}$}}
\put(98,26){\vector(1,0){20}}
\put(128,26){\makebox(0,0){$\xx$}}
\put(55,0){\framebox(54,52){}}
\end{picture}
\\
\begin{picture}(176,70)(0,0)
\put(82,60){\makebox(0,0){Decoder}}
\put(30,35){\makebox(0,0){$\cc$}}
\put(35,35){\vector(1,0){10}}
\put(0,17){\makebox(0,0){$\yy$}}
\put(10,17){\vector(1,0){35}}
\put(45,10){\framebox(18,32){$g_A$}}
\put(63,26){\vector(1,0){10}}
\put(83,26){\makebox(0,0){$\xx$}}
\put(93,26){\vector(1,0){10}}
\put(103,19){\framebox(30,14){$B$}}
\put(133,26){\vector(1,0){20}}
\put(167,26){\makebox(0,0){$\mm$}}
\put(20,0){\framebox(124,52){}}
\end{picture}
\end{center}
\caption{Construction of Channel Code}
\label{fig:channel-code}
\end{figure}

The code for this problem is given below
(illustrated in Fig.\ \ref{fig:channel}).
We fix functions
\begin{align*}
 A&:\X^n\to\X^{l_{\A}}
 \\
 B&:\X^n\to\X^{l_{\B}}
\end{align*}
and a vector $\cc\in\X^{l_{\A}}$
available to construct an encoder and a decoder, where
\begin{align*}
 l_{\A}
 &\equiv \frac{n[H(X|Y)+\eA]}{\log|\X|}
 \\
 l_{\B}
 &\equiv \frac{n[I(X;Y)-\eB]}{\log|\X|}.
\end{align*}
We define the encoder and the decoder
\begin{align*}
 \Encoder&:\X^{l_{\B}}\to\X^n
 \\
 \Decoder&:\Y^n\to\X^{l_{\B}}
\end{align*}
as
\begin{align*}
 \Encoder(\mm)
 &\equiv g_{AB}(\cc,\mm)
 \\
 \Decoder(\mm)
 &\equiv Bg_A(\cc|\yy),
\end{align*}
where
\begin{align*}
 g_{AB}(\cc,\mm)
 &\equiv\arg\max_{\xx'\in\C_{AB}(\cc,\mm)}\mu_{X}(\xx')
 \\
 g_{A}(\cc|\yy)
 &\equiv\arg\max_{\xx'\in\C_A(\cc)}\mu_{XY}(\xx'|\yy).
\end{align*}

Let $M$ be the random variable corresponding to
the message $\mm$, where the probability $p_M(\mm)$
is given by
\begin{align*}
 p_M(\mm)&\equiv
 \begin{cases}
	\frac 1{|\im\B|},
	&\text{if}\ \mm\in\im\B
	\\
	0,
	&\text{if}\ \mm\notin\im\B.
 \end{cases}
\end{align*}
The rate $\Rate(B)$ of this code is given by
\begin{align*}
 \Rate(B)&\equiv\frac{\log|\im\B|}n
 \\
 &=\frac{l_B\log|\W|}n-\frac{\log\frac{|\W|^{l_B}}{|\im\B|}}n
\end{align*}
and the decoding error probability $\Error_{Y|X}(A,B,\cc)$
is given by
\begin{align*}
 \Error_{Y|X}(A,B,\cc)
 &\equiv
 \sum_{\mm,\yy}
 p_M(\mm)\mu_{Y|X}(\yy|\Encoder(\mm))
 \chi(\Decoder(\yy)\neq\mm).
\end{align*}

In the following, we provide an
intuitive interpretation of the construction of the code,
which is illustrated in Fig.~\ref{fig:channel-code}.
Assume that $\cc$ is shared by the encoder and the decoder.
For $\cc$ and a message $\mm$,
the function $g_{AB}$ generates a typical sequence
$\xx\in\T_{X,\gamma}$ as a channel input.
The decoder 
reproduces the channel input $\xx$
by using $g_{A}$ from $\cc$ and a channel output $\yy$.
Since $(\xx,\yy)$ is jointly typical and $B\xx=\mm$,
the decoding succeed if
the amount of information of $\cc$ is greater than $H(X|Y)$
to satisfy the collision-resistant property.
On the other hand,
the total rate of $\cc$ and $\mm$
should be less than $H(X)$ to satisfy the saturating property.
Then we can set the encoding rate of $\mm$ close to
$H(X)-H(X|Y)=I(X;Y)$.

We have the following theorem
It should be noted that
alphabets $\X$ and $\Y$ are allowed to be non-binary,
and the channel 
is allowed to be asymmetric.
\begin{thm}
 \label{thm:channel}
 For given $\eA,\eB>0$ satisfying
 \begin{gather*}
	\eB-\eA\leq\sqrt{6[\eB-\eA]}\log|\X|<\eA,
 \end{gather*}
 assume
 that $(\A,p_A)$ and $(\A\times\B,p_A\times p_B)$
 have hash property.
 Let $\mu_{Y|X}$ be the conditional probability distribution
 of a stationary memoryless channel.
 Then, for all $\delta>0$ and sufficiently large $n$
 there are functions (sparse matrices) $A\in\A$, $B\in\B$,
 and a vector $\cc\in\im A$ such that
\begin{gather*}
 \Rate(B) \geq I(X;Y)-\eB-\delta
 \\
 \Error_{Y|X}(A,B,\cc)<\delta.
\end{gather*}
 By assuming that $\mu_{X}$ attains the channel capacity
 and $\delta\to 0$, $\eB\to 0$, 
 the rate of the proposed code is close to the capaticy.
\end{thm}

Next, we consider 
the Gel'fand-Pinsker problem illustrated in Fig.\ \ref{fig:gp}.
A channel with side information is given by
the conditional probability distribution $\mu_{Y|XZ}$,
where $X$, $Y$ and $Z$ are random variables corresponding to
the channel input, channel output, and channel side information, respectively.
The capacity of a channel with side information
is given by
\[
 \Capacity\equiv\max_{\mu_{XW|Z}}\lrB{I(W;Y)-I(W;Z)},
\]
where the maximum is taken over all conditional probability
distributions $\mu_{XW|Z}$
and the joint distribution of random variable $(X,Y,Z,W)$
is given by
\begin{align}
 \begin{split}
 \mu_{XYZW}(x,y,z,w)
 &
 \equiv \mu_{XW|Z}(x,w|z)\mu_{Y|XZ}(y|x,z)\mu_Z(z).
 \end{split}
 \label{eq:markov-gp}
\end{align}

\begin{figure}[t]
\begin{center}
\unitlength 0.4mm
\begin{picture}(156,99)(-6,6)
\put(76,89){\makebox(0,0){Encoder}}
\put(30,55){\makebox(0,0){$\cc$}}
\put(35,55){\vector(1,0){10}}
\put(0,37){\makebox(0,0){$\mm$}}
\put(10,37){\vector(1,0){35}}
\put(45,30){\framebox(18,32){$g_{AB}$}}
\put(63,46){\vector(1,0){10}}
\put(83,46){\makebox(0,0){$\ww$}}
\put(93,46){\vector(1,0){10}}
\put(83,64){\makebox(0,0){$\hcc$}}
\put(93,64){\vector(1,0){10}}
\put(103,39){\framebox(18,32){$g_{\hA}$}}
\put(121,55){\vector(1,0){20}}
\put(0,19){\makebox(0,0){$\zz$}}
\put(10,19){\line(1,0){102}}
\put(54,19){\vector(0,1){11}}
\put(112,19){\vector(0,1){20}}
\put(147,55){\makebox(0,0){$\xx$}}
\put(20,6){\framebox(112,75){}}
\end{picture}
\\
\begin{picture}(156,70)(0,0)
\put(82,60){\makebox(0,0){Decoder}}
\put(30,35){\makebox(0,0){$\cc$}}
\put(35,35){\vector(1,0){10}}
\put(0,17){\makebox(0,0){$\yy$}}
\put(10,17){\vector(1,0){35}}
\put(45,10){\framebox(18,32){$g_A$}}
\put(63,26){\vector(1,0){10}}
\put(83,26){\makebox(0,0){$\ww$}}
\put(93,26){\vector(1,0){10}}
\put(103,19){\framebox(30,14){$B$}}
\put(133,26){\vector(1,0){20}}
\put(167,26){\makebox(0,0){$\mm$}}
\put(20,0){\framebox(124,52){}}
\end{picture}
\end{center}
\caption{Construction of Gel'fand-Pinsker Channel Code}
\label{fig:gp-code}
\end{figure}
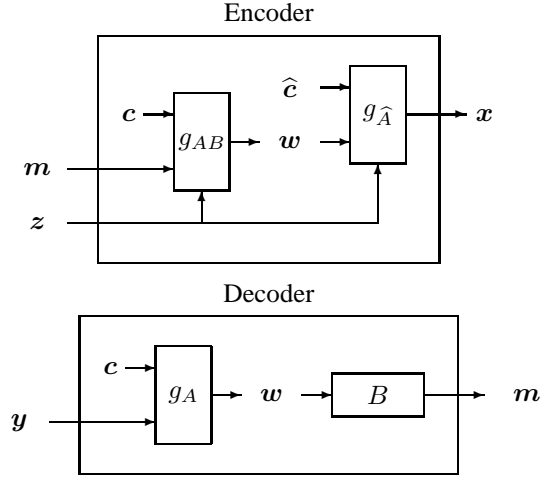

In the following, we assume that $\mu_{XW|Z}$ is fixed.
We fix functions
\begin{align*}
 A&:\W^n\to\W^{l_{\A}}
 \\
 B&:\W^n\to\W^{l_{\B}}
 \\
 \hA&:\X^n\to\X^{l_{\hcA}}
\end{align*}
and vectors $\cc\in\W^{l_{\A}}$ and $\hcc\in\X^{l_{\hA}}$
available to construct an encoder and a decoder, where
\begin{align*}
 l_{\A}
 &\equiv \frac{n[H(W|Y)+\eA]}{\log|\W|}
 \\
 l_{\B}
 &\equiv \frac{n[H(W|Z)-H(W|Y)-\eB]}{\log|\W|}
 \\
 &=\frac{n[I(W;Y)-I(W;Z)-\eB]}{\log|\W|}
 \\
 l_{\hcA}&\equiv \frac{n[H(X|Z,W)-\ehA]}{\log|\X|}.
\end{align*}
We define the encoder and the decoder
\begin{align*}
 \Encoder&:\W^{l_{\B}}\times\Z^n\to\X^n
 \\
 \Decoder&:\Y^n\to\W^{l_{\B}}
\end{align*}
as
\begin{align*}
 \Encoder(\mm|\zz)
 &\equiv g_{\hA}(\hcc|\zz,g_{AB}(\cc,\mm|\zz))
 \\
 \Decoder(\mm)
 &\equiv Bg_A(\cc|\yy),
\end{align*}
where
\begin{align*}
 g_{AB}(\cc,\mm|\zz)
 &\equiv\arg\max_{\ww'\in\C_{AB}(\cc,\mm)}\mu_{W|Z}(\ww'|\zz)
 \\
 g_{\hA}(\hcc|\zz,\ww)
 &\equiv\arg\max_{\xx'\in\C_{\hA}(\hcc)}\mu_{X|ZW}(\xx|\zz,\ww)
 \\
 g_{A}(\cc|\yy)
 &\equiv\arg\max_{\ww'\in\C_A(\cc)}\mu_{W|Y}(\ww'|\yy).
\end{align*}

Let $M$ be the random variable corresponding to
the message $\mm$, where the probability $p_M(\mm)$ and
$p_{MZ}(\mm,\zz)$ are given by
\begin{align*}
 p_M(\mm)&\equiv
 \begin{cases}
	\frac 1{|\im\B|},
	&\text{if}\ \mm\in\im\B
	\\
	0,
	&\text{if}\ \mm\notin\im\B
 \end{cases}
 \\
 p_{MZ}(\mm,\zz)
 &\equiv p_M(\mm)\mu_Z(\zz).
\end{align*}
The rate $\Rate(B)$ of this code is given by
\begin{align*}
 \Rate(B)&\equiv\frac{\log|\im\B|}n
 \\
 &=\frac{l_B\log|\W|}n-\frac{\log\frac{|\W|^{l_B}}{|\im\B|}}n
\end{align*}
and the decoding error probability $\Error_{Y|XZ}(A,B,\hA,\cc,\hcc)$
is given by
\begin{align*}
 \Error_{Y|XZ}(A,B,\hA,\cc,\hcc)
 &\equiv
 \sum_{\mm,\yy,\zz}
 p_M(\mm)\mu_Z(\zz)\mu_{Y|XZ}(\yy|\Encoder(\mm,\zz),\zz)
 \chi(\Decoder(\yy)\neq\mm).
\end{align*}

In the following, we provide an
intuitive interpretation of the construction of the code,
which is illustrated in Fig.~\ref{fig:gp-code}.
Assume that $\cc$ is shared by the encoder and the decoder.
For $\cc$, a message $\mm$, and a side information $\zz$,
the function $g_{AB}$ generates a typical sequence
$\ww\in\T_{Z,\gamma}(\zz)$
and the function $g_{\hA}$ generates a typical sequence
$\xx\in\T_{X|WZ,\gamma}(\ww,\zz)$
as a channel input.
The decoder 
reproduces the channel input $\ww$
by using $g_{A}$ from $\cc$ and a channel output $\yy$.
Since $(\ww,\yy)$ is jointly typical and $B\ww=\mm$,
the decoding succeed if
the rate of $\cc$ is greater than $H(W|Y)$
to satisfy the collision-resistant property.
On the other hand,
the rate of $\hcc$ should be less than $H(X|Z,W)$ and
the total rate of $\cc$ and $\mm$
should be less than $H(W|Z)$ to satisfy the saturating property.
Then we can set the encoding rate of $\mm$ close to
$H(W|Z)-H(W|Y)=I(W;Y)-I(W;Z)$.

We have the following theorem.
It should be noted that
alphabets $\X$, $\Y$, $\W$, and $\Z$
are allowed to be non-binary,
and the channel 
is allowed to be asymmetric.
\begin{thm}
 \label{thm:gp}
 For given $\eA,\eB,\ehA>0$ satisfying
 \begin{gather}
	\eB-\eA\leq\sqrt{6[\eB-\eA]}\log|\Z||\W|<\eA
	\label{eq:gp-e1}
	\\
	2\zeta_{\Y\W}(6\ehA)<\eA,
	\label{eq:gp-e2}
 \end{gather}
 assume
 that $(\A,p_A)$, $(\A\times\B,p_A\times p_B)$ and $(\hcA,p_{\hA})$
 have hash property.
 Let $\mu_{Y|XZ}$ be the conditional probability distribution
 of a stationary memoryless channel.
 Then, for all $\delta>0$ and sufficiently large $n$
 there are functions (sparse matrices) $A\in\A$, $B\in\B$, $\hA\in\hcA$,
 and vectors $\cc\in\im A$, $\hcc\in\im\hA$ such that
\begin{gather*}
 \Rate(B) \geq I(W;Y)-I(W;Z)-\eB-\delta
 \\
 \Error_{Y|XZ}(A,B,\hA,\cc,\hcc)<\delta.
\end{gather*}
 By assuming that $\mu_{XW|Z}$ attains the Gel'fand-Pinsker bound,
 and $\delta\to 0$, $\eB\to 0$, 
 the rate of the proposed code is close to this bound.
\end{thm}
The proof is given in Section \ref{sec:proof-gp}.
It should be noted that Theorem~\ref{thm:channel}
is a special case of Gel'fand-Pinsker problem
with  $|\Z|\equiv 1$ and $W\equiv X$.

\begin{rem}
In \cite{MW06b}, the code for the Gel'fand-Pinsker problem
is proposed by using a combination of two sparse matrices
when all the alphabets are binary
and the channel side information and noise are additive.
In their constructed encoder,
they obtain a vector called the `middle layer'
by using one of two matrices
and obtain a channel input
by operating another matrix on the middle layer and adding
the side information.
In our construction,
we obtain $\ww$ by using two matrices $A$, $B$, and $g_{AB}$,
where the dimension of $\ww$ differs from that of the middle layer.
We obtain channel input $\xx$ by
using $\hA$ and $g_{\hA}$ instead of adding the side information.
It should be noted that
our approach is based on the construction
of the channel code presented in \cite{SWLDPC}\cite{MM08},
which is also different from the construction presented
in \cite{GA68}\cite{BB}.
\end{rem}

\subsection{Wyner-Ziv Problem}
\label{sec:wz}

In this section we consider 
the Wyner-Ziv problem introduced in \cite{WZ76}
(illustrated in Fig.\ \ref{fig:wz}).

First, we construct a code for the
standard lossy source coding problem
illustrated in Fig.\ \ref{fig:rd},
which is a special case of the Wyner-Ziv problem.

Let $\rho:\X\times\Y\to[0,\infty)$ be the distortion measure
satisfying
\[
\rho_{\max}\equiv \max_{x,y}\rho(x,y)<\infty.
\]
We define $\rho_n(\xx,\yy)$ as
\begin{align*}
 \rho_n(\xx,\yy)&\equiv\sum_{i=1}^n\rho(x_i,y_i)
\end{align*}
for each $\xx\equiv(x_1,\ldots,x_n)$ and $\yy\equiv(y_1,\ldots,y_n)$.
For a probability distribution $\mu_{X}$,
the rate-distortion function $R_X(D)$ is given by
\[
 R_{X}(D) =
 \min_{\substack{
 \mu_{Y|X}:
 \\
 E_{XY}\lrB{\rho(X,Y)}\leq D
 }}
 I(X;Y),
\]
where the minimum is taken over
all conditional probability distributions $\mu_{Y|X}$
and the joint distribution
$\mu_{XY}$ of $(X,Y)$ is given by
\[
 \mu_{XY}(\xx,\yy)\equiv \mu_{X}(\xx)\mu_{Y|X}(\yy|\xx).
\]

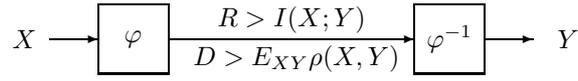
\begin{figure}[t]
\begin{center}
\unitlength 0.65mm
\begin{picture}(120,20)(0,5)
\put(5,7){\makebox(0,0){$X$}}
\put(10,7){\vector(1,0){10}}
\put(20,0){\framebox(14,14){$\Encoder$}}
\put(34,7){\vector(1,0){50}}
\put(59,11){\makebox(0,0){$R>I(X;Y)$}}
\put(84,0){\framebox(14,14){$\Decoder$}}
\put(98,7){\vector(1,0){10}}
\put(115,7){\makebox(0,0){$Y$}}
\put(60,3){\makebox(0,0){$D>E_{XY}\rho(X,Y)$}}
\end{picture}
\end{center}
\caption{Lossy Source Coding}
\label{fig:rd}
\end{figure}

\begin{figure}[t]
\begin{center}
\unitlength 0.4mm
\begin{picture}(176,70)(0,0)
\put(82,60){\makebox(0,0){Encoder}}
\put(30,35){\makebox(0,0){$\cc$}}
\put(35,35){\vector(1,0){10}}
\put(0,17){\makebox(0,0){$\xx$}}
\put(10,17){\vector(1,0){35}}
\put(45,10){\framebox(18,32){$g_A$}}
\put(63,26){\vector(1,0){10}}
\put(83,26){\makebox(0,0){$\yy$}}
\put(93,26){\vector(1,0){10}}
\put(103,19){\framebox(30,14){$B$}}
\put(133,26){\vector(1,0){20}}
\put(167,26){\makebox(0,0){$B\yy$}}
\put(20,0){\framebox(124,52){}}
\end{picture}
\\
\begin{picture}(176,70)(0,0)
\put(82,60){\makebox(0,0){Decoder}}
\put(65,35){\makebox(0,0){$\cc$}}
\put(70,35){\vector(1,0){10}}
\put(30,17){\makebox(0,0){$B\yy$}}
\put(45,17){\vector(1,0){35}}
\put(80,10){\framebox(18,32){$g_{AB}$}}
\put(98,26){\vector(1,0){20}}
\put(128,26){\makebox(0,0){$\yy$}}
\put(55,0){\framebox(54,52){}}
\end{picture}
\end{center}
\caption{Construction of Lossy Source Code}
\label{fig:lossy-code}
\end{figure}
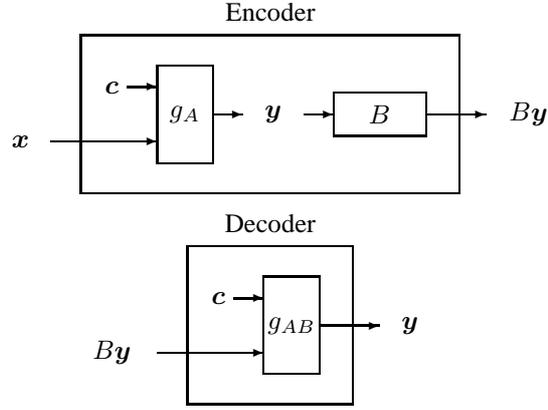

The code for this problem is given in the following
(illustrated in Fig.\ \ref{fig:lossy-code}).
We fix functions
\begin{align*}
 A&:\Y^n\to\Y^{l_{\A}}
 \\
 B&:\Y^n\to\Y^{l_{\B}}
\end{align*}
and a vector $\cc\in\Y^{l_{\A}}$
available to construct an encoder and a decoder, where
\begin{align*}
 l_{\A}
 &\equiv \frac{n[H(Y|X)-\eA]}{\log|\Y|}
 \\
 l_{\B}
 &\equiv \frac{n[I(Y;X)+\eB]}{\log|\Y|}.
\end{align*}
We define the encoder and the decoder
\begin{align*}
 \Encoder&:\X^{n}\to\Y^{l_{\B}}
 \\
 \Decoder&:\Y^{l_{\B}}\to\Y^{n}
\end{align*}
as
\begin{align*}
 \Encoder(\xx)
 &\equiv Bg_A(\cc|\xx)
 \\
 \Decoder(\bb)
 &\equiv g_{AB}(\cc,\bb),
\end{align*}
where
\begin{align*}
 g_{A}(\cc|\xx)
 &\equiv\arg\max_{\yy'\in\C_A(\cc)}\mu_{Y|X}(\yy'|\xx)
 \\
 g_{AB}(\cc,\bb)
 &\equiv\arg\max_{\yy'\in\C_{AB}(\cc,\bb)}\mu_{Y}(\yy').
\end{align*}

In the following, we provide an
intuitive interpretation of the construction of the code,
which is illustrated in Fig.~\ref{fig:lossy-code}.
Assume that $\cc$ is shared by the encoder and the decoder.
For $\cc$ and $\xx$,
the function $g_A$ generates
$\yy$ such that $A\yy=\cc$ and $(\xx,\yy)$ is a jointly typical sequence.
The rate of $\cc$
should be less than $H(Y|X)$ to satisfy the saturation property.
Then the encoder obtains the codeword $B\yy$.
The decoder obtains the reproduction $\yy$
by using $g_{AB}$ from $\cc$ and the codeword $B\yy$
if the rate of $\cc$ and $B\yy$ is greater than $H(Y)$
to satisfy the collision-resistant property.
Then we can set the encoding rate close to $H(Y)-H(Y|X)=I(X;Y)$.
Since $(\xx,\yy)$ is jointly typical,
$\rho_n(\xx,\yy)$ is close to the distortion criterion.

We have the following theorem.
It should be noted that
a source is allowed to be non-binary and unbiased
and the distortion measure $\rho$ is arbitrary.
\begin{thm}
 \label{thm:lossy}
 For given $\eA>,\eB>0$ satisfying
 \begin{equation*}
	\eA+2\zeta_{\Y}(3\eA)<\eB,
 \end{equation*}
 assume
 that $(\A,p_A)$ and $(\B,p_B)$ have hash property.
 Let $X$ be a stationary memoryless source.
 Then for all sufficiently large $n$
 there are functions (sparse matrices) $A\in\A$, $B\in\B$, and a vector
 $\cc\in\im\A$ such that
 \begin{align*}
	\Rate(B)
	&=I(X;Y)+\eB
	\\
	\frac{E_{X}\lrB{\rho_n(X^n,\Decoder(\Encoder(X^n)))}}n
	&\leq
	E_{XY}\lrB{\rho(X,Y)}
	+3|\X||\Y|\rho_{\max}\sqrt{\eA}.
 \end{align*}
 By  assuming that  $\mu_{Y|X}$ attain
 the rate-distortion bound
 and by letting $\eA,\eB\to 0$,
 the rate-distortion pair of the proposed code
 is close to this bound.
\end{thm}

Next,
we consider  the Wyner-Ziv problem introduced in \cite{WZ76}
(illustrated in Fig.\ \ref{fig:wz}).
Let $\rho:\X\times\W\to[0,\infty)$ be the distortion measure
satisfying
\[
\rho_{\max}\equiv \max_{x,w}\rho(x,w)<\infty.
\]
We define $\rho_n(\xx,\ww)$ as
\begin{align*}
 \rho_n(\xx,\ww)&\equiv\sum_{i=1}^n\rho(x_i,w_i)
\end{align*}
for each $\xx\equiv(x_1,\ldots,x_n)$ and $\ww\equiv(w_1,\ldots,w_n)$.
For a probability distribution $\mu_{XZ}$,
the rate-distortion function $R_{X|Z}(D)$ is given by
\[
 R_{X|Z}(D) =
 \min_{\substack{
 \mu_{Y|X},f:
 \\
 E_{XYZ}\lrB{\rho(X,f(Y,Z))}\leq D
 }}
 \lrB{I(X;Y)-I(Y;Z)},
\]
where the minimum is taken over
all conditional probability distributions $\mu_{Y|X}$
and functions $f:\Y\times\Z\to\W$ and the joint distribution
$\mu_{XYZ}$ of $(X,Y,Z)$ is given by
\begin{equation}
 \mu_{XYZ}(\xx,\yy,\zz)\equiv \mu_{XZ}(\xx,\zz)\mu_{Y|X}(\yy|\xx).
	\label{eq:markov-wz}
\end{equation}

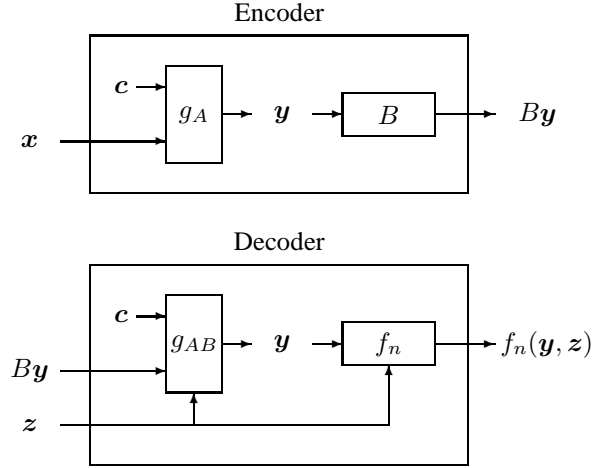
\begin{figure}[t]
\begin{center}
\unitlength 0.4mm
\begin{picture}(176,70)(0,0)
\put(82,60){\makebox(0,0){Encoder}}
\put(30,35){\makebox(0,0){$\cc$}}
\put(35,35){\vector(1,0){10}}
\put(0,17){\makebox(0,0){$\xx$}}
\put(10,17){\vector(1,0){35}}
\put(45,10){\framebox(18,32){$g_A$}}
\put(63,26){\vector(1,0){10}}
\put(83,26){\makebox(0,0){$\yy$}}
\put(93,26){\vector(1,0){10}}
\put(103,19){\framebox(30,14){$B$}}
\put(133,26){\vector(1,0){20}}
\put(167,26){\makebox(0,0){$B\yy$}}
\put(20,0){\framebox(124,52){}}
\end{picture}
\\
\begin{picture}(176,90)(0,6)
\put(82,80){\makebox(0,0){Decoder}}
\put(30,55){\makebox(0,0){$\cc$}}
\put(35,55){\vector(1,0){10}}
\put(0,37){\makebox(0,0){$B\yy$}}
\put(10,37){\vector(1,0){35}}
\put(45,30){\framebox(18,32){$g_{AB}$}}
\put(63,46){\vector(1,0){10}}
\put(83,46){\makebox(0,0){$\yy$}}
\put(93,46){\vector(1,0){10}}
\put(103,39){\framebox(30,14){$f_n$}}
\put(133,46){\vector(1,0){20}}
\put(170,46){\makebox(0,0){$f_n(\yy,\zz)$}}
\put(20,6){\framebox(124,66){}}
\put(0,19){\makebox(0,0){$\zz$}}
\put(10,19){\line(1,0){108}}
\put(118,19){\vector(0,1){20}}
\put(54,19){\vector(0,1){11}}
\end{picture}
\end{center}
\caption{Construction of Wyner-Ziv Source Code}
\label{fig:wz-code}
\end{figure}

In the following, we assume that $\mu_{Y|X}$ is fixed.
We fix functions
\begin{align*}
 A&:\Y^n\to\Y^{l_{\A}}
 \\
 B&:\Y^n\to\Y^{l_{\B}}
\end{align*}
and a vector $\cc\in\Y^{l_{\A}}$
available to construct an encoder and a decoder, where
\begin{align*}
 l_{\A}
 &\equiv \frac{n[H(Y|X)-\eA]}{\log|\Y|}
 \\
 l_{\B}
 &\equiv
 \frac{n[H(Y|Z)-H(Y|X)+\eB]}{\log|\Y|}
 \\
 &=\frac{n[I(X;Y)-I(Y;Z)+\eB]}{\log|\Y|}.
\end{align*}
We define the encoders and the decoder (illustrated in Fig.\ \ref{fig:wz-code})
\begin{align*}
 \Encoder&:\X^n\to\Y^{l_{\B}}
 \\
 \Decoder&:\Y^{l_{\B}}\times\Z^n\to\W^n
\end{align*}
as
\begin{align*}
 \Encoder(\xx)
 &\equiv Bg_A(\cc|\xx)
 \\
 \Decoder(\bb|\zz)
 &\equiv f_n(g_{AB}(\cc,\bb|\zz),\zz)
\end{align*}
where
\begin{align*}
 g_{A}(\cc|\xx)
 &\equiv\arg\max_{\yy'\in\C_A(\cc)}\mu_{Y|X}(\yy'|\xx)
 \\
 g_{AB}(\cc,\bb|\zz)
 &\equiv\arg\max_{\yy'\in\C_{AB}(\cc,\bb)}\mu_{Y|Z}(\yy'|\zz)
\end{align*}
and we define $f_n(\yy,\zz)\equiv(w_1,\ldots,w_n)$ by
\begin{align*}
 w_i&\equiv f(y_i,z_i)
\end{align*}
for each $\yy\equiv(y_1,\ldots,y_n)$ and
$\zz\equiv(z_1,\ldots,z_n)$.

The rate $\Rate(B)$ of this code is given by
\[
 \Rate(B)\equiv\frac{l_{\B}\log|\Y|}n.
\]

In the following, we provide an
intuitive interpretation of the construction of the code,
which is illustrated in Fig.~\ref{fig:wz-code}.
Assume that $\cc$ is shared by the encoder and the decoder.
For $\cc$ and $\xx$,
the function $g_A$ generates
$\yy$ such that $A\yy=\cc$ and $(\xx,\yy)$ is a jointly typical sequence.
The rate of $\cc$
should be less than $H(Y|X)$ to satisfy the saturation property.
Then the encoder obtains the codeword $B\yy$.
The decoder
obtains the reproduction $\yy$
by using $g_{AB}$ from $\cc$, the codeword $B\yy$, and the side
information $\zz$
if the rate of $\cc$ and $B\yy$ is greater than $H(Y|Z)$
to satisfy the collision-resistant property.
Then we can set the encoding rate close to
$H(Y|Z)-H(Y|X)=I(X;Y)-I(Y;Z)$.
Since $(\xx,\yy,\zz)$ is jointly typical,
$\rho_n(\xx,f(\yy,\zz))$ is close to the distortion criterion.

We have the following theorem.
It should be noted that
a source is allowed to be non-binary and unbiased,
side information is allowd to be asymmetric,
and the distortion measure $\rho$ is arbitrary.
\begin{thm}
 \label{thm:wz}
 For given $\eA,\eB>0$ satisfying
 \begin{equation}
	\eA+2\zeta_{\Y\Z}(3\eA)<\eB,
	 \label{eq:wz-e}
 \end{equation}
 assume that $(\A,p_A)$ and $(\B,p_B)$ have hash property.
 Let $(X,Z)$ be a pair of stationary memoryless sources.
 Then for all sufficiently large $n$
 there are functions (sparse matrices) $A\in\A$, $B\in\B$, and a vector
 $\cc\in\im\A$ such that
 \begin{align*}
	\Rate(B)
	&=I(X;Y)-I(Y;Z)+\eB
	\\
	\frac{E_{XZ}\lrB{\rho_n(X^n,\Decoder(\Encoder(X^n),Z^n))}}n
	&
	\leq
	E_{XYZ}\lrB{\rho(X,f(Y,Z))}
	+3|\X||\Y||\Z|\rho_{\max}\sqrt{\eA}.
 \end{align*}
 By  assuming that  $\mu_{Y|X}$ and  $f$ attain
 the Wyner-Ziv bound and letting $\eA,\eB\to 0$,
 the rate-distortion pair of the proposed code
 is close to this bound.
\end{thm}
The proof is given in Section \ref{sec:proof-wz}.
It should be noted that Theorem~\ref{thm:lossy}
is a special case of the Wyner-Ziv problem
with  $|\Z|\equiv 1$, $\W\equiv\Y$, and $f(y,z)\equiv y$.

\begin{rem}
In  \cite{MY03}\cite{M04}\cite{MW06a}\cite{M06}\cite{GV07},
the lossy source code is proposed using sparse matrices
for the  binary alphabet and Hamming distance.
In their constructed encoder proposed in
 \cite{M04}\cite{MW06a}\cite{M06}\cite{GV07},
they obtain a codeword vector called the `middle layer' (see \cite{MW06a})
by using a matrix.
In their constructed decoder,
they operate another matrix on the codeword vector.
In our construction of the decoder,
we obtain the reproduction $\yy$
by using a sparse matrix $A$ and $g_{A}$
and compress $\yy$ with another matrix $B$.
It should be noted that
the dimension of $\yy$ is different from that of the middle layer
and we need ML decoder $g_{AB}$
in the construction of the decoder,
because $\yy$ is compressed by using $B$.
In \cite{MW06b},
the code for the Wyner-Ziv problem is proposed and there are similar
differences.
Our approach is based on the code presented in \cite{MM07}\cite{SWLOSSY}
and similar to the code presented in
 \cite{WAD03aew}\cite{WAD03ieice}\cite{ZSE02}.
\end{rem}

\subsection{One-helps-one Problem}

In this section, we consider the One-helps-one problem
illustrated in Fig.~\ref{fig:psi}.
The achievable rate region for this problem is given by
a set of encoding rate pair $(\Rate_X,\Rate_Y)$
satisfying
	\begin{align*}
	 &R_X\geq H(X|Z)
	 \\
	 &R_Y\geq I(Y;Z),
	\end{align*}
where the joint distribution $\mu_{XYZ}$ is given by
\begin{equation}
 \mu_{XYZ}(x,y,z)=\mu_{XY}(x,y)\mu_{Z|Y}(z|y).
	\label{eq:markov-psi}
\end{equation}

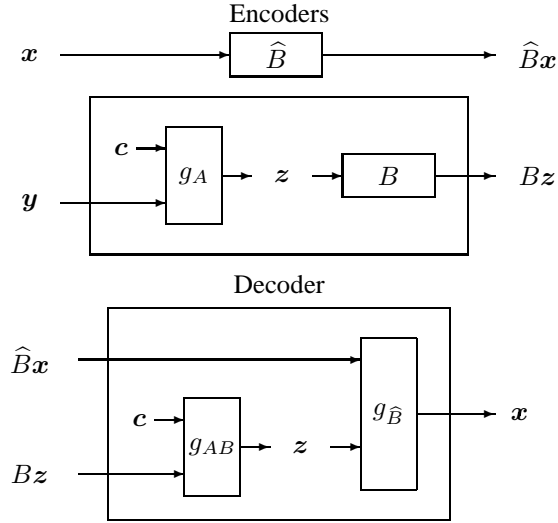
\begin{figure}[t]
\begin{center}
\unitlength 0.4mm
\begin{picture}(176,80)(0,0)
\put(82,80){\makebox(0,0){Encoders}}
\put(0,66){\makebox(0,0){$\xx$}}
\put(10,66){\vector(1,0){56}}
\put(66,59){\framebox(30,14){$\hB$}}
\put(96,66){\vector(1,0){57}}
\put(167,66){\makebox(0,0){$\hB\xx$}}
\put(30,35){\makebox(0,0){$\cc$}}
\put(35,35){\vector(1,0){10}}
\put(0,17){\makebox(0,0){$\yy$}}
\put(10,17){\vector(1,0){35}}
\put(45,10){\framebox(18,32){$g_A$}}
\put(63,26){\vector(1,0){10}}
\put(83,26){\makebox(0,0){$\zz$}}
\put(93,26){\vector(1,0){10}}
\put(103,19){\framebox(30,14){$B$}}
\put(133,26){\vector(1,0){20}}
\put(167,26){\makebox(0,0){$B\zz$}}
\put(20,0){\framebox(124,52){}}
\end{picture}
\\
\begin{picture}(176,90)(0,0)
\put(82,80){\makebox(0,0){Decoder}}
\put(36,35){\makebox(0,0){$\cc$}}
\put(41,35){\vector(1,0){10}}
\put(0,17){\makebox(0,0){$B\zz$}}
\put(16,17){\vector(1,0){35}}
\put(51,10){\framebox(18,32){$g_{AB}$}}
\put(69,26){\vector(1,0){10}}
\put(89,26){\makebox(0,0){$\zz$}}
\put(99,26){\vector(1,0){10}}
\put(26,2){\framebox(112,70){}}
\put(0,55){\makebox(0,0){$\hB\xx$}}
\put(16,55){\vector(1,0){93}}
\put(109,12){\framebox(18,50){$g_{\hB}$}}
\put(127,37){\vector(1,0){26}}
\put(161,37){\framebox(0,0){$\xx$}}
\end{picture}
\end{center}
\caption{Construction of One-helps-one Source Code}
\label{fig:psi-code}
\end{figure}

In the following, we
construct a code by combining a Slepian-Wolf code
and a lossy source code.
We assume that $\mu_{Z|Y}$ is fixed.
We fix functions
\begin{align*}
 \hB&:\X^n\to\X^{l_{\hcB}}
 \\
 A&:\Z^n\to\Z^{l_{\A}}
 \\
 B&:\Z^n\to\Z^{l_{\B}}
\end{align*}
and a vector $\cc\in\Z^{l_{\A}}$
available to construct an encoder and a decoder, where
\begin{align*}
 l_{\hcB}&\equiv
 \frac{n[H(X|Z)+\ehB]}
 {\log|\X|}
 \\
 l_{\A}&\equiv \frac{n[H(Z|Y)-\eA]}{\log|\Z|}
 \\
 l_{\B}&\equiv \frac{n[I(Y;Z)+\eB]}{\log|\Z|}.
\end{align*}
We define the encoders and the decoder (illustrated in
Fig.\ \ref{fig:psi-code})
\begin{align*}
 \Encoder_X&:\X^n\to\X^{l_{\hcB}}
 \\
 \Encoder_Y&:\Y^n\to\Z^{l_{\B}}
 \\
 \Decoder&:\X^{l_{\hcB}}\times\Z^{l_{\B}}\to\X^n
\end{align*}
as
\begin{align*}
 \Encoder_X(\xx)&\equiv \hB\xx
 \\
 \Encoder_Y(\yy)&\equiv Bg_A(\cc|\yy)
 \\
 \Decoder(\bb_X,\bb_Y)
 &\equiv g_{\hB}(\bb_X, g_{AB}(\cc,\bb_Y)),
\end{align*}
where
\begin{align*}
 g_{A}(\cc|\yy)
 &\equiv\arg\max_{\zz'\in\C_A(\cc)}\mu_{Z|Y}(\zz'|\yy)
 \\
 g_{AB}(\cc,\bb_Y)
 &\equiv\arg\max_{\zz'\in\C_{AB}(\cc,\bb_Y)}\mu_{Z}(\zz')
 \\
 g_{\hB}(\bb_X|\zz)
 &\equiv\arg\max_{\xx'\in\C_{\hB}(\bb_X)}\mu_{X|Z}(\xx'|\zz).
\end{align*}

The pair of encoding rates $(\Rate_X,\Rate_Y)$ is given by
\begin{align*}
 \Rate_X\equiv\frac{l_{\hB}\log|\X|}n
 \\
 \Rate_Y\equiv\frac{l_B\log|\Z|}n
\end{align*}
and the decoding error probability $\Error_{XY}(A,B,\hB,\cc)$ is given by
\begin{align*}
 \Error_{XY}(A,B,\hB,\cc)
 &
 \equiv
 \mu_{XY}\lrsb{\lrb{
 (\xx,\yy):
 \Decoder(\Encoder_X(\xx),\Encoder_Y(\yy))\neq \xx
 }}.
\end{align*}

We have the following theorem.
\begin{thm}
 \label{thm:psi}
 For given $\eA,\eB,\ehB>0$ satisfying
 \begin{gather}
	\eB>\eA+\zeta_{\Z}(3\eA)
	 \label{eq:psi-e1}
	 \\
	\ehB>2\zeta_{\X\Z}(3\eA),
	 \label{eq:psi-e2}
 \end{gather}
 assume that $(\A,p_A)$, $(\B,p_B)$, and $(\hcB,p_{\hB})$
 have hash property.
 Let $(X,Y)$ be a pair of stationary memoryless sources.
 Then, for any $\delta>0$ and and all sufficiently large $n$,
 there are functions (sparse matrices) $A\in\A$, $B\in\B$, $\hB\in\hcB$,
 and a vector $\cc\in\im\A$ such that
 \begin{gather*}
	\Rate_X = H(X|Z)+\ehB
	\\
	\Rate_Y = I(X;Z)+\eB
	\\
	\Error_{XY}(A,B,\hB,\cc)
	\leq
	\delta.
 \end{gather*}
\end{thm}
The proof is given in Section \ref{sec:proof-psi}.

\section{Proof of Lemmas and Theorems}
\label{sec:proof}
In the proof, we use the method of types,
which is given in  Appendix.
Throughout this section,
we assume that
the probability distributions of $p_C$, $p_{\hC}$, $p_M$ are uniform
and the random variables $A$, $B$, $\hA$, $\hB$, $C$, $\hC$ and
$M$ are mutually independent.

\subsection{Proof of Lemmas \ref{lem:Anotempty}--\ref{lem:md-typical}}
\label{sec:proof-hash}

We prepare the following 
two lemmas,
which come from the fact that
$p_C$ is the uniform distribution on $\im\A$
and random variables $A$ and $C$ are mutually independent.
\begin{lem}
 \label{lem:E}
Let $p_{A}$ be the distribution on the set of functions
and $p_{C}$ be the uniform distribution on $\im\A$.
We assume that a joint distribution $p_{AC}$ satisfies
\[
 p_{AC}(A,\cc)=p_A(A)p_C(\cc)
\]
for any $A$ and $\cc\in\im\A$.
Then 
\begin{equation}
 \label{eq:EC}
	\sum_{c}p_C(\cc)\chi(A\uu=\cc)
	=\frac 1{|\im\A|}
\end{equation}
for any $A$ and $\uu\in\U^n$,
\begin{equation}
 \label{eq:EAC}
	\sum_{A,\cc}p_{AC}(A,\cc)
	\chi(A\uu=\cc)=\frac 1{|\im\A|}
\end{equation}
for any $\uu\in\U^n$, and
\begin{align}
 \begin{split}
	\sum_{A,\cc}p_{AC}(A,\cc)
	\left|\G\cap\C_A(\cc)\right|
	&
	=
	\sum_{\uu\in\G}
	\sum_{A,\cc}p_{AC}(A,\cc)\chi(A\uu=\cc)
	\\
	&
	=\frac {|\G|}{|\im\A|}
 \end{split}
 \label{eq:EACG}
 \\
 p_{AC}\lrsb{\lrb{
 (A,\cc):
 \G\cap\C_A(\cc)\neq \emptyset
 }}
 &\leq
 \frac{|\G|}{|\im\A|}
 \label{eq:largeprob}
\end{align}
for any $\G\subset\U^n$.
\end{lem}
\begin{proof}
 First, we prove (\ref{eq:EC}).
 Since $A\uu$ is determined uniquely, we have
\[
 \sum_{\cc}\chi(A\uu=\cc)=1.
\]
Then we have
\begin{align*}
 \sum_{\cc}p_{C}(\cc)\chi(A\uu=\cc)
 &=
 \sum_{\cc}\frac{\chi(A\uu=\cc)}{|\im\A|}
 \\
 &=\frac 1{|\im\A|}.
\end{align*}

Next, we prove (\ref{eq:EAC}).
From (\ref{eq:EC}), we have
\begin{align*}
 \sum_{A,\cc}p_{AC}(A,\cc)
 \chi(A\uu=\cc)
 &=
 \sum_{A}p_{A}(A)\sum_{\cc}p_{\cc}(\cc)
 \chi(A\uu=\cc)
 \\
 &=
 \sum_{A}\frac{p_A(A)}{|\im\A|}
 \\
 &=
 \frac 1{|\im\A|}.
\end{align*}

Next, we prove (\ref{eq:EACG}).
From (\ref{eq:EAC}), we have
\begin{align*}
 \sum_{A,\cc}p_{AC}(A,\cc)
 \left|\G\cap\C_A(\cc)\right|
 &=
 \sum_{A,\cc}p_{AC}(A,\cc)
 \sum_{\uu\in\G}\chi(A\uu=\cc)
 \\
 &=
 \sum_{\uu\in\G}
 \sum_{A,\cc}p_{AC}(A,\cc)
 \chi(A\uu=\cc)
 \\
 &=
 \sum_{\uu\in\G}
 \frac 1{|\im\A|}
 \\
 &=\frac {|\G|}{|\im\A|}.
\end{align*}

Finally, we prove (\ref{eq:largeprob}).
From (\ref{eq:EACG}), we have
\begin{align*}
 p_{AC}\lrsb{\lrb{
 (A,\cc):
 \G\cap\C_A(\cc)\neq \emptyset
 }}
 &=
 p_{AC}\lrsb{\lrb{
 (A,\cc):
 \exists \uu\in\G\cap C_A(\cc)
 }}
 \\
 &
 \leq
 \sum_{\uu\in\G}
 p_{AC}\lrsb{\lrb{
 (A,\cc):
 \uu\in C_A(\cc)
 }}
 \\
 &=
 \sum_{\uu\in\G}
 \sum_{A,\cc}p_{AC}(A,\cc)\chi(A\uu=\cc)
 \\
 &
 =\frac{|\G|}{|\im\A|}.
\end{align*}
\end{proof}

{\it Proof of Lemma \ref{lem:Anotempty}:}
Since $(\A,p_A)$ has an $(\aalpha_A,\bbeta_A)$-hash property,
we have
\begin{align*}
 p_A\lrsb{\lrb{
 A:
 \lrB{\G\setminus\{\uu\}}\cap\C_A(A\uu)\neq \emptyset
 }}
 &\leq
 \sum_{\uu'\in\G\setminus\{\uu\}}
 p_A\lrsb{\lrb{
 A: A\uu = A\uu'
 }}
 \\
 &\leq |\{\uu\}\cap\lrB{\G\setminus\{\uu\}}|
 +\frac{\left|\G\setminus\{\uu\}\right|\alpha_A}{|\im\A|}
 + \min\{|\{\uu\}|,\left|\G\setminus\{\uu\}\right|\}\beta_A
 \\
 &\leq
 \frac{|\G|\alpha_A}{|\im\A|} + \beta_A.
\end{align*}
\hfill\QED

{\it Proof of Lemma \ref{lem:noempty}:}
First, since $(\A,p_A)$ has an $(\aalpha_A,\bbeta_A)$-hash property, we have
\begin{align}
 \sum_{\uu,\uu'\in\T}
 \sum_{A,\cc}p_{AC}(A,\cc)
 \chi(A\uu=\cc)\chi(A\uu'=\cc)
 &=
 \sum_{\uu,\uu'\in\T}
 \sum_{A,\cc}p_{AC}(A,\cc)
 \chi(A\uu=\cc)\chi(A\uu = A\uu')
 \notag
 \\
 &=
 \sum_{\uu,\uu'\in\T}
 \sum_{A}p_A(A)\chi(A\uu=A\uu')
 \sum_{\cc}p_C(\cc)\chi(A\uu=\cc)
 \notag
 \\
 &=
 \frac 1{|\im\A|}
 \sum_{\uu,\uu'\in\T}
 \sum_{A}p_A(A)\chi(A\uu=A\uu')
 \notag
 \\
 &\leq
 \frac{|\T|}{|\im\A|}
 +\frac{|\T|^2\alpha_A}{|\im\A|^2}
 +\frac{|\T|\beta_A}{|\im\A|},
 \label{eq:variance}
\end{align}
where the third equality comes from (\ref{eq:EC}).

Next, we have
 \begin{align}
	&
	\sum_{A,\cc}p_{AC}(A,\cc)
	\lrB{
	\sum_{\uu\in\T}
	\chi(A\uu=\cc)-\frac{|\T|}{|\im\A|}
	}^2
	\notag
	\\*
	&=
	\sum_{A,\cc}p_{AC}(A,\cc)
	\lrB{
	\sum_{\uu\in\T}
	\chi(A\uu=\cc)
	}^2
	-\frac {2|\T|}{|\im\A|}
	\sum_{A,\cc}p_{AC}(A,\cc)\sum_{\uu\in\T}\chi(A\uu=\cc)
	+
	\frac {|\T|^2}{|\im\A|^2}
	\notag
	\\
	&=
	\sum_{\uu,\uu'\in\T}
	\sum_{A,\cc}p_{AC}(A,\cc)
	\chi(A\uu=\cc)\chi(A\uu'=\cc)
	-\frac {2|\T|}{|\im\A|}
	\sum_{\uu\in\T}\sum_{A,\cc}p_{AC}(A,\cc)\chi(A\uu=\cc)
	+
	\frac {|\T|^2}{|\im\A|^2}
	\notag
	\\
	&\leq
	\frac{|\T|^2\lrB{\alpha_A-1}}{{|\im\A|}^2}
	+\frac{|\T|\lrB{\beta_A+1}}{|\im\A|},
	\label{eq:markov}
 \end{align}
where the last inequality comes from (\ref{eq:EACG}) and (\ref{eq:variance}).

Finally, from the fact that $\T\neq\emptyset$, we have
\begin{align*}
 p_{AC}\lrsb{\lrb{(A,\cc):
 \T\cap\C_A(\cc)=\emptyset
 }}
 &=
 p_{AC}\lrsb{\lrb{(A,\cc):
 \forall \uu\in\T,
 A\uu\neq\cc
 }}
 \\
 &=
 p_{AC}\lrsb{\lrb{(A,\cc):
 \sum_{\uu\in\T}
 \chi(A\uu=\cc)=0
 }}
 \\
 &\leq
 p_{AC}\lrsb{\lrb{
 (A,\cc):
 \left|
 \sum_{\uu\in\T}
 \chi(A\uu=\cc)-\frac{|\T|}{|\im\A|}
 \right|
 \geq  \frac{|\T|}{|\im\A|}
 }}
 \\
 &\leq
 \frac{
  \sum_{A,\cc}p_{AC}(A,\cc)
 \lrB{
 \sum_{\uu\in\T}
	\chi(A\uu=\cc)- \frac{|\T|}{|\im\A|}
 }^2
 }
 {\frac{|\T|^2}{|\im\A|^{2}}}
 \\
 &\leq
 \frac{
 \frac{|\T|^2\lrB{\alpha_A-1}}{|\im\A|^2}
 +\frac{|\T|\lrB{\beta_A+1}}{|\im\A|}
 }
 {\frac{|\T|^2}{|\im\A|^{2}}}
 \\
 &=
 \alpha_A-1+\frac{|\im\A|\lrB{\beta_A+1}}{|\T|},
\end{align*}
where the second inequality comes from the Markov inequality
and the third inequality comes from (\ref{eq:markov}).
\hfill\QED

{\it Proof of Lemma \ref{lem:ACnotempty}:}
Since $(\A,p_A)$ has an $(\aalpha_A,\bbeta_A)$-hash property,
we have
\begin{align*}
 p_{AC}\lrsb{\lrb{
 (A,\cc):
 \begin{aligned}
	&\G\cap\C_A(\cc)\neq \emptyset
	\\
	&\uu\in\C_A(\cc)
 \end{aligned}
 }}
 &=
 p_{AC}\lrsb{\lrb{
 (A,\cc):
 \begin{aligned}
	&\G\cap\C_A(A\uu)\neq \emptyset
	\\
	&\uu\in\C_A(\cc)
 \end{aligned}
 }}
 \\
 &
 =
 \sum_A p_A(A)
 \chi(\G\cap\C_A(A\uu)\neq \emptyset)
 \sum_{\cc} p_C(\cc)
 \chi(A\uu=\cc)
 \\
 &
 =
 \frac{
 p_A\lrsb{\lrb{
 A: \G\cap\C_A(A\uu)\neq \emptyset
 }}
 }{|\im\A|}
 \\
 &\leq
 \frac{|\G|\alpha_A}{|\im\A|^2} + \frac{\beta_A}{|\im\A|},
\end{align*}
where the second equality
comes from the fact that
random variables A and C are independent,
the third equality comes from (\ref{eq:EC}),
and the inequality comes form
Lemma~\ref{lem:Anotempty}.
\hfill\QED

{\it Proof of Lemma \ref{lem:ABCnoempty}:}
By applying 
Lemma \ref{lem:Anotempty}
to the set $\lrB{\G\setminus\{\uu_{A,\cc}\}}\cap\C_A(\cc)$, we have
\begin{align*}
 p_{ABC}\lrsb{\lrb{
 (A,B,\cc): 
 \lrB{\G\setminus\{\uu_{A,\cc}\}}\cap\C_{AB}(\cc,B\uu_{A,\cc})\neq\emptyset
 }}
 &\leq
 \sum_{A,\cc}p_{AC}(A,\cc)
 \lrB{
 \frac{\left|\lrB{\G\setminus\{\uu_{A,\cc}\}}\cap\C_A(\cc)\right|\alpha_B}
 {|\im\B|}
 +\beta_B
 }
 \\
 &\leq
 \sum_{A,\cc}p_{AC}(A,\cc)
 \lrB{
 \frac{\left|\G\cap\C_A(\cc)\right|\alpha_B}{|\im\B|}
 +\beta_B
 }
 \\
 &=
 \frac{|\G|\alpha_B}{|\im\A||\im\B|}+\beta_B,
\end{align*}
where the last equality comes from (\ref{eq:EACG}).
\hfill\QED

{\it Proof of Lemma \ref{thm:joint-typical}:}
Let 
\begin{align*}
 \G(\vv)&\equiv\{\uu: \mu_{U|V}(\uu|\vv)> 2^{-n[H(U|V)-2\e]}\}.
\end{align*}
If $\G(\vv)\cap\C_A(\cc)=\emptyset$ and $\T(\vv)\cap\C_A(\cc)\neq\emptyset$,
then there is $\uu\in\T(\vv)\cap\C_A(\cc)$
and $g_{A}(\cc|\vv)$ satisfies
\begin{align*}
 \mu_{U|V}(g_{A}(\cc|\vv)|\vv)\leq 2^{-n[H(U|V)-2\e]}.
\end{align*}
Since $\uu\in\C_A(\cc)$, we have
\begin{align*}
 \mu_{U|V}(g_A(\cc|\vv)|\vv)
 &\geq
 \mu_{U|V}(\uu|\vv).
\end{align*}
This implies that $g_{A}(\cc|\vv)\in\T(\vv)$ from the assumption of $\T(\vv)$.
From Lemma \ref{lem:noempty} and (\ref{eq:largeprob}), we have
\begin{align*}
 p_{AC}\lrsb{\lrb{(A,\cc):
 g_{A}(\cc|\vv)\notin\T(\vv)
 }}
 &\leq
 1-
 p_{AC}\lrsb{\lrb{(A,\cc):
 g_{A}(\cc|\vv)\in\T(\vv)
 }}
 \\
 &\leq
 1-
 p_{AC}\lrsb{\lrb{(A,\cc):
 \begin{aligned}
	&\T(\vv)\cap\C_A(\cc)\neq\emptyset
	\\
	&\G(\vv)\cap\C_A(\cc)=\emptyset
 \end{aligned}
 }}
 \\
 &\leq
 p_{AC}\lrsb{\lrb{(A,\cc):
 \T(\vv)\cap\C_A(\cc)=\emptyset
 }}
 +
 p_{AC}\lrsb{\lrb{(A,\cc):
	\G(\vv)\cap\C_A(\cc)\neq\emptyset
 }}
 \\
 &
 \leq
 \alpha_A-1
 +\frac{|\im\A|\lrB{\beta_A+1}}{|\T(\vv)|}
 +\frac{|\G(\vv)|}{|\im\A|}
 \\
 &
 \leq
 \alpha_A-1
 +\frac{|\im\A|\lrB{\beta_A+1}}{|\T(\vv)|}
 +\frac{2^{-n\e}|\U|^{l_{\A}}}{|\im\A|},
\end{align*}
where the last inequality comes from the fact that
\begin{align*}
 |\G(\vv)|\leq 2^{n[H(U|V)-2\e]}.
\end{align*}
\hfill\QED

{\it Proof of Lemma \ref{lem:md-typical}:}
When $\T(\vv)\cap\C_A(\cc)\neq\emptyset$,
we can always find the member of $\T(\vv)$ by using $\hg_A$.
From Lemma \ref{lem:noempty}, we have
\begin{align*}
 p_{AC}\lrsb{\lrb{(A,\cc):
 \hg_{A}(\cc|\vv)\notin\T(\vv)
 }}
 &\leq
 p_{AC}\lrsb{\lrb{(A,\cc):
 \T(\vv)\cap\C_A(\cc)=\emptyset
 }}
 \\
 &
 \leq
 \alpha_A-1
 +\frac{|\im\A|\lrB{\beta_A+1}}{|\T(\vv)|}.
\end{align*}
\hfill\QED

\subsection{Proof of Theorem~\ref{thm:hash-linA} and Lemmas
	\ref{thm:hash-linApB} and \ref{thm:hash-linAB}}
\label{sec:proof-linear}

For a type $\bt$, let $\C_{\bt}$ be defined as
\begin{gather*}
 \C_{\bt} \equiv \lrb{\uu\in\U^n :\ \bt(\uu)=\bt}.
\end{gather*}
We assume that $p_A\lrsb{\lrb{A: A\uu=\zero}}$
depends on $\uu$ only through the type $\bt(\uu)$.
For a given $\uu\in\C_{\bt}$,
we define 
\begin{align*}
 u_{A,\bt}
 &\equiv u_A\lrsb{\lrb{A: A\uu=\zero}}
 \\
 p_{A,\bt}
 &\equiv p_A\lrsb{\lrb{A: A\uu=\zero}},
\end{align*}
where $u_{A}$ denotes the uniform distribution
on the set of all $l_{\A}\times n$ matrices
and we omit $\uu$ from the left hand side
because the probabilities
$u_A\lrsb{\lrb{A: A\uu=\zero}}$
and $p_A\lrsb{\lrb{A: A\uu=\zero}}$
depend on $\uu\in\C_{\bt}$ only through the type $\bt$.

We use the following lemma in the proof.
\begin{lem}
\begin{align}
 \begin{split}
	\alpha_A(n)
	&=
	|\im\A|\max_{\bt\in \hcH}p_{A,\bt}	
 \end{split}
 \label{eq:alpha2}
 \\
 \beta_A(n)
 &=
 \sum_{\bt\in \cH\setminus\hcH}|\C_{\bt}|p_{A,\bt},
 \label{eq:beta2}
\end{align}
 where $\cH$ is a set of all types of length $n$ except the type
 of the zero vector.
\end{lem}
\begin{proof}
Since we can find
$|\U|^{[n-1]l_{\A}}$ matrices $A$
to satisfy $A\uu=\zero$ for $\uu\in\C_{\bt}$, we have
\begin{align*}
 u_{A,\bt}
 &= 
 \frac {|\U|^{[n-1]l_{\A}}}{|\U|^{nl_{\A}}}
 \\
 &= |\U|^{-l_{\A}}.
\end{align*}
We have
\begin{align*}
 S(p_A,\bt)
 &=
 \sum_{A}p_A(A)
 \sum_{\substack{\uu\in\C_{\bt}\\ A\uu=\zero}}
 1
 \\*
 &=
 \sum_{\uu\in\C_{\bt}}
 \sum_{A:A\uu=\zero}
 p_A(A)
 \\
 &=
 |\C_{\bt}|p_{A,\bt}.
\end{align*}
Similarly, we have
\begin{align*}
 S(u_{A},\bt)
 &=
 |\C_{\bt}|u_{A,\bt}.
\end{align*}
The lemma can be shown immediately from
(\ref{eq:alpha-linear}),
(\ref{eq:beta-linear}),
and the above equalities.
\end{proof}

{\it Proof of Theorem~\ref{thm:hash-linA}:}
Without loss of generality, we can assume that $|\T|\leq |\T'|$.
We have
\begin{align*}
 \sum_{\substack{
 \uu\in\T
 \\
 \uu'\in\T'
 }}
 p_{A}\lrsb{\lrb{A: A\uu=A\uu'}}
 &=
 \sum_{\substack{
 \uu\in\T
 \\
 \uu'\in\T'
 }}
 p_{A}\lrsb{\lrb{A: A[\uu-\uu']=\zero}}
 \\
 &\leq
 \sum_{\uu\in\T\cap\T'}p_{A}\lrsb{\lrb{A: A\zero=\zero}}
 +
 \sum_{\bt\in\cH}
 \sum_{\substack{
 \uu\in\T
 \\
 \uu'\in\T'
 \\
 \bt(\uu-\uu')=\bt
 }}
 p_{A,\bt}
 \\
 &\leq
 \sum_{\uu\in\T\cap\T'}1
 +
 \sum_{\bt\in\hcH}
 \sum_{\substack{
 \uu\in\T
 \\
 \uu'\in\T'
 \\
 \bt(\uu-\uu')=\bt
 }}
 p_{A,\bt}
 +
 \sum_{\bt\in\cH\setminus\hcH}\sum_{\uu\in\T}|\C_{\bt}|p_{A,\bt}
 \\
 &\leq
 |\T\cap\T'|
 +
 \sum_{\bt\in\hcH}
 \sum_{\substack{
 \uu\in\T
 \\
 \uu'\in\T'
 \\
 \bt(\uu-\uu')\in\hcH
 }}
 \frac{\alpha_A(n)}{|\im\A|}
 + |\T|\sum_{\bt\in\cH\setminus\hcH}S(p_A,\bt)
 \\
 &\leq
 |\T\cap\T'|
 +
 \frac{|\T||\T'|\alpha_A(n)}
 {|\im\A|}
 +|\T|\beta_A(n)
 \\
 &=
 |\T\cap\T'|
 +
 \frac{|\T||\T'|\alpha_A(n)}{|\im\A|}
 +\min\{|\T|,|\T'|\}\beta_A(n),
\end{align*}
where the third inequality comes from (\ref{eq:alpha2})
and the last equality comes from the assumption $|\T|\leq|\T'|$.
Since $(\aalpha_A,\bbeta_A)$ satisfies (\ref{eq:alpha}) and
(\ref{eq:beta}),
we have the fact that $(\A,p_A)$ has an $(\aalpha_A,\bbeta_A)$-hash property.
\hfill\QED

{\it Proof of Lemma~\ref{thm:hash-linApB}:}
Without loss of generality, we can assume that
$|\T|\leq |\T'|$ and $\beta_A(n)\leq\beta_B(n)$.
Similar to the proof of Theorem~\ref{thm:hash-linA}, we have
\begin{align*}
 &\sum_{\substack{
 \uu\in\T
 \\
 \uu'\in\T'
 }}
 p_{AB}\lrsb{\lrb{(A,B): (A\uu,B\uu) = (A\uu',B\uu')}}
 \\*
 &=
 \sum_{\uu\in\T\cap\T'}
 p_{A}\lrsb{\lrb{A: A\zero=\zero}}p_{B}\lrsb{\lrb{B: B\zero=\zero}}
 +
 \sum_{\bt\in\cH}
 \sum_{\substack{
 \uu\in\T
 \\
 \uu'\in\T'
 \\
 \bt(\uu-\uu')=\bt
 }}
 p_{A,\bt}p_{B,\bt}
 \\
 &\leq
 |\T\cap\T'|
 +
 \sum_{\bt\in\hcH}
 \sum_{\substack{
 \uu\in\T
 \\
 \uu'\in\T'
 \\
 t(\uu-\uu')\in\hcH
 }}
 \frac{\alpha_A(n)\alpha_B(n)}{|\im\A||\im\B|}
 +
 |\T|\sum_{\bt\in\cH\setminus\hcH}S(p_A,t)
 \\
 &\leq
 |\T\cap\T'|
 +
 \frac{|\T||\T'|\alpha_A(n)\alpha_B(n)}{|\im\A||\im\B|}
 +|\T|\beta_A(n)
 \\
 &=
 |\T\cap\T'|
 +
 \frac{|\T||\T'|\alpha_{AB}(n)}{|\im\A||\im\B|}
 +\min\{|\T|,|\T'|\}\beta_{AB}(n),
\end{align*}
where the first inequality comes from the fact that $p_{B,\bt}\leq 1$.
Since $(\alpha_{AB}(n),\beta_{AB}(n))$ satisfies
(\ref{eq:alpha}) and (\ref{eq:beta}),
$(\A\times\B,p_{AB})$ has an $(\aalpha_{AB},\bbeta_{AB})$-hash property.
\hfill\QED

{\it Proof of Lemma~\ref{thm:hash-linAB}:}
Without loss of generality, we can assume that
$|\T|\leq |\T'|$.
Let $\cH_{\U}$ and $\cH_{\V}$ be
defined similarly to the definition
of $\cH$,
and
$\hcH_{\U}$ and $\hcH_{\V}$ be
defined similarly to the definition
of $\hcH$.
Similar to the proof of Theorem~\ref{thm:hash-linA}, we have
\begin{align*}
 &\sum_{\substack{
 (\uu,\vv)\in\T
 \\
 (\uu',\vv')\in\T'
 }}
 p_{AB}\lrsb{\lrb{
 (A,B): (A\uu,B\vv) = (A\uu',B\vv')
 }}
 \\*
 &=
 \sum_{(\uu,\vv)\in\T\cap\T'}
 p_{A}\lrsb{\lrb{A: A\zero =\zero}}p_{B}\lrsb{\lrb{B: B\zero=\zero}}
 +
 \sum_{\substack{
 \bt_{\U}\in\cH_{\U}
 \\
 \bt_{\V}\in\cH_{\V}
 }}
 \sum_{\substack{
 (\uu,\vv)\in\T
 \\
 (\uu',\vv')\in\T'
 \\
 \bt(\uu-\uu')=\bt_{\U}
 \\
 \bt(\vv-\vv')=\bt_{\V}
 }}
 p_{A,\bt_{\U}}p_{B,\bt_{\V}}
 \\
 &
 \leq
 \sum_{(\uu,\vv)\in\T\cap\T'}1
 +
 \sum_{\substack{
 \bt_{\U}\in\hcH_{\U}
 \\
 \bt_{\V}\in\hcH_{\V}
 }}
 \sum_{\substack{
 (\uu,\vv)\in\T
 \\
 (\uu',\vv')\in\T'
 \\
 \bt(\uu-\uu')=\bt_{\U}
 \\
 \bt(\vv-\vv')=\bt_{\V}
 }}
 \frac{\alpha_A(n)\alpha_B(n)}{|\im\A||\im\B|}
 +
 \sum_{\bt_{\U}\in\hcH_{\U}}
 \sum_{\bt_{\V}\in\cH_{\V}\setminus\hcH_{\V}}
 \sum_{(\uu,\vv)\in\T}
 \sum_{\substack{
 (\uu',\vv')\in\T'
 \\
 \bt(\uu-\uu')=\bt_{\U}
 \\
 \bt(\vv-\vv')=\bt_{\V}
 }}
 \frac{\alpha_A(n)p_{B,\bt_{\V}}}{|\im\A|}
 \\*
 &\quad
 +
 \sum_{\bt_{\U}\in\cH_{\U}\setminus\hcH_{\U}}
 \sum_{\bt_{\V}\in\cH_{\V}}
 \sum_{(\uu,\vv)\in\T}
 \sum_{\substack{
 (\uu',\vv')\in\T'
 \\
 \bt(\uu-\uu')=\bt_{\U}
 \\
 \bt(\vv-\vv')=\bt_{\V}
 }}
 \frac{\alpha_B(n)p_{A,\bt_{\U}}}{|\im\B|}
 +
 \sum_{\bt_{\U}\in\cH_{\U}\setminus\hcH_{\U}}
 \sum_{\bt_{\V}\in\cH_{\V}\setminus\hcH_{\V}}
 \sum_{(\uu,\vv)\in\T}
 \sum_{\substack{
 (\uu',\vv')\in\T'
 \\
 \bt(\uu-\uu')=\bt_{\U}
 \\
 \bt(\vv-\vv')=\bt_{\V}
 }}
 p_{A,\bt_{\U}}
 p_{B,\bt_{\V}}
 \\
 &
 \leq
 \sum_{(\uu,\vv)\in\T\cap\T'}1
 +
 \sum_{\substack{
 \bt_{\U}\in\hcH_{\U}
 \\
 \bt_{\V}\in\hcH_{\V}
 }}
 \sum_{\substack{
 (\uu,\vv)\in\T
 \\
 (\uu',\vv')\in\T'
 \\
 \bt(\uu-\uu')=\bt_{\U}
 \\
 \bt(\vv-\vv')=\bt_{\V}
 }}
 \frac{\alpha_A(n)\alpha_B(n)}{|\im\A||\im\B|}
 +
 \sum_{(\uu,\vv)\in\T}
 \frac{\alpha_A(n)}
 {|\im\A|}
 \sum_{\bt_{\V}\in\cH_{\V}\setminus\hcH_{\V}}
 |\C_{\bt_{\V}}|p_{B,\bt_{\V}}
 \\
 &\quad
 +
 \sum_{(\uu,\vv)\in\T}
 \frac{\alpha_B(n)}
 {|\im\B|}
 \sum_{\bt_{\U}\in\cH_{\U}\setminus\hcH_{\U}}
 |\C_{\bt_{\U}}|
 p_{A,\bt_{\U}}
 +
 \sum_{(\uu,\vv)\in\T}
 \sum_{\bt_{\U}\in\cH_{\U}\setminus\hcH_{\U}}
 |\C_{\bt_{\U}}|
 p_{A,\bt_{\U}}
 \sum_{\bt_{\V}\in\cH_{\V}\setminus\hcH_{\V}}
 |\C_{\bt_{\V}}|
 p_{B,\bt_{\V}}
 \\
 &\leq
 |\T\cap\T'|
 +
 \frac{|\T||\T'|\alpha_A(n)\alpha_B(n)}{|\im\A||\im\B|}
 +|\T|
 \lrB{
 \frac{\alpha_A(n)\beta_B(n)}{|\im\A|}+\frac{\alpha_B(n)\beta_A(n)}{|\im\B|}
 +\beta_A(n)\beta_B(n)
 }
 \\
 &=
 |\T\cap\T'|
 +
 \frac{|\T||\T'|\alpha_{AB}(n)}{|\im\A||\im\B|}
 +\min\{|\T|,|\T'|\}\beta'_{AB}(n),
\end{align*}
where the first inequality comes from the fact that $p_{A,\bt_{\U}}\leq 1$
and  $p_{B,\bt_{\V}}\leq 1$ .
Since $(\aalpha_{AB},\bbeta'_{AB})$ satisfies
(\ref{eq:alpha}) and (\ref{eq:beta}),
$(\A\times\B,p_{AB})$ has an $(\aalpha_{AB},\bbeta'_{AB})$-hash property.
\hfill\QED

\subsection{Proof of Theorem \ref{thm:qMacKay}}
\label{sec:alphabeta}

Throughout this section,
let $\U\equiv\GFq$, $l\equiv nR$,
and $p_A$ be an ensemble of $l\times n$ sparse matrices
as specified in Section~\ref{sec:linear},
where we omit dependence on the ensemble of $l$.
It should be noted that $l\to \infty$ by letting $n\to\infty$.

First, we prepare lemmas that provide the analytic expression
of $p_{A,\bt}$.
\begin{lem}
\label{lem:random-walk}
We consider a random-walk on $\GF(q^l)$ defined as the following.
Let $\cc_n\in\GFql$ be the position after $n$ steps.
At each unit step, the position is renewed in the following rule.
\begin{enumerate}
\item Choose $(i,u)\in\{1,\ldots,l\}\times\GFq$ uniformly at random.
\item Add $u$ to the $i$-th element of $\cc_n$.
\end{enumerate}
Then, the probability $P_n(\cc)$ 
of the position $\cc$ after $n$ steps
starting from the zero vector is described by
\begin{align}
 \begin{split}
 P_n(\cc)
 &=
 \frac 1{q^l}
 \sum_{k=0}^l
 \lrB{1-\frac {qk}{[q-1]l}}^n
 \sum_{k'=0}^{w(\cc)}
 \binom{w(\cc)}{k'}\binom{l-w(\cc)}{k-k'}
 (-1)^{k'}(q-1)^{k-k'}.
	\end{split}
 \label{eq:Pn}
\end{align}
\end{lem}
\begin{proof}
Let $\hcC\subset\GFql$ be defined as
\[
 \hcC\equiv\lrb{
 (\underbrace{0,\ldots,0}_{[j-1]},c,
 \underbrace{0,\ldots,0}_{[l-j]})
 :
 \begin{aligned}
	&j\in\{1,\ldots,l\}
	\\
	&c\in\GFq
 \end{aligned}
 }.
\]
Then the transition rule of this random walk is equivalent to
the following.
\begin{enumerate}
 \item Choose $\hcc\in\hcC$ uniformly at random.
 \item Add $\cc$ to $\cc_n$, that is,
			\[
			 \cc_{n+1}\equiv\cc_n+\hcc.
			\]
\end{enumerate}
We have the following recursion formula for $P_n(\cc)$.
\begin{gather*}
 P_1(\cc)
 =
 \begin{cases}
	\frac 1{[q-1]l},&\text{if}\ \cc\in\hcC,
	\\
	0,&\text{otherwise}.
 \end{cases}
 \\
 P_{n+1}(\cc)
 =\sum_{\cc'\in\GFql}P_n(\cc')P_1(\cc-\cc')
 =[P_n\ast P_1](\cc),
\end{gather*}
where $P_n*P_{1}$ denotes the convolution.
We have (\ref{eq:Pn}) by using the following formulas
\begin{gather*}
 \dft P_n = [\dft P_{n-1}][\dft P_1] = \cdots = [\dft P_1]^n
 \\
 P_n = \dfti \dft P_n = \dfti[[\dft P_1]^n],
\end{gather*}
where $\dft$ is the discrete Fourier transform and $\dfti$ is its inverse.
\end{proof}

\begin{lem}
\label{lem:weight}
The probability
$p_A\lrsb{\lrb{A: A\uu = \zero}}$
depends on $\uu$ only through the type $\bt(\uu)$,
that is,
if $w(\bt)=w(\bt')$ then $p_{A,\bt}=p_{A,\bt'}$.
Furthermore,
 \begin{align*}
 p_{A,\bt}
 &=
 \frac 1{q^{l}}
 \sum_{k=0}^{l}
 \lrB{1-\frac {qk}{[q-1]l}}^{w(\bt)\tau}
 \binom{l}{k}
 (q-1)^{k}.
 \end{align*}
\end{lem}
\begin{proof}
For $\uu\equiv(u_1,\ldots,u_n)$,
we define $\uu^*\equiv(u^*_1,\ldots,u^*_n)$ as
\[
 u^*_i
 \equiv
 \begin{cases}
	1,&\text{if}\ u_i\neq 0
	\\
	0,&\text{if}\ u_i= 0.
 \end{cases}
\]
Similarly as in the proof of \cite[Lemma 1]{EM05},
we can prove that
two sets $\{A: A\uu=\zero\}$ 
and $\{A:A\uu^*=\zero\}$ are in one-to-one correspondence.
Then we have
\[
 p_A\lrsb{\lrb{A: A\uu=\zero}}
 =p_A\lrsb{\lrb{A: A\uu^*=\zero}},
\]
that is,
$p_A\lrsb{\lrb{A: A\uu=\zero}}$
depends on $\uu$ only through $w(\bt)$.

Since $p_A\lrsb{\lrb{A: A\uu^*=\zero}}$ is equal to the probability
that the position 
of the random walk defined in Lemma \ref{lem:random-walk}
starts from the zero vector and returns to the zero vector
after $w(\uu^*)\tau$ steps, we have
\begin{align*}
 p_{A,\bt}
 &=
 P_{w(\bt)\tau}(\zero)
 \\
 &=
 \frac 1{q^{l}}
 \sum_{k=0}^{l}
 \lrB{1-\frac {qk}{[q-1]l}}^{w(\bt)\tau}
 \binom{l}{k}
 (q-1)^{k}.
 \end{align*}
\end{proof}

Next, we prove the following lemma.
\begin{lem}
If the column weight $\tau$ is even, then
\[
 \im\A
 =
 \begin{cases}
	\{\uu\in\U^l: w(\uu)\ \text{is even}\},&\text{if}\ q=2
	\\
	\U^l,&\text{if}\ q>2,
 \end{cases}
\]
 which implies
\[
 \frac{|\im\A|}{|\U|^{l}}
 =
 \begin{cases}
	2,&\text{if}\ q=2
	\\
	1,&\text{if}\ q>2.
 \end{cases}
\]
 \label{lem:imApUl}
\end{lem}
\begin{proof}
 Let $a_{i,j}$ be the $(i,j)$ element of $A$.

 First, we assume that $q=2$.
 Then it is sufficient to prove that
 $w(A\uu)$ is even for any possible $A$ and $\uu\in\U^l$
 because 
 \begin{align*}
	\sum_{\substack{
	\cc:\\
	w(\cc)\ \text{is even}
	}}
	1
	-
	\sum_{\substack{
	\cc:\\
	w(\cc)\ \text{is odd}
	}}
	1
	&=
	\sum_{w=0}^l
	\binom{n}{w}[-1]^w
	\\
	&=0
 \end{align*}
 which implies that $|\im\A|=|\U|^l/2$.
 Without loss of generality, we can assume that
 $w(\uu)=w$ and $\uu=(1,\ldots,1,0,\ldots,0)$.
 Let $\ba_i\equiv (a_{i,1},\ldots,a_{i,w(\uu)})$.
 Since every column vecotor has an even weight,
 we have the fact that
 $\sum_{i=1}^{w(\uu)}w(\ba_i)$ is even.
 In addition, we have
 \begin{align*}
	\sum_{\substack{
	i:
	\\
	w(\ba_i)\ \text{is odd}
	}}
	w(\ba_i)
	&=
	\sum_{i=1}^{w(\uu)}
	w(\ba_i)
	-
	\sum_{\substack{
	i:
	\\
	w(\ba_i)\ \text{is even}
	}}
	w(\ba_i).
 \end{align*}
 This implies that the number of odd-weight vectors $\ba_i$
 is even because
 the right hand side of the above equality is even.
 Since $w(A\uu)$ is a number of odd-weight vectors $\ba_i$,
 we have the fact that $w(A\uu)$ is even for any $A$ and $\uu\in\U^l$.

 Next, we assume that $q>2$.
 It is sufficient to prove that,
 for any $\cc=(c_1,\ldots,c_l)\in\U^l$,
 there is $A$ generated by the scheme and $\uu\in\U^l$
 such that $A\uu=\cc$.
 This fact implies that $\im\A=\U^l$.
 Let $\uu=(1,\ldots,1)$.
 It is possible to generate $A$ satisfying
 \[
 a_{i,j}=
 \begin{cases}
	2a,&\text{if}\ i=j
	\\
	0,&\text{if}\ i\neq j,
 \end{cases}
 \]
 where $a\in\GFq$ is arbitrary.
 Since $q>3$, we have $A\uu=\cc$ by letting $a\equiv c_i/2$.
\end{proof}

Finally we prove that
$(\aalpha_A,\bbeta_A)$ satisfies (\ref{eq:alpha}) and (\ref{eq:beta}).
We define the function $h$ as
\[
 h(\theta)\equiv -\theta\log_e(\theta)-[1-\theta]\log_e(1-\theta),
\]
where $e$ is the base of the natural logarithm.
We use the following lemmas to derive the asymptotic behavior
of $(\aalpha_A,\bbeta_A)$.

\begin{lem}
\label{lem:theta}
Let $a$ be a real number. Then
\[
	\max_{0\leq\theta\leq 1}\lrB{h(\theta)+a\theta}
 \leq \log_e\lrsb{1+e^a}.
\]

If $a\leq -\log_e(l-1)$, then
 \[
	\max_{1/l\leq\theta\leq 1}\lrB{h(\theta)+a\theta}
 \leq h\lrsb{\frac 1l}+\frac a l.
 \]
\end{lem}

\begin{lem}
 \label{lem:h}
 \[
	lh\lrsb{\frac 1l}\leq 1+\log_e l.
 \]
\end{lem}

\begin{lem}
\label{lem:k2}
 \begin{equation*}
	\sum_{k=1}^{l-1}
	\left|1-\frac {2k}{l}\right|^{w\tau}
	\binom{l}{k}
	\leq
	2
	\sum_{k=1}^{\lrfloor{\frac l2}}
	\exp\lrsb{-\frac {2kw\tau}{l}}
	\binom{l}{k}
 \end{equation*}
\end{lem}
\begin{proof}
Since
\[
\left|1-\frac{2k}l\right|^{w\tau}\binom{l}{k}
=
\left|1-\frac{2[l-k]}l\right|^{w\tau}\binom{l}{1-k},
\]
then we have
\begin{align*}
	\sum_{k=1}^{l-1}
	\left|1-\frac {2k}{l}\right|^{w\tau}
	\binom{l}{k}
 &=
	2
	\sum_{k=1}^{\lrfloor{\frac l2}}
	\left|1-\frac {2k}{l}\right|^{w\tau}
	\binom{l}{k}
 \\
 &\leq
 2
 \sum_{k=1}^{\lrfloor{\frac l2}}
 \exp\lrsb{-\frac {2kw\tau}{l}}
 \binom{l}{k},
\end{align*}
where the inequality comes from the fact that $2k/l\leq 1$.
\end{proof}

\begin{lem}
\label{lem:k}
 \begin{equation}
	\sum_{k=1}^{l}
	\left|1-\frac {qk}{[q-1]l}\right|^{w\tau}
	\binom{l}{k}
	[q-1]^{k}
	\leq
	\sum_{k=1}^{\lrfloor{\frac{[q-1]l}q}}
	\exp\lrsb{-\frac {qkw\tau}{[q-1]l}}
	\binom{l}{k}
	[q-1]^{k}
	+
	\sum_{k=\lrceil{\frac{[q-1]l}q}}^l
	\binom{l}{k}[q-1]^{k-w\tau}.
	\label{eq:k}
 \end{equation}
\end{lem}
\begin{proof}
We can show the lemma from the fact that
\[
 \frac{qk}{[q-1]l}\leq 1
\]
when $k\leq [q-1]l/q$
and
\begin{align*}
 \left|1-\frac{qk}{[q-1]l}\right|
 &=
 \frac{q[k-l]+l}{[q-1]l}
 \\
 &\leq
 \frac{l}{[q-1]l}
 \\
 &=
 [q-1]^{-1}
\end{align*}
when $[q-1]l/q<k\leq l$.
\end{proof}

Let $\tau$ be the parameter given in the procedure used for generating
a sparse matrix.
We assume that $\tau$ and  $\xi$ satisfy
 \begin{gather}
	\tau \equiv 2\lrceil{\log_e \frac{l^2}R}
	\label{eq:tau}
	\\
	\frac{h\lrsb{\xi R}}R+\xi \log_e(q-1)< \frac 13.
	\label{eq:eta}
 \end{gather}
Then we have
 \begin{align}
	\xi\tau&\geq 3\log_e l
	\label{eq:eta-tau}
 \end{align}
for all sufficiently large $l$.

Now we are in position to prove the following two lemmas
which provides the proof of Theorem \ref{thm:qMacKay}.

\begin{lem}
 \[
 \limn\alpha_A(n)=1.
 \]
 \label{lem:alpha-lim}
\end{lem}
\begin{proof}
In the following, we first show that
\begin{equation}
 \lim_{l\to\infty}\sum_{k=1}^{\lrfloor{\frac{[q-1]l}q}}
 \exp\lrsb{-\frac {qkw\tau}{[q-1]l}}\binom{l}{k}[q-1]^{k}
 =0
 \label{eq:pw-small}
\end{equation}
for all $q\geq 2$ and $w>\xi l$.
By assuming $w>\xi l$,
we have
\begin{align*}
 \sum_{k=1}^{\lrfloor{\frac{[q-1]l}q}}
 \exp\lrsb{-\frac {qkw\tau}{[q-1]l}}
 \binom{l}{k}[q-1]^{k}
 &\leq
 l
 \max_{1/l\leq\theta\leq 1}
 \exp\lrsb{-w\tau\theta}
 \exp\lrsb{lh(\theta)}
 [q-1]^{l\theta}
 \notag
 \\*
 &\leq
 l
 \max_{1/l\leq\theta\leq 1}
 \exp\lrsb{
 -\xi l\tau\theta+lh(\theta)
 +l\theta\log_e(q-1)
 }
 \notag
 \\*
 &\leq
 l
 \max_{1/l\leq\theta\leq 1}
 \exp\lrsb{
 l
 \lrB{
 h(\theta)
 +\lrB{
 \log_e(q-1)-\xi\tau
 }\theta
 }}
 \notag
 \\
 &\leq
 l
 \exp\lrsb{
 l
 \lrB{
 h\lrsb{\frac 1l}+\frac{\log_e(q-1)-\xi\tau}l
 }}
 \notag
 \\
 &\leq
 \exp\lrsb{
 1+\log_e l+\log_e(q-1)-\xi\tau+\log_e l
 }
 \notag
 \\
 &\leq
 \exp\lrsb{
 -\xi\tau+2\log_e l+\log_e [q-1]e
 },
\end{align*}
where the fifth inequality comes from 
(\ref{eq:eta-tau}) and Lemma \ref{lem:theta},
and the sixth inequality comes from Lemma \ref{lem:h}.
Hence we have (\ref{eq:pw-small})
for all $q\geq 2$ and $w>\xi l$.

Next, we show the lemma by assumimg that $q=2$.
From Lemma \ref{lem:k2}, (\ref{eq:pw-small}),
and the fact that $w\tau$ is even, we have
\begin{align*}
 \lim_{l\to\infty}
 \max_{w> \xi l}
 \sum_{k=1}^{l-1}
 \lrB{1-\frac {2k}{l}}^{w\tau}
 \binom{l}{k}
 &
 =
 \lim_{l\to\infty}
 \max_{w> \xi l}
 \sum_{k=1}^{l-1}
 \left|1-\frac {2k}{l}\right|^{w\tau}
 \binom{l}{k}
 \\
 &
 \leq
 2\lim_{l\to\infty}
 \max_{w> \xi l}
 \sum_{k=1}^{\lrfloor{\frac l2}}
 \exp\lrsb{-\frac {2kw\tau}{l}}
 \binom{l}{k}
 \\
 &=0.
\end{align*}
From (\ref{eq:alpha2}) and Lemma \ref{lem:imApUl}, we have
\begin{align*}
 \limn\alpha_A(n)
 &=
 \limn\frac 12\max_{w>l\xi}
 \sum_{k=0}^{l}
 \lrB{1-\frac {2k}{l}}^{w\tau}
 \binom{l}{k}
 \\
 &=
 1+
 \frac 12\limn\max_{w> l\xi}
 \sum_{k=1}^{l-1}
 \lrB{1-\frac {2k}{l}}^{w\tau}
 \binom{l}{k}
 \\
 &=
 1.
\end{align*}

Finally, we show the lemma by assuming that $q>2$.
From (\ref{eq:pw-small}),
the first term on the right hand side of (\ref{eq:k})
vanishes by letting $l\to\infty$.
Since $[q-1]l/q\geq 1/2$, then
the second term on the right hand side of (\ref{eq:k})
is evaluated by
\begin{align*}
 \sum_{k=\lrceil{\frac{[q-1]l}q}}^l
 \binom{l}{k}[q-1]^{k-w\tau}
 &\leq
 l\binom{l}{\lrceil{[q-1]l/q}}[q-1]^{l-w\tau}
 \\
 &\leq
 l\exp\lrsb{lh\lrsb{\frac{q-1}q}}[q-1]^{l-w\tau}
 \\
 &\leq
 l\exp\lrsb{l\log_e eq -w\tau\log_e(q-1)}
  \\
 &<
 \exp\lrsb{-l[\xi\tau\log_e(q-1)-\log_e eq-\log_e l]},
\end{align*}
where the third inequality comes from $h(\theta)\leq 1$.
From $q>2$ and (\ref{eq:eta-tau}), the second term on the right hand side of
(\ref{eq:k}) vanishes by letting $l\to\infty$.
From the above two observations
and the fact that $w\tau$ is even, we have
\begin{align*}
 \lim_{l\to\infty}
	\max_{w> \xi l}
	\sum_{k=1}^{l}
	\lrB{1-\frac {qk}{[q-1]l}}^{w\tau}
	\binom{l}{k}
 [q-1]^{k}
 &
 =
 \lim_{l\to\infty}
	\max_{w> \xi l}
	\sum_{k=1}^{l}
 \left|1-\frac {qk}{[q-1]l}\right|^{w\tau}
	\binom{l}{k}
 [q-1]^{k}
 \\
 &=0.
 \label{eq:alpha-max}
\end{align*}
From (\ref{eq:alpha2}) and Lemma \ref{lem:imApUl}, we have
\begin{align*}
 \limn\alpha_A(n)
 &=
 \limn\max_{w>l\xi}
 \sum_{k=0}^{l}
 \lrB{1-\frac {qk}{[q-1]l}}^{w\tau}
 \binom{l}{k}
 [q-1]^{k}
 \\
 &=
 1+
 \limn\max_{w> l\xi}
 \sum_{k=1}^{l}
 \lrB{1-\frac {qk}{[q-1]l}}^{w\tau}
 \binom{l}{k}
 [q-1]^{k}
 \\
 &=
 1.
\end{align*}
\end{proof}

\begin{lem}
\[
\limn \beta_A(n)=0.
\]
 \label{lem:beta-lim}
\end{lem}
\begin{proof}
Let $\C_w\equiv\{\xx: w(\xx)=w\}$. Then we have
\begin{align}
 |\C_w|
 &=\binom{n}{w}[q-1]^{w}
 \label{eq:cteq}
 \\
 &\leq
 \exp\lrsb{nh\lrsb{\frac {w}n}+w\log_e (q-1)}.
 \label{eq:ct}
\end{align}

In the following, we first show that
\begin{equation}
 \lim_{l\to\infty}
	\sum_{w=1}^{\xi l}
	\frac {|\C_w|}{q^{l}}
 \sum_{k=0}^{l}
 \exp\lrsb{-\frac {qkw\tau}{[q-1]l}}
 \binom{l}{k}
 [q-1]^{k}
 =0.
\label{eq:cwpwlim}
\end{equation}
We have
\begin{align}
 &\sum_{w=1}^{\xi l}
 \frac {|\C_w|}{q^{l}}
 \sum_{k=0}^{l}
 \exp\lrsb{-\frac {qkw\tau}{[q-1]l}}
 \binom{l}{k}
 [q-1]^{k}
 \notag
 \\*
 &=
 \sum_{w=1}^{\xi l}
 \frac {|\C_w|}{q^{l}}
 \lrB{1+[q-1]\exp\lrsb{-\frac{qw\tau}{[q-1]l}}}^l
 \notag
 \\
 &=
 \sum_{w=1}^{\lrfloor{\frac{l}{2\tau}}}
 |\C_w|
 \lrB{\frac{1+[q-1]\exp\lrsb{-\frac{qw\tau}{[q-1]l}}}q}^l
 +
 \sum_{w=\lrceil{\frac{l}{2\tau}}}^{\xi l}
 |\C_w|
 \lrB{\frac{1+[q-1]\exp\lrsb{-\frac{qw\tau}{[q-1]l}}}q}^l.
 \label{eq:cwpw0}
\end{align}
The first term on the right hand side of (\ref{eq:cwpw0}) is evaluated
by
\begin{align}
 \sum_{w=1}^{\lrfloor{\frac{l}{2\tau}}}
 |\C_w|
 \lrB{\frac{1+[q-1]\exp\lrsb{-\frac{qw\tau}{[q-1]l}}}q}^l
 &\leq
 \sum_{w=1}^{\lrfloor{\frac{l}{2\tau}}}
 |\C_w|
 \lrB{\frac{1+[q-1]\lrB{1-\frac{qw\tau}{2[q-1]l}}}q}^l
	\notag
 \\*
 &=
 \sum_{w=1}^{\lrfloor{\frac{l}{2\tau}}}
 \binom{n}{w}[q-1]^{w}
 \lrB{1-\frac{w\tau}{2l}}^l
	\notag
 \\
 &\leq
 \sum_{w=1}^{\lrfloor{\frac{l}{2\tau}}}
 n^{w}q^{w}
 \exp\lrsb{-\frac{w\tau}2}
	\notag
 \\
 &\leq
 \sum_{w=1}^{\lrfloor{\frac{l}{2\tau}}}
 nq\exp\lrsb{-\frac{\tau}2}
	\notag
 \\
 &\leq
 \frac{nql}{2\tau}\exp\lrsb{-\frac{\tau}2}
	\notag
 \\
 &=
 \frac{q}{2\tau}\exp\lrsb{\log_e\frac{l^{2}}R-\frac{\tau}2}
	\notag
 \\
 &\leq
 \frac{q}{4\log_e\frac{l^2}R}.
 \label{eq:beta11}
\end{align}
The first inequality comes from the fact that
$\exp(-x)\leq 1-x/2$ for $0\leq x\leq 1/2$.
The first equality comes from (\ref{eq:cteq}).
The second inequality comes from the fact that $[1+x]^l\leq\exp(lx)$.
The third inequality comes from the fact that
$n^{w}q^{w}\exp\lrsb{-\frac{w\tau}2}$ is a non-increasing function
of $w$. The fifth inequality comes from (\ref{eq:tau}).
From (\ref{eq:beta11}),
the first term on the right hand side of (\ref{eq:cwpw0}) vanishes
by letting $l\to\infty$.
The second term of (\ref{eq:cwpw0}) is evaluated by
\begin{align}
 \sum_{w=\lrfloor{\frac{l}{2\tau}}}^{\xi l}
 |\C_w|
 \lrB{\frac{1+[q-1]\exp\lrsb{-\frac{qw\tau}{[q-1]l}}}q}^l
 &\leq
 \sum_{w=\lrfloor{\frac{l}{2\tau}}}^{\xi l}
 |\C_w|
 \lrB{\frac{1+[q-1]\exp\lrsb{-\frac{q}{2[q-1]}}}q}^l
 \notag
 \\
 &\leq
 \sum_{w=\lrfloor{\frac{l}{2\tau}}}^{\xi l}
 |\C_w|
 \exp\lrsb{-\frac l3}
 \notag
 \\
 &\leq
 \xi l\exp\lrsb{
 l\lrB{\frac{h\lrsb{\xi R}}R+\xi\log_e (q-1)-\frac 13}
 },
 \label{eq:beta12}
\end{align}
where the second inequality comes from the fact that
\[
 \frac{1+[q-1]\exp\lrsb{-\frac{q}{2[q-1]}}}q
 \leq e^{-\frac 13}
\]
and the third inequality comes from (\ref{eq:ct}).
From (\ref{eq:eta}) and (\ref{eq:beta12}),
the second term on the right hand side of (\ref{eq:cwpw0})
vanishes by letting $l\to\infty$.
From the above two observations,
we have (\ref{eq:cwpwlim}).

Next, we show the lemma by assuming that $q=2$.
From (\ref{eq:beta2}), the fact that $w\tau$ is even,
and Lemma \ref{lem:k2}, we have
\begin{align*}
 \beta_A(n)
 &=
 \sum_{w=1}^{\xi l}
 \frac {|\C_w|}{2^{l}}
 \sum_{k=0}^{l}
 \lrB{1-\frac {2k}{l}}^{w\tau}
 \binom{l}{k}
 \notag
 \\
 &\leq
 2\sum_{w=1}^{\xi l}
 \frac {|\C_w|}{2^{l}}
 \sum_{k=0}^{\lrfloor{\frac l2}}
 \exp\lrsb{-\frac {2kw\tau}{l}}
 \binom{l}{k}
 \notag
 \\
 &\leq
 2\sum_{w=1}^{\xi l}
 \frac {|\C_w|}{2^{l}}
 \sum_{k=0}^{l}
 \exp\lrsb{-\frac {2kw\tau}{l}}
 \binom{l}{k}
\end{align*}
From (\ref{eq:cwpwlim}), we have the lemma for $q=2$

Finally, we show the lemma by assuming that $q>2$.
From (\ref{eq:beta2}), and Lemmas \ref{lem:weight} and \ref{lem:k}, we have
\begin{align}
 \beta_A(n)
 &=
 \sum_{w=1}^{\xi l}
 \frac {|\C_w|}{q^{l}}
 \sum_{k=0}^{l}
 \lrB{1-\frac {qk}{[q-1]l}}^{w\tau}
 \binom{l}{k}
 [q-1]^{k}
 \notag
 \\
 &\leq
 \sum_{w=1}^{\xi l}
 \frac {|\C_w|}{q^{l}}
 \sum_{k=0}^{l}
 \exp\lrsb{-\frac {qkw\tau}{[q-1]l}}
 \binom{l}{k}
 [q-1]^{k}
 +
 \sum_{w=1}^{\xi l}
 \frac {|\C_w|}{q^{l}}
 \sum_{k=0}^{l}
 \binom{l}{k}
 [q-1]^{k-w\tau}
 \notag
 \\
 &=
 \sum_{w=1}^{\xi l}
 \frac {|\C_w|}{q^{l}}
 \sum_{k=0}^{l}
 \exp\lrsb{-\frac {qkw\tau}{[q-1]l}}
 \binom{l}{k}
 [q-1]^{k}
 +
 \sum_{w=1}^{\xi l}
 |\C_w|[q-1]^{-w\tau}.
 \label{eq:beta1}
\end{align}
From (\ref{eq:cwpwlim}),
the first term on the right hand side of
(\ref{eq:beta1}) vanishes by letting $l\to\infty$.
From (\ref{eq:ct}), the second term on the right hand side of
(\ref{eq:beta1}) is evaluated by
\begin{align*}
 \sum_{w=1}^{\xi l}
 |\C_w|[q-1]^{-w\tau}
 &\leq
 \sum_{w=1}^{\xi l}
 \exp\lrsb{nh\lrsb{\frac{w}n}+w[1-\tau]\log_e(q-1)}
 \\
 &\leq
 \xi l
 \exp\lrsb{
 n\max_{1/n\leq \theta\leq 1}
 \lrB{h(\theta)+n[1-\tau]\log_e (q-1)\theta}
 }
 \\
 &\leq
 \xi l
 \exp\lrsb{
 nh\lrsb{\frac 1n}+[1-\tau]\log_e(q-1)
 }
 \\
 &\leq
 \exp\lrsb{
 1+\log_e n+[1-\tau]\log_e(q-1)+\log_e \xi l,
 }
\end{align*}
where the third inequality comes from Lemma \ref{lem:theta} and
the fact that
\[
 [1-\tau]\log_e(q-1)<-\log_e(n-1)
\]
for all sufficiently large $n$ and $q>2$.
The fourth inequality comes from Lemma \ref{lem:h}.
From (\ref{eq:tau}), we have
\[
 1+\log_e n+[1-\tau]\log_e(q-1)+\log_e \xi l\to -\infty
\]
by letting  $n\to\infty$.
Then the third term on the right hand side of (\ref{eq:beta1}) vanishes
by letting $n\to\infty$.
Hence we have the lemma for $q>2$.
\end{proof}

\subsection{Proof of Theorem \ref{thm:sw}}
\label{sec:proof-sw}

We define the set $\T$ as
\[
 \T\equiv\lrb{(\xx,\yy):
 \begin{aligned}
	&-\frac 1n\log\mu_{X|Y}(\xx|\yy)\leq H(X|Y)+\gamma
	\\
	&-\frac 1n\log\mu_{Y|X}(\yy|\xx)\leq H(Y|X)+\gamma
	\\
	&-\frac 1n\log\mu_{XY}(\xx,\yy)\leq H(XY)+\gamma
 \end{aligned}
 }.
\]
It should be noted that
the above definition can be replaced by that defined in \cite{SWLDPC}.
This implies that the theorem is valid for general correlated sources.

Let $(\xx,\yy)$ be the output of correlated sources.
We define
\begin{align}
 &\bullet (\xx,\yy)\notin\T
 \tag{SW1}
 \\
 &\bullet \exists \xx'\neq\xx\ \text{s.t.}
 \ \xx'\in\C_A(A\xx),
 \mu_{X,Y}(\xx',\yy)\geq\mu_{X,Y}(\xx,\yy)
 \tag{SW2}
 \\
 &\bullet \exists \yy'\neq\yy\ \text{s.t.}
 \ \yy'\in\C_B(B\yy),
 \mu_{X,Y}(\xx,\yy')\geq\mu_{X,Y}(\xx,\yy)
 \tag{SW3}
 \\
 &\bullet \exists (\xx',\yy')\neq(\xx,\yy)\ \text{s.t.}
 \ \xx'\in\C_A(A\xx),\ \yy'\in\C_B(B\yy),
 \mu_{X,Y}(\xx',\yy')\geq\mu_{X,Y}(\xx,\yy).
 \tag{SW4}
\end{align}
Since a decoding error occurs when
at least one of the conditions (SW1)--(SW4) is satisfied,
the error probability is upper bounded by
\begin{align}
 \Error_{XY}(A,B)
 &\leq
 \mu_{XY}(\E_1^c)+\mu_{XY}(\E_1^c\cap\E_2)
 +\mu_{XY}(\E_1^c\cap\E_3)+\mu_{XY}(\E_1^c\cap\E_4),
 \label{eq:sw0}
\end{align}
where we define
\[
 \E_i\equiv\{(\xx,\yy): \text{(SW$i$)}\}.
\]

First, we evaluate $E_{AB}\lrB{\mu_{XY}(\E_1)}$.
From Lemma \ref{lem:typical-prob}, we have
\begin{equation}
 E_{AB}\lrB{\mu_{XY}(\E_1)}
	\leq
	\frac{\delta}4
	\label{eq:sw1}
\end{equation}
for all sufficiently large $n$.

Next, we evaluate $E_{AB}\lrB{\mu_{XY}(\E_1^c\cap\E_2)}$
and $E_{AB}\lrB{\mu_{XY}(\E_1^c\cap\E_3)}$.
Since
\[
\mu_{X|Y}(\xx'|\yy)\geq\mu_{X|Y}(\xx|\yy)\geq 2^{-n[H(X|Y)+\gamma]}
\]
When (SW1) and (SW2),
we have
\[
\lrB{\G(\yy)\setminus\{\xx\}}\cap\C_{A}(A\xx)\neq\emptyset,
\]
where
\[
 \G(\yy)\equiv\lrb{
 \xx: \mu_{X|Y}(\xx|\yy)\geq 2^{-n[H(X|Y)+\gamma]}
 }.
\]
From Lemma \ref{lem:Anotempty}, we have
\begin{align}
 E_{AB}\lrB{\mu_{XY}(\E_1^c\cap\E_2)}
 &=
 \sum_{(\xx,\yy)\in\T}
 \mu_{XY}(\xx,\yy)
 p_{AB}\lrsb{\lrb{
 (A,B): \text{(SW2)}
 }}
 \notag
 \\
 &\leq
 \sum_{(\xx,\yy)\in\T}
 \mu_{XY}(\xx,\yy)
 p_{A}\lrsb{\lrb{
 A: \lrB{\G(\yy)\setminus\{\xx\}}\cap\C_{A}(A\xx)\neq\emptyset
 }}
 \notag
 \\
 &\leq
 \sum_{(\xx,\yy)\in\T}
 \mu_{XY}(\xx,\yy)
 \lrB{
 \frac{|\G(\yy)|\alpha_A}{|\im\A|}
 +\beta_A
 }
 \notag
 \\
 &\leq
 \sum_{(\xx,\yy)\in\T}
 \mu_{XY}(\xx,\yy)
 \lrB{
 \frac{2^{n[H(X|Y)+\gamma]}\alpha_A}{|\im\A|}
 +\beta_A
 }
 \notag
 \\
 &\leq
 2^{n[H(X|Y)+\gamma]}|\X|^{-l_{\A}}
 \frac{|\X|^{l_{\A}}\alpha_A}{|\im\A|}
 +\beta_A
 \notag
 \\
 &\leq\frac{\delta}4
 \label{eq:sw2}
\end{align}
for all sufficiently large $n$ by taking an appropriate $\gamma>0$,
where the last inequality comes from (\ref{eq:swx})
and an $(\aalpha_A,\bbeta_A)$-hash property of $(\A,p_A)$.
Similarly, we have
\begin{align}
 E_{AB}\lrB{\mu_{XY}(\E_1^c\cap\E_3)}
 &
 \leq
 2^{n[H(Y|X)+\gamma]}|\Y|^{-l_{\B}}
 \frac{|\Y|^{l_{\B}}\alpha_B}{|\im\B|}
 +\beta_B
 \notag
 \\
 &\leq
 \frac {\delta}4
 \label{eq:sw3}
\end{align}
for all sufficiently large $n$ by taking an appropriate $\gamma>0$,
where the last inequality comes from (\ref{eq:swy})
and an $(\aalpha_B,\bbeta_B)$-hash property of $(\B,p_B)$.

Next, we evaluate $E_{AB}\lrB{\mu_{XY}(\E_1^c\cap\E_4)}$.
When (SW1) and (SW4),
we have
\begin{align*}
 &\lrB{\G\setminus\{(\xx,\yy)\}}\cap\lrB{\C_{A}(A\xx)\times\C_B(B\yy)}
 \neq\emptyset,
\end{align*}
where
\begin{align*}
 \G&\equiv\lrb{
 (\xx,\yy): \mu_{XY}(\xx,\yy)\geq 2^{-n[H(XY)+\gamma]}
 }.
\end{align*}
Applying Lemma \ref{lem:Anotempty}
to the joint ensemble $(\A\times\B,p_{AB})$
of a set of functions
$AB:\X^n\times\Y^n\to\X^{l_{\A}}\times\Y^{l_{\B}}$,
we have
\begin{align}
 E_{AB}\lrB{\mu_{XY}(\E_1^c\cap\E_4)}
 &=
 \sum_{(\xx,\yy)\in\T}\mu_{XY}(\xx,\yy)
 p_{AB}\lrsb{\lrb{
 (A,B): \text{(SW4)}
 }}
 \notag
 \\
 &\leq
 \sum_{(\xx,\yy)\in\T}\mu_{XY}(\xx,\yy)
 p_{AB}\lrsb{\lrb{
 (A,B):
 \lrB{\G\setminus\{(\xx,\yy)\}}
 \cap\lrB{\C_{A}(A\xx)\times\C_B(B\yy)}\neq\emptyset
 }}
 \notag
 \\
 &\leq
 \sum_{(\xx,\yy)\in\T}\mu_{XY}(\xx,\yy)
 \lrB{
 \frac{|\G|\alpha_{AB}}{|\im\A||\im\B|}
 +\beta'_{AB}
 }
 \notag
 \\
 &\leq
 2^{n[H(X,Y)+\gamma]}|\X|^{-l_{\A}}|\Y|^{-l_{\B}}
 \frac{|\X|^{l_{\A}}|\Y|^{l_{\B}}\alpha_{AB}}{|\im\A||\im\B|}
 +\beta'_{AB}
 \notag
 \\
 &\leq
 \frac{\delta}4
 \label{eq:sw4}
\end{align}
for all sufficiently large $n$ by taking an appropriate $\gamma>0$,
where the last inequality comes from (\ref{eq:swxy})
and an $(\aalpha_{AB},\bbeta_{AB})$-hash property of
$(\A\times\B,p_A\times p_B)$.

Finally, from (\ref{eq:sw0})--(\ref{eq:sw4}),
for all $\delta>0$ and for all sufficiently large $n$
there are $A$ and $B$ such that
\[
\Error_{XY}(A,B)<\delta.
\]
\hfill\QED

\subsection{Proof of Theorem \ref{thm:gp}}
\label{sec:proof-gp}

For $\bbeta_A$ satisfying $\limn\beta_A(n)=0$,
let $\kkappa\equiv\{\kappa(n)\}_{n=1}^{\infty}$ be a sequence satisfying
\begin{gather}
 \limn\kappa(n)=\infty
 \label{eq:gp-k1}
 \\
 \limn \kappa(n)\beta_A(n)=0
 \label{eq:gp-k2}
 \\
 \limn\frac{\log\kappa(n)}n=0.
 \label{eq:gp-k3}
\end{gather}
For example, there is such  $\kkappa$ by letting
\begin{equation*}
 \kappa(n)\equiv
	\begin{cases}
	 n^{\xi}
	 &\text{if}\ \beta_A(n)=o\lrsb{n^{-\xi}}
	 \\
	 \frac 1{\sqrt{\beta_A(n)}},
	 &\text{otherwise}
	\end{cases}
\end{equation*}
for every $n$.
If $\beta_A(n)$ is not $o\lrsb{n^{-\xi}}$,
there is $\kappa'>0$ such that
$\beta_A(n)n^{\xi}>\kappa'$
 and
\begin{align*}
 \frac{\log\kappa(n)}n
 &=
 \frac{\log\frac 1{\beta_A(n)}}{2n}
 \\
 &\leq \frac{\log\frac{n^{\xi}}{\kappa'}}{2n}
 \\
 &=\frac{\xi\log n-\log\kappa'}{2n}
\end{align*}
for all sufficiently large $n$.
This implies that $\kkappa$ satisfies (\ref{eq:gp-k3}).
In the following, $\kappa$ denotes $\kappa(n)$.

Let $\e\equiv\eB-\eA$.
Then, from (\ref{eq:gp-k3}),
there are $\gamma$ and $\gamma'$ such that
\begin{gather}
 0<\gamma<\sqrt{2\gamma}\log|\Z|<\e
 \label{eq:gp-gamma}
 \\
 0<\gamma'\leq 2\e
 \label{eq:gp-gammap}
 \\
 \eta_{\W|\Z}(\gamma'|\gamma)\leq\e-\frac{\log\kappa}{2n}
 \label{eq:gp-etazw}
 \\
 \eta_{\X|\Z\W}(\gamma'|\gamma+2\e)\leq\ehA-\frac{\log\kappa}{2n}.
 \label{eq:gp-etaxzw}
\end{gather}
for all sufficiently large $n$.
It should be noted here that
there is such $\gamma$ and $\gamma'$ by assuming (\ref{eq:gp-e1}).

Let $\zz$ be the output of channel side information,
and $\xx$ and $\yy$ be an input and an output of the channel, respectively,
and $\mm$ be a message.
From (\ref{eq:gp-gammap}), we have
$\T_{W|Z,\gamma'}(\zz)\neq\emptyset$ for all $\zz$ and sufficiently
large $n$.
Then we have
\begin{align*}
 |\T_{W|Z,2\e}(\zz)|
 &\geq
 |\T_{W|Z,\gamma'}(\zz)|
 \\
 &\geq 2^{n[H(W|Z)-\eta_{\W|\Z}(\gamma'|\gamma)]}
 \\
 &\geq
 \sqrt{\kappa}2^{n[H(W|Z)-\eA+\eB]}
 \\
 &=
 \sqrt{\kappa}|\W|^{l_{\A}+l_{\B}}
 \\
 &\geq
 \sqrt{\kappa}|\im\A||\im\B|
\end{align*}
for all $\zz\in\T_{Z,\gamma}$ and sufficiently large $n$,
where the first inequality comes from (\ref{eq:gp-gammap}),
the second inequality comes from Lemma \ref{lem:typical-number},
the third inequality comes from (\ref{eq:gp-etazw}),
and the fourth inequality comes from the fact that
$\im\A\subset\W^{l_{\A}}$ and $\im\B\subset\W^{l_{\B}}$.
This implies that
for all $\zz\in\T_{Z,\gamma}$
there is $\T_{W|Z}(\zz)\subset\T_{W|Z,2\e}(\zz)$ such that
\begin{align}
 \sqrt{\kappa}
 &\leq
 \frac{|\T_{W|Z}(\zz)|}{|\im\A||\im\B|}
 \leq 
 2\sqrt{\kappa}
 \label{eq:TZW}
\end{align}
for all $\zz\in\T_{Z,\gamma}$ and sufficiently large $n$.
We assume that $\T_{W|Z}(\zz)$ satisfies the assumption described in
Lemma~\ref{thm:joint-typical}.
Similarly, from (\ref{eq:gp-etaxzw}),
we obtain $\T_{X|ZW}(\zz,\ww)\subset\T_{X|ZW,2\ehA}(\zz,\ww)$ such that
\begin{align}
 \sqrt{\kappa}
 \leq
 \frac{|\T_{X|ZW}(\zz,\ww)|}{|\im\hcA|}
 \leq 
 2\sqrt{\kappa}
 \label{eq:TXZW}
\end{align}
for all $(\zz,\ww)\in\T_{ZW,2\e}$ and sufficiently large $n$.

We define
\begin{align}
 &\bullet \zz\in\T_{Z,\gamma}
 \tag{GP1}
 \\
 &\bullet g_{AB}(\cc,\mm|\zz)\in\T_{W|Z}(\zz)
 \tag{GP2}
 \\
 &\bullet g_{\hA}(\hcc|\zz,g_{AB}(\cc,\mm|\zz))
 \in\T_{X|ZW}(\zz,g_{AB}(\cc,\mm|\zz))
 \tag{GP3}
 \\
 &\bullet
 \yy\in\T_{Y|XZ,\gamma}(g_{\hA}(\hcc|\zz,g_{AB}(\cc,\mm|\zz)),\zz)
 \tag{GP4}
 \\
 &\bullet g_{A}(\cc|\yy)=g_{AB}(\cc,\mm|\zz).
 \tag{GP5}
\end{align}
Under condition (GP5), we have
 \[
 \Decoder(\yy)=Bg_{A}(\cc|\yy)=Bg_{AB}(\cc,\mm|\zz)=\mm,
 \]
which implies that the decoding succeeds.
Then the error probability is upper bounded by
\begin{align}
 &
 \Error_{Y|XZ}(A,B,\hA,\cc,\hcc)
 \notag
 \\*
 &=
 \sum_{\mm,\yy}p_M(\mm)\mu_{Z}(\zz)
 \mu_{Y|XZ}(\yy|g_{\hA}(\hcc|\zz,g_{AB}(\cc,\mm|\zz)),\zz)
 \notag
 \chi(g_{A}(\cc|\yy)\neq g_{AB}(\cc,\mm|\zz))
 \notag
 \\
 &\leq
 p_{MYZ}(\cS_1^c)
 + p_{MYZ}(\cS_1\cap\cS_2^c)
 + p_{MYZ}(\cS_1\cap\cS_2\cap\cS_3^c)
 + p_{MYZ}(\cS_4^c)
 + p_{MYZ}(\cS_1\cap\cS_2\cap\cS_3\cap\cS_4\cap\cS_5^c),
 \label{eq:gp0}
\end{align}
where
\begin{align*}
 \cS_i
 &\equiv
 \lrb{
 (\mm,\yy,\zz): \text{(GP$i$)}
 }.
\end{align*}
Let $\delta$ be an arbitrary positive number.

First we evaluate
$E_{AB\hA C\hC}\lrB{p_{MYZ}(\cS_1^c)}$ and
$E_{AB\hA C\hC}\lrB{p_{MYZ}(\cS_4^c)}$.
From Lemma \ref{lem:typical-prob}, we have
\begin{align}
 E_{AB\hA C\hC}\lrB{p_{MYZ}(\cS_1^c)}
 \leq \frac{\delta}5
 \label{eq:gp-error1}
 \\
 E_{AB\hA C\hC}\lrB{p_{MYZ}(\cS_4^c)}
 \leq \frac{\delta}5
 \label{eq:gp-error4}
\end{align}
for sufficiently large $n$.

Next we evaluate
$E_{ABC}[p_{MYZ}(\cS_1\cap\cS_2^c)]$ and
$E_{\hA\hC}[p_{MYZ}(\cS_1\cap\cS_2\cap\cS_3^c)]$.
From Lemma \ref{thm:joint-typical}, we have
\begin{align}
 E_{ABC}\lrB{p_{MYZ}(\cS_1\cap\cS_2^c)}
 &=
 \sum_{\zz\in\T_{Z,\gamma}}
 \mu_Z(\zz)
 p_{ABCM}\lrsb{\lrb{
 (A,B,\cc,\mm):
 g_{AB}(\cc,\mm|\zz)\notin\T_{W|Z}(\zz)
 }}
 \notag
 \\
 &\leq
 \sum_{\zz\in\T_{Z,\gamma}}p_Z(\zz)
 \lrB{
 \alpha_{AB}-1
 +\frac{|\im\A||\im\B|\lrB{\beta_{AB}+1}}
 {|\T_{W|Z}(\zz)|}
 +\frac{2^{-n\e}|\W|^{l_{\A}+l_{\B}}}
 {|\im\A||\im\B|}
 }
 \notag
 \\
 &\leq
 \alpha_{AB}-1
 +\frac{\beta_{AB}+1}{\sqrt{\kappa}}
 +\frac{2^{-n\e}|\W|^{l_{\A}+l_{\B}}}
 {|\im\A||\im\B|}
 \notag
 \\
 &\leq
 \frac{\delta}5
 \label{eq:gp-error2}
\end{align}
for all sufficiently large $n$,
where the second inequality comes from (\ref{eq:TZW}),
and the last inequality comes from (\ref{eq:gp-k1})
and the properties of $(\aalpha_{AB},\bbeta_{AB})$
and $|\W|^{l_{\A}}/|\im\A|$.
Similarly, by using (\ref{eq:TXZW}),
we have
\begin{align}
 E_{\hA\hC}\lrB{p_{MYZ}(\cS_1\cap\cS_2\cap\cS_3^c)}
 &\leq
 \alpha_{\hA}-1+\frac{\beta_{\hA}+1}{\sqrt{\kappa}}
 +\frac{2^{-n\ehA}|\X|^{l_{\hcA}}}
 {|\im\hcA|}
 \notag
 \\
 &\leq
 \frac{\delta}5
 \label{eq:gp-error3}
\end{align}
for all sufficiently large $n$.

Next,
we evaluate
$E_{AB\hA C\hC}[p_{MYZ}(\cS_1\cap\cS_2\cap\cS_3\cap\cS_4\cap\cS_5^c)]$.
In the following, we assume that
\begin{align*}
 &\bullet \zz\in\T_{Z,\gamma}
 \\
 &\bullet \ww\in\T_{W|Z}(\zz)\subset\T_{W|Z,2\e}(\zz)
 \\
 &\bullet \xx\in\T_{X|ZW}(\zz,\ww)\subset\T_{X|ZW,2\ehA}(\zz,\ww)
 \\
 &\bullet \yy\in\T_{Y|XZ,\gamma}(\xx,\zz)=\T_{Y|XZW,\gamma}(\xx,\zz,\ww)
 \\
 &\bullet g_{A}(\cc|\yy)\neq \ww,
\end{align*}
where
the relation $\T_{Y|XZ,\gamma}(\xx,\zz)=\T_{Y|XZW,\gamma}(\xx,\zz,\ww)$
comes from (\ref{eq:markov-gp}).
From Lemma~\ref{lem:typical-trans},
(\ref{eq:gp-e1}), and (\ref{eq:gp-gamma}), we have
$(\xx,\yy,\zz,\ww)\in\T_{XYXZW,6\ehA}$.
Then there is  $\ww'\in\C_{A}(\cc)$
such that $\ww'\neq \ww$ and
\begin{align*}
 \mu_{W|Y}(\ww'|\yy)
 &\geq
 \mu_{W|Y}(\ww|\yy)
 \\
 &=
 \frac{\mu_{WY}(\ww,\yy)}{\mu_{Y}(\yy)}
 \\
 &\geq
 \frac{2^{-n[H(W,Y)+\zeta_{\Y\W}(6\ehA)]}}
 {2^{-n[H(Y)-\zeta_{\Y}(6\ehA)]}}
 \\
 &\geq
 2^{-n[H(W|Y)+2\zeta_{\Y\W}(6\ehA)]},
\end{align*}
where the second inequality comes from Lemma~\ref{lem:typical-prob}.
This implies that
\[
 \lrB{\G(\yy)\setminus\{\ww\}}\cap\C_A(\cc)\neq\emptyset,
\]
where
\[
 \G(\yy)\equiv
 \lrb{
 \ww':
 \mu_{W|Y}(\ww'|\yy)\geq
 2^{-n[H(W|Y)+2\zeta_{\Y\W}(6\ehA)]}
 }.
\]
Then, we have
\begin{align}
 &E_{AB\hA C\hC}\lrB{
 p_{MYZ}(\cS_1\cap\cS_2\cap\cS_3\cap\cS_4\cap\cS_5^c)
 }
 \notag
 \\*
 &\leq
 E_{AB\hA CM\hC}\left[
 \sum_{\zz\in\T_{Z,\gamma}}\mu_Z(\zz)
 \sum_{\ww\in\T_{W|Z}(\zz)}
 \chi(g_{AB}(\cc,\mm|\zz)=\ww)
 \right.
 \notag
 \\*
 &\qquad\qquad\qquad\quad
 \left.
 \sum_{\xx\in\T_{X|ZW}(\zz,\ww)}
 \chi(g_{\hA}(\hcc|\zz,\ww)=\xx)
 \sum_{\yy\in\T_{Y|XZ,\gamma}(\xx,\zz)}
 \mu_{Y|XZ}(\yy|\xx,\zz)\chi(g_{A}(\cc|\yy)\neq \ww)
 \right]
 \notag
 \\
 &\leq
 E_{AB\hA CM\hC}\left[
 \sum_{\zz\in\T_{Z,\gamma}}\mu_Z(\zz)
 \sum_{\ww\in\T_{W|Z}(\zz)}
 \chi(A\ww=\cc)\chi(B\ww=\mm)
 \right.
 \notag
 \\*
 &\qquad\qquad\qquad\quad
 \left.
 \sum_{\xx\in\T_{X|ZW}(\zz,\ww)}
 \chi(\hA\xx=\hcc)
 \sum_{\yy\in\T_{Y|XZ,\gamma}(\xx,\zz)}
 \mu_{Y|XZ}(\yy|\xx,\zz)
 \vphantom{\sum_{\zz\in\T_{Z,\gamma}}\mu_Z(\zz)}
 \chi(g_{A}(\cc|\yy)\neq \ww)
 \right]
 \notag
 \\
 &=
 \sum_{\zz\in\T_{Z,\gamma}}\mu_Z(\zz)
 \sum_{\ww\in\T_{W|Z}(\zz)}
 \sum_{\xx\in\T_{X|ZW}(\zz,\ww)}
 \sum_{\yy\in\T_{Y|XZ,\gamma}(\xx,\zz)}
 \mu_{Y|XZ}(\yy|\xx,\zz)
 \notag
 \\*
 &\qquad\cdot
 E_{AC}\left[
 \chi(g_{A}(\cc|\yy)\neq \ww)
 \chi(A\ww=\cc)
 E_{B\hA M\hC}\lrB{
 \chi(B\ww=\mm)
 \chi(\hA\xx=\hcc)
 }
 \right]
 \notag
 \\
 &\leq
 \frac 1{|\im\B||\im\hcA|}
 \sum_{\zz\in\T_{Z,\gamma}}\mu_Z(\zz)
 \sum_{\ww\in\T_{W|Z}(\zz)}
 \sum_{\xx\in\T_{X|ZW}(\zz,\ww)}
 \sum_{\yy\in\T_{Y|XZ,\gamma}(\xx,\zz)}
 \mu_{Y|XZ}(\yy|\xx,\zz)
 \notag
 \\*
 &\qquad\cdot
 p_{AC}\lrsb{\lrb{
 (A,\cc):
 \begin{aligned}
	&\lrB{\G(\yy)\setminus\{\ww\}}\cap\C_A(\cc)\neq\emptyset
	\\
	&
	\ww\in\C_A(\cc)
 \end{aligned}
 }}
 \notag
 \\
 &\leq
 \frac 1{|\im\B||\im\hcA|}
 \sum_{\zz\in\T_{Z,\gamma}}\mu_Z(\zz)
 \sum_{\ww\in\T_{W|Z}(\zz)}
 \sum_{\xx\in\T_{X|ZW}(\zz,\ww)}
 \sum_{\yy\in\T_{Y|XZ,\gamma}(\xx,\zz)}
 \mu_{Y|XZ}(\yy|\xx,\zz)
 \notag
 \\*
 &\qquad\cdot
 \lrB{
 \frac{2^{n[H(W|Y)+2\zeta_{\Y\W}(6\ehA)]}\alpha_A}
 {|\im\A|^2}
 +\frac{\beta_A}{|\im\A|}
 }
 \notag
 \\
 &\leq
 \lrB{
 \frac{2^{n[H(W|Y)+2\zeta_{\Y\W}(6\ehA)]}\alpha_A}
 {|\im\A|}
 +\beta_A
 }
 \sum_{\zz\in\T_{Z,\gamma}}\mu_Z(\zz)
 \sum_{\ww\in\T_{W|Z}(\zz)}
 \frac 1{|\im\A||\im\B|}
 \sum_{\xx\in\T_{X|ZW}(\zz,\ww)}
 \frac 1{|\im\hcA|}
 \notag
 \\
 &\leq
 \frac{
 4\kappa|\W|^{l_{\A}}
 2^{-n[\eA-2\zeta_{\Y\W}(6\ehA)]}\alpha_A
 }
 {|\im\A|}
 +4\kappa\beta_A
 \notag
 \\
 &\leq
 \frac{\delta}5.
 \label{eq:gp-error5}
\end{align}
where the third inequality comes from (\ref{lem:E}),
the fourth inequality comes from
Lemma~\ref{lem:ACnotempty}
and the fact that
\[
 |\G(\yy)|\leq 	2^{n[H(W|Y)+2\zeta_{\Y\W}(6\eA)]},
\]
the sixth inequality comes from (\ref{eq:TZW}) and (\ref{eq:TXZW}),
and the last inequality comes from (\ref{eq:gp-e2}), (\ref{eq:gp-k2}),
and the properties of $(\aalpha_A,\bbeta_A)$ and $|\W|^{l_{\A}}/|\im\A|$.

Finally, from (\ref{eq:gp0})--(\ref{eq:gp-error5}),
we have the fact that
for all $\delta>0$ and sufficiently large $n$
there are $A\in\A$, $B\in\B$, $\hA\in\hcA$, $\cc\in\im\A$,
and  $\hcc$ such that
\begin{align*}
 \Error_{Y|XZ}(A,B,\hA,\cc,\hcc)
 \leq
 \delta.
\end{align*}
\hfill\QED

\subsection{Proof of Theorem \ref{thm:wz}}
\label{sec:proof-wz}

We define
\begin{align}
 &\bullet (\xx,\zz)\in\T_{XZ,\gamma}
 \tag{WZ1}
 \\
 &\bullet g_{A}(\cc|\xx)\in\T_{Y|XZ,2\eA}(\xx,\zz)
 \tag{WZ2}
 \\
 &\bullet g_{AB}(\cc,Bg_A(\cc|\xx)|\zz)=g_A(\cc|\xx)
 \tag{WZ3}
\end{align}
and assume that $\gamma>0$ satisfies
\begin{equation}
 \gamma+\sqrt{2\gamma}\log|\X||\Z|<\eA.
	\label{eq:wz-gamma}
\end{equation}

We prove the following lemma.
\begin{lem}
 \label{lem:error-wz}
 For any $(\xx,\zz)$ satisfying (WZ1)
 \begin{align*}
	p_{ABC}\lrsb{\lrb{
	(A,B,\cc): \text{(WZ2), not (WZ3)}
	}}
	&\leq
	\frac{2^{-n[\eB-\eA-2\zeta_{\Y\Z}(3\eA)]}|\Y|^{l_{\A}+l_{\B}}\alpha_{B}}
	{|\im\A||\im\B|}
	+ \beta_{B}.
 \end{align*}
\end{lem}
\begin{proof}
If $(\xx,\zz,A,B,\cc)$ satisfies (WZ1) and (WZ2) but not (WZ3),
there is $\yy'\in\C_{AB}(\cc,Bg_{A}(\cc|\xx))$
such that $\yy'\neq g_A(\cc|\xx)$.
Then, from (\ref{eq:wz-gamma})
and Lemmas \ref{lem:typical-trans} and \ref{lem:typical-aep},
we have $(\xx,g_{A}(\cc|\xx),\zz)\in\T_{YXZ,3\eA}$ and
\begin{align*}
 \mu_{Y|Z}(\yy'|\zz)
 &\geq
 \mu_{Y|Z}(g_{A}(\cc|\xx)|\zz)
 \\*
 &=
 \frac{\mu_{YZ}(g_{A}(\cc|\xx),\zz)}
 {\mu_{Z}(\zz)}
 \\
 &\geq
 \frac{2^{-n[H(Y,Z)+\zeta_{\Y\Z}(3\eA)]}}
 {2^{-n[H(Z)-\zeta_{\Z}(3\eA)]}}
 \\
 &=
 2^{-n[H(Y|Z)+2\zeta_{\Y\Z}(3\eA)]}.
\end{align*}
This implies that
\[
 \lrB{\G\setminus\{g_A(\cc|\xx)\}}\cap\C_A(\cc,Bg_A(\cc|\xx))
 \neq\emptyset,
\]
where
\[
 \G\equiv
 \lrb{
 \yy': \mu_{Y|Z}(\yy'|\zz)\geq
 2^{-n[H(Y|Z)+2\zeta_{\Y\Z}(3\eA)]}
 }.
\]

Let $\yy_{A,\cc}\equiv g_{A}(\cc|\xx)$.
From Lemma~\ref{lem:ABCnoempty}, we have
 \begin{align*}
	p_{ABC}\lrsb{\lrb{
	(A,B,\cc): \text{(WZ2), not (WZ3)}
	}}
	&\leq
	p_{ABC}\lrsb{\lrb{
	(A,B,\cc):
	\lrB{\G\setminus\{\yy_{A,\cc}\}}\cap\C_{AB}(\cc,B\yy_{A,\cc})
	\neq\emptyset
	}}
	\\
	&\leq
	\frac{|\G|\alpha_B}{|\im\A||\im\B|}+\beta_B
	\\
	&\leq
	\frac{2^{n[H(Y|Z)+2\zeta_{\Y\Z}(3\eA)]}\alpha_B}
	{|\im\A||\im\B|}
	+\beta_{B}
	\\
	&=
	\frac{2^{-n[\eB-\eA-2\zeta_{\Y\Z}(3\eA)]}|\Y|^{l_{\A}+l_{\B}}\alpha_{B}}
	{|\im\A||\im\B|}
	+ \beta_{B},
 \end{align*}
where the second inequality comes from Lemma \ref{lem:ABCnoempty}
and the third inequality comes from the fact that
\begin{align*}
 |\G|
 &\leq
 2^{n[H(Y|Z)+2\zeta_{\Y\Z}(3\eA)]}.
\end{align*}
\end{proof}

{\it Proof of Theorem \ref{thm:wz}:}
Let $\Error_{XZ}(A,B,\cc)$ be defined as
\begin{align*}
 \Error_{XZ}(A,B,\cc)
 &\equiv\mu_{XZ}\lrsb{\lrb{
 (\xx,\zz): (\xx,\Decoder(\Encoder(\xx),\zz),\zz)\notin\T_{XYZ,3\eA}
 }}.
\end{align*}
Since
\begin{align*}
 (\xx,\Decoder(\Encoder(\xx),\zz),\zz)
 &=(\xx,g_{AB}(\cc,Bg_A(\cc|\xx)),\zz)
 \\
 &=
 (\xx,g_{A}(\cc|\xx),\zz)
 \\
 &\in
 \T_{XYZ,3\eA}.
\end{align*}
under conditions (WZ1)--(WZ3), then we have
\begin{align}
 \begin{split}
 \Error_{XZ}(A,B,\cc)
 &\leq
 \mu_{XZ}(\cS_1^c)+\mu_{XZ}(\cS_1\cap\cS_2^c)
 +\mu_{XZ}(\cS_1\cap\cS_2\cap\cS_3^c),
 \end{split}
 \label{eq:wz0}
\end{align}
where
\begin{align*}
 \cS_i
 &\equiv
 \lrb{
 (\xx,\zz): \text{(WZ$i$)}
 }.
\end{align*}
Let $\delta>0$ be an arbitrary positive number.

First, we evaluate $E_{ABC}\lrB{\mu_{XZ}(\cS_1^c)}$.
From Lemma \ref{lem:typical-prob}, we have
\begin{align}
 E_{ABC}\lrB{\mu_{XZ}(\cS_1^c)}
 &
 \leq
 \frac{\delta}3
 \label{eq:wz-error1}
\end{align}
for all sufficiently large $n$.

Next, we evaluate $E_{ABC}\lrB{\mu_{XZ}(\cS_1\cap\cS_2^c)}$.
From (\ref{eq:markov-wz}), we have
\begin{align*}
 \mu_{XYZ}(\xx,\yy,\zz)
 &=
 \mu_{XZ}(\xx,\zz)\mu_{Y|X}(\yy|\xx)
 \\
 &=
 \frac{\mu_{XZ}(\xx,\zz)\mu_{XY}(\xx,\yy)}
 {\mu_{X}(\xx)}
 \\
 &=\mu_{XY}(\xx,\yy)\mu_{Z|X}(\zz|\xx)
\end{align*}
and
\begin{align*}
 \arg\max_{\yy'\in\C_B(\cc)}\mu_{XY}(\xx,\yy')
 &=\arg\max_{\yy'\in\C_B(\cc)}\mu_{XY}(\xx,\yy')\mu_{Z|X}(\zz|\xx)
 \\
 &=\arg\max_{\yy'\in\C_B(\cc)}\mu_{XYZ}(\xx,\yy',\zz).
\end{align*}
This implies that
ML coding by using  $\mu_{XY}$
is equivalent that using $\mu_{XYZ}$.
Since $\gamma>0$ satisfies
(\ref{eq:wz-gamma}),
we have the fact that
there is $\gamma'>0$ such that
\[
 \eta_{\Y|\X\Z}(\gamma'|\gamma)\leq\eA-\gamma
\]
for all sufficiently large $n$.
We have $\T_{Y|XZ,\gamma'}(\xx,\zz)\neq\emptyset$
for all $(\xx,\zz)\in\T_{XZ,\gamma}$ and sufficiently large $n$.
Then, from Lemma~\ref{lem:typical-number}, we have
\begin{align*}
 |\T_{Y|XZ,2\eA}(\xx,\zz)|
 &\geq
 |\T_{Y|XZ,\gamma'}(\xx,\zz)|
 \\
 &\geq 2^{n[H(Y|XZ)-\eta_{\Y|\X\Z}(\gamma'|\gamma)]}
 \\
 &\geq 2^{n[H(Y|XZ)-\eA+\gamma]}
 \\
 &\geq
 2^{n\gamma}|\Z|^{l_{\A}}
 \\
 &\geq
 2^{n\gamma}|\im\A|
\end{align*}
for all $(\xx,\zz)\in\T_{XZ,\gamma}$ and sufficiently large $n$.
This implies that
there is $\T(\xx,\zz)\subset\T_{Y|XZ,2\eA}$ such that
\begin{align}
 |\T(\xx,\zz)|
 \geq
 2^{n\gamma}{|\im\A|}.
 \label{eq:wz-T}
\end{align}
We assume that $\T(\xx,\zz)$ satisfies the assumption described in
Lemma~\ref{thm:joint-typical}.
Then from Lemma \ref{thm:joint-typical}, we have
\begin{align}
 E_{ABC}\lrB{\mu_{XZ}(\cS_1\cap\cS_2^c)}
 &=
 \sum_{(\xx,\zz)\in\T_{XZ,\gamma}}\mu_{XZ}(\xx,\zz)
 p_{AC}\lrsb{\lrb{
 (A,\cc): g_{A}(\cc|\xx)\notin\T_{Y|XZ,2\eA}(\xx,\zz)
 }}
 \notag
 \\
 &\leq
 \sum_{(\xx,\zz)\in\T_{XZ,\gamma}}\mu_{XZ}(\xx,\zz)
 \left[
 \alpha_A-1
 +\frac{|\im\A|\lrB{\beta_A+1}}{|\T(\xx,\zz)|}
 +\frac{2^{-n\eA}|\Y|^{l_{\A}}}{|\im\A|}
 \right]
 \notag
 \\
 &\leq
 \sum_{(\xx,\zz)\in\T_{XZ,\gamma}}\mu_{XZ}(\xx,\zz)
 \left[
 \alpha_A-1
 +2^{-n\gamma}\lrB{\beta_A+1}
 +\frac{2^{-n\eA}|\Y|^{l_{\A}}}{|\im\A|}
 \right]
 \notag
 \\
 &\leq
 \alpha_A-1
 +2^{-n\gamma}\lrB{\beta_A+1}
 +\frac{2^{-n\eA}|\Y|^{l_{\A}}}{|\im\A|}
 \notag
 \\
 &\leq
 \frac {\delta}3
 \label{eq:wz-error2}
\end{align}
for all sufficiently large $n$,
where the third inequality comes form (\ref{eq:wz-T}),
and the last inequality comes from the properties 
of $(\aalpha_A,\bbeta_A)$ and $|\Y|^{l_{\A}}/|\im\A|$.

Finally, we evaluate $E_{ABC}\lrB{\mu_{XZ}(\cS_1\cap\cS_2\cap\cS_3^c)}$.
From Lemma \ref{lem:error-wz}, we have
 \begin{align}
	E_{ABC}\lrB{
	\mu_{XZ}(\cS_1\cap\cS_2\cap\cS_3^c)
	}
	&\leq
	\sum_{(\xx,\zz)\in\T_{XZ,\gamma}}\mu_{XZ}(\xx,\zz)
	p_{ABC}\lrsb{\lrb{
	(A,B,\cc):
	\text{(WZ2), not (WZ3)}
	}}
	\notag
	\\
	&\leq
	\frac{2^{-n[\eB-\eA-2\zeta_{\Y\Z}(3\eA)]}|\Y|^{l_{\A}+l_{\B}}\alpha_{B}}
	{|\im\A||\im\B|}
	+ \beta_{B}
	\notag
	\\
	&\leq
	\frac{\delta}3
	\label{eq:wz-error3}
 \end{align}
for sufficiently large $n$, where the last inequality comes from
(\ref{eq:wz-e}) and the properties of $(\aalpha_B,\bbeta_B)$,
$|\Y|^{l_{\A}}/|\im\A|$ and  $|\Y|^{l_{\B}}/|\im\B|$.

From (\ref{eq:wz0})--(\ref{eq:wz-error3}),
we have the fact that for any $\delta>0$ and for all sufficiently large $n$
there are $A$, $B$, and $\cc$ such that
\[
\Error_{XY}(A,B,\cc)\leq \delta.
\]
From Lemma~\ref{lem:type}, we have
 \begin{align*}
	\frac{\rho_n(\xx,f_n(\yy,\zz))}n
	&
	=
	\sum_{(x,y,z)\in\X\times\Y\times\Z}
	\nu_{\xx\yy\zz}(x,y,z)\rho(x,f(y,z))
	\\
	&
	\leq
	\sum_{(x,y,z)\in\X\times\Y\times\Z}
	\lrB{
	\mu_{XYZ}(x,y,z) + \sqrt{6\eA}
	}
	\rho(x,f(y,z))
	\\
	&
	\leq
	\sum_{(x,y,z)\in\X\times\Y\times\Z}
	\mu_{XYZ}(x,y,z)\rho(x,f(y,z))
	+ |\X||\Y||\Z|\rho_{\max}\sqrt{6\eA}
	\\
	&
	=
	E_{XYZ}[\rho(X,f(Y,Z))] + |\X||\Y||\Z|\rho_{\max}\sqrt{6\eA}
 \end{align*}
for	$(\xx,\yy,\zz)\in\T_{XYZ,3\eA}$.
Then we have
\begin{align*}
 \frac{E_{XZ}\lrB{\rho_n(X^n,f_n(\Decoder(\Encoder(X^n),Z^n),Z^n))}}n
 &\leq
 E_{XYZ}\lrB{\rho(X,f(Y,Z))}
 +|\X||\Y||\Z|\rho_{\max}\sqrt{6\eA}
 +\delta\rho_{\max}
 \\
 &\leq
 E_{XYZ}\lrB{\rho(X,f(Y,Z))}
 +3|\X||\Y||\Z|\rho_{\max}\sqrt{\eA}
\end{align*}
for all sufficiently large $n$
by letting
\[
 \delta\leq [3-\sqrt{6}]|\X||\Y||\Z|\sqrt{\eA}.
\]
\hfill\QED

\subsection{Proof of Theorem \ref{thm:psi}}
\label{sec:proof-psi}

We define
\begin{align}
 &\bullet (\xx,\yy)\in\T_{XY,\gamma}
 \tag{OHO1}
 \\
 &\bullet g_{A}(\cc|\yy)\in\T_{Z|XY,2\eA}(\xx,\yy)
 \tag{OHO2}
 \\
 &\bullet g_{AB}(\cc,Bg_A(\cc|\yy))=g_A(\cc|\yy)
 \tag{OHO3}
 \\
 &\bullet g_{\hB}(\hB\xx|g_{AB}(\cc,Bg_A(\cc|\yy)))=\xx
 \tag{OHO4}
\end{align}
and assume that $\gamma>0$ satisfies
\begin{equation}
 \gamma+\sqrt{2\gamma}\log|\X||\Y|<\eA.
	\label{eq:psi-gamma}
\end{equation}

We prove the following lemma.
\begin{lem}
\label{lem:error-psi}
 For any $(\xx,\yy)$ satisfying (OHO1),
 \begin{align*}
	p_{AB\hB C}\lrsb{\lrb{
	(A,B,\hB,\cc): \text{(OHO2),(OHO3), not (OHO4)}
	}}
	&\leq
	\frac{2^{-n[\ehB-2\zeta_{\X\Z}(3\eA)]}
	|\X|^{l_{\hcB}}\alpha_{\hB}}
	{|\im\hcB|}
	+\beta_{\hB}.
 \end{align*}
\end{lem}

\begin{proof}
We define
\begin{align*}
 \xx_{A,B,\hB,\cc}&\equiv g_{\hB}(\hB\xx|g_{AB}(\cc,Bg_A(\cc,\xx)))
 \\
 \zz_{A,B,\cc}&\equiv g_{AB}(\cc,Bg_A(\cc,\xx)).
\end{align*}
Assume that conditions (OHO1)-(OHO3) are satisfied but (OHO4) is not.
From Lemma~\ref{lem:typical-trans} and (\ref{eq:psi-gamma}), we have
$(\xx,\yy,g_{A}(\cc,\yy))\in\T_{XYZ,3\eA}$
and there is $\xx'\in\C_{\hB}(\hB\xx)$
such that $\xx'\neq g_{\hB}(\hB\xx,g_{AB}(\cc,Bg_A(\cc|\xx)))$.
From Lemma~\ref{lem:typical-prob}, we have
\begin{align*}
 \mu_{X|Z}(\xx'|\zz_{A,B,\cc})
 &\geq
 \mu_{X|Z}(\xx_{A,B,\hB,\cc}|\zz_{A,B,\cc})
 \\*
 &=
 \frac{\mu_{XZ}(\xx_{A,B,\hB,\cc},\zz_{A,B,\cc})}
 {\mu_Z(\zz)}
 \\
 &\geq
 \frac{2^{-n[H(XZ)+\zeta_{\X\Z}(3\eA)]}}
 {2^{-n[H(Z)-\zeta_{\Z}(3\eA)]}}
 \\
 &=
 2^{-n[H(X|Z)+2\zeta_{\X\Z}(3\eA)]}.
\end{align*}
This implies that
\[
 \lrB{\G(\zz_{A,B,\cc})\setminus\{\xx_{A,B,\hB,\cc}\}}
 \cap\C_{\hB}(\hB\xx)
 \neq\emptyset,
\]
where
\begin{align*}
 \G(\zz)
 &\equiv
 \lrb{
 \xx':
 \mu_{X|Z}(\xx'|\zz)
 \geq
 2^{-n[H(X|Z)+2\zeta_{\X\Z}(3\eA)]}
 }.
\end{align*}

From Lemma~\ref{lem:Anotempty}, we have
\begin{align*}
 &p_{AB\hB C}\lrsb{\lrb{
 (A,B,\hB,\cc): \text{(OHO2),(OHO3), not (OHO4)}
 }}
 \\*
 &\leq
 p_{AB\hB C}\lrsb{\lrb{
 (A,B,\hB,\cc):
 \lrB{\G(\zz_{A,B,\cc})\setminus\{\xx_{A,B,\hB,\cc}\}}
 \cap\C_{\hB}(\hB\xx_{A,B,\hB,\cc})
 \neq \emptyset
 }}
 \\
 &=
 \sum_{A,B,\cc}p_{A,B,\cc}(A,B,\cc)
 p_{\hB}\lrsb{\lrb{
 \hB:
 \lrB{\G(\zz_{A,B,\cc})\setminus\{\xx_{A,B,\hB,\cc}\}}
 \cap\C_{\hB}(\hB\xx_{A,B,\hB,\cc})
 \neq \emptyset
 }}
 \\
 &\leq
 \sum_{A,B,\cc}p_{A,B,\cc}(A,B,\cc)
 \lrB{
 \frac{\left|
 \G(\zz_{A,B,\cc})\setminus\{\xx_{A,B,\hB,\cc}\}
 \right|\alpha_{\hB}}
 {|\im\hcB|}
 +\beta_{\hB}
 }
 \\
 &\leq
 \frac{2^{n[H(X|Z)+2\zeta_{\X\Z}(3\eA)]}\alpha_{\hB}}
 {|\im\hcB|}
 +\beta_{\hB}
 \\
 &=
 \frac{2^{-n[\ehB-2\zeta_{\X\Z}(3\eA)]}
 |\X|^{l_{\hcB}}\alpha_{\hB}}
 {|\im\hcB|}
 +\beta_{\hB},
\end{align*}
where the last inequality comes from the fact that
\begin{align*}
 |\G(\zz_{A,B,\cc})|
 &\leq
 2^{n[H(X|Z)+2\zeta_{\X\Z}(3\eA)]},
\end{align*}
for all $A$, $B$ and $\cc$.
\end{proof}

{\it Proof of Theorem \ref{thm:psi}:}
Under the conditions (OHO1)--(OHO4), we have
\begin{align*}
 \Decoder(\Encoder_X(\xx),\Encoder_Y(\yy))
 &=g_{\hB}(\hB\xx,g_{AB}(\cc,Bg_A(\cc|\yy)))
 \\
 &=\xx.
\end{align*}
Then the decoding error probability
is upper bounded by
\begin{align}
 \Error_{XY}(A,B,\hB,\cc)
 \leq
 \mu_{XY}(\cS_1^c)+\mu_{XY}(\cS_2^c)
 +\mu_{XY}(\cS_1\cap\cS_2\cap\cS_3^c)
 +\mu_{XY}(\cS_1\cap\cS_2\cap\cS_3\cap\cS_4^c)
 \label{eq:psi0}
\end{align}
where we define
\begin{align*}
 \cS_i
 &\equiv
 \lrb{
 (\vv,\xx): \text{(OHO$i$)}
 }.
\end{align*}

From (\ref{eq:markov-psi}), we have
\begin{align*}
 \mu_{XYZ}(\xx,\yy,\zz)
 &=
 \mu_{XY}(\xx,\yy)\mu_{Z|Y}(\zz|\yy)
 \\
 &=
 \frac{\mu_{XY}(\xx,\yy)\mu_{YZ}(\yy,\zz)}
 {\mu_{Y}(\yy)}
 \\
 &=\mu_{X|Y}(\xx|\yy)\mu_{YZ}(\yy,\zz)
\end{align*}
and
\begin{align*}
 \arg\max_{\zz'\in\C_A(\cc)}\mu_{Z|Y}(\zz'|\yy)
 &=\arg\max_{\zz'\in\C_A(\cc)}\mu_{YZ}(\yy,\zz')
 \\
 &=\arg\max_{\zz'\in\C_A(\cc)}\mu_{YZ}(\yy,\zz')\mu_{X|Y}(\xx|\yy)
 \\
 &=\arg\max_{\zz'\in\C_A(\cc)}\mu_{XYZ}(\xx,\yy,\zz')
 \\
 &=\arg\max_{\zz'\in\C_A(\cc)}\mu_{Z|XY}(\zz'|\xx,\yy)
\end{align*}
This implies that
ML coding by using  $\mu_{Z|Y}$
is equivalent to that using $\mu_{Z|XY}$.
By 
applying a similar argument to that in the proof of
Theorem~\ref{thm:wz},
we have
\begin{align}
 E_{AB\hB C}\lrB{\mu_{XY}(\cS_1^c)}
 &\leq
 \frac{\delta}4
 \label{eq:psi-error1}
 \\
 E_{AB\hB C}\lrB{\mu_{XY}(\cS_2^c)}
 &\leq
	\alpha_A-1
	+2^{-n\gamma}\lrB{\beta_A+1}
 +\frac{2^{-n\eA}|\Z|^{l_{\A}}}{|\im\A|}
 \notag
 \\
 &\leq
 \frac{\delta}4
 \label{eq:psi-error2}
 \\
 E_{AB\hB C}\lrB{\mu_{XY}(\cS_1\cap\cS_2\cap\cS_3^c)}
 &\leq
 \frac{2^{-n[\eB-\eA-\zeta_{\Z}(3\eA)]}|\Z|^{l_{\A}+l_{\B}}\alpha_{AB}}
 {|\im\A||\im\B|}
 + \beta_{AB}
 \notag
 \\
 &\leq
 \frac{\delta}4
 \label{eq:psi-error3}
\end{align}
for all sufficiently large $n$ by
assuming (\ref{eq:psi-e1}) and (\ref{eq:psi-gamma}).
Furthermore, from Lemma \ref{lem:error-psi}, we have
 \begin{align}
	&
	E_{AB\hB C}\lrB{
	\mu_{XY}(\cS_1\cap\cS_2\cap\cS_3\cap\cS_4^c)
	}
	\notag
	\\*
	&\leq
	\sum_{(\xx,\yy)\in\T_{XY,\gamma}}\mu_{XY}(\xx,\yy)
	p_{AB\hB C}\lrsb{\lrb{
	(A,B,\cc): \text{(OHO2),(OHO3), not (OHO4)}
	}}
	\notag
	\\
	&\leq
	\sum_{(\xx,\yy)\in\T_{XY,\gamma}}\mu_{XY}(\xx,\yy)
	\lrB{
	\frac{2^{-n[\ehB-2\zeta_{\X\Z}(3\eA)]}
	|\X|^{l_{\hcB}}\alpha_{\hB}}
	{|\im\hcB|}
	+\beta_{\hB}
	}
	\notag
	\\
	&\leq
	\frac{2^{-n[\ehB-2\zeta_{\X\Z}(3\eA)]}
	|\X|^{l_{\hcB}}\alpha_{\hB}}
	{|\im\hcB|}
	+\beta_{\hB}
	\notag
	\\
	&\leq
	\frac{\delta}4
	\label{eq:psi-error4}
 \end{align}
for all sufficiently large $n$,
where the last inequality comes from
(\ref{eq:psi-e2}) and the properties of
$(\aalpha_{\hB},\bbeta_{\hB})$ and $|\X|^{l_{\hcB}}/|\im\hB|$.

From (\ref{eq:psi0})--(\ref{eq:psi-error4}),
we have the fact that
for all $\delta>0$ and sufficiently large $n$
there are $A$, $B$, $\hB$, and $\cc$ such that
\[
\Error_{XY}(A,B,\hB,C)\leq \delta.
\]
\hfill\QED

\section{Conclusion}
In this paper we introduced the notion of the hash property of an
ensemble of
functions and proved that
an ensemble of $q$-ary sparse matrices satisfies the hash property.
Based on this property, we proved
the achievability of the coding theorems
for the Slepian-Wolf problem, the Gel'fand-Pinsker problem,
the Wyner-Ziv problem, and the One-helps-one problem.
This implies that the rate of codes using sparse matrices combined with ML
coding can achieve the optimal rate.
We believe that
the hash property is essential for coding problems
and our theory can also be applied
to other ensembles of functions suitable for
efficient coding algorithms.
In other words,
it is enough to prove the hash property of a new ensemble
to obtain several coding theorems.
It is a future challenge to derive the performance of codes
 when ML coding is replaced by one of these efficient algorithms given in
\cite{GDL}\cite{KFL01}\cite{FWK05}.
It is also a future challenge to apply the hash property
to other coding problems.
For example, there are studies of the fixed-rate universal source coding
and the fixed-rate universal channel coding \cite{HASH-UNIV},
and the wiretap channel coding and the secret key agreement \cite{HASH-CRYPT}.

\appendix

\subsection*{Method of Types}
\label{sec:type-theory}

We use the following lemmas for a set of typical sequences.
It should be noted that
our definition of a set of typical sequences
is introduced in \cite{HK89}\cite{UYE}
and differs from that defined in \cite{CK}\cite{CT}\cite{HK02}\cite{Y06}.

\begin{lem}[{\cite[Lemma 2.6]{CK}}]
 \label{lem:exprob}
\begin{align*}
 \frac 1n\log \frac 1{\mu_{U}(\uu)}
 &= H(\nu_{\uu})+D(\nu_{\uu}\|\mu_{U})
 \\
 \frac 1n\log\frac 1{\mu_{U|V}(\uu|\vv)}
 &= H(\nu_{\uu|\vv}|\nu_{\vv})
 +D(\nu_{\uu|\vv}\|\mu_{U|V}|\nu_{\vv}).
\end{align*}
\end{lem}

\begin{lem}[{\cite[Theorem 2.5]{UYE}}]
 \label{lem:typical-trans}
 If $\vv\in\T_{V,\gamma}$ and $\uu\in\T_{U|V,\gamma'}(\vv)$,
 then $(\uu,\vv)\in\T_{UV,\gamma+\gamma'}$.
 If $(\uu,\vv)\in\T_{UV,\gamma}$, then $\uu\in\T_{U,\gamma}$.
\end{lem}
\begin{proof}
The first statement can be proved from the fact that
\begin{equation}
 \label{eq:divergence}
 \begin{split}
	D(\nu_{\uu\vv}\parallel \mu_{UV})
	=
	D(\nu_{\vv}\parallel \mu_V)
	+D(\nu_{\uu|\vv}\parallel \mu_{U|V} | \nu_{\vv}).
 \end{split}
\end{equation}
The second statement can be proved from the fact that
\[
 D(\nu_{\vv}\parallel \mu_V)\leq
 D(\nu_{\uu,\vv}\parallel \mu_{UV}),
\]
which is derived from
(\ref{eq:divergence})
and the non-negativity of the divergence.
\end{proof}

\begin{lem}[{\cite[Theorem 2.6]{UYE}}]
 \label{lem:type}
 If $\uu\in\T_{U,\gamma}$, then
 \begin{gather*}
	\left|
	\nu_{\uu}(u)-\mu_U(u)
	\right|
	\leq \sqrt{2\gamma},\quad\text{for all $u\in\U$},
	\\
	\nu_{\uu}(u)=0,\quad\text{if $\mu_{U}(u)=0$}.
 \end{gather*}
\end{lem}
\begin{proof}
The lemma can be proved directly from the fact that
\[
 \sum_{u\in\U}|\nu(u)-\mu_U(u)|\leq\sqrt{\frac{2D(\nu\parallel
 \mu_U)}{\log_2 e}},
\]
where $e$ is the base of the natural logarithm (see~\cite[Lemma 12.6.1]{CT}).
\end{proof}

\begin{lem}[{\cite[Theorem 2.7]{UYE}}]
 \label{lem:typical-aep}
 Let $0<\gamma\leq 1/8$.
 Then,
 \begin{align}
	\begin{split}
	\left|
	\frac 1{n}\log_2\frac 1{\mu_{U}(\uu)} - H(U)
	\right|
	&\leq
	\zeta_{\U}(\gamma)
	\end{split}
	\label{eq:typical-aep}
 \end{align}
 for all $\uu\in\T_{U,\gamma}$, and
 \begin{align}
	\begin{split}
	\left|
	\frac 1{n}\log_2\frac 1{\mu_{U|V}(\uu|\vv)} - H(U|V)
	\right|
	&\leq
	\zeta_{\U|\V}(\gamma'|\gamma)
	\end{split}
	\label{eq:joint-typical-aep}
 \end{align}
 for $\vv\in\T_{V,\gamma}$ and $\uu\in\T_{U|V,\gamma'}(\vv)$,
 where $\zeta_{\U}(\gamma)$ and $\zeta_{\U|\V}(\gamma'|\gamma)$
 are defined in (\ref{eq:zeta}) and (\ref{eq:zetac}), respectively.
\end{lem}
\begin{proof}
From Lemma \ref{lem:exprob}, we have
\begin{align*}
 \left|
	\frac 1{n}\log_2\frac 1{\mu_V(\vv)} - H(V)
	\right|
 \leq
 D(\nu_{\vv}\parallel \mu_V)+|H(\nu_{\vv})-H(V)|.
\end{align*}
We have (\ref{eq:typical-aep}) from \cite[Lemma 2.7]{CK}.

From Lemmas \ref{lem:exprob} and \ref{lem:type}, we have
\begin{align*}
 &
 \left|
 \frac 1{n}\log_2\frac 1{\mu_{U|V}(\uu|\vv)} - H(U|V)
 \right|
 \\*
 &
 \leq
 D(\nu_{\uu|\vv}\parallel \mu_{U|V}|\nu_{\vv})
 +|H(\nu_{\uu|\vv}|\nu_{\vv})-H(\mu_{U|V}|\nu_{\vv})|
 +|H(\mu_{U|V}|\nu_{\vv})-H(U|V)|,
\end{align*}
and
\begin{align*}
 |H(\mu_{U|V}|\nu_{\vv})-H(U|V)|
 &\leq
 \sqrt{2\gamma}\log_2|\U|,
\end{align*}
respectively.
We have (\ref{eq:joint-typical-aep}) from the above inequalities
 and \cite[Lemma 2.7]{CK}.
\end{proof}

\begin{lem}[{\cite[Theorem 2.8]{UYE}}]
 \label{lem:typical-prob}
 For any $\gamma>0$,and $\vv\in\V^n$,
 \begin{align*}
	\mu_U([\T_{U,\gamma}]^c)
	&\leq
	2^{-n[\gamma-\lambda_{\U}]}
	\\
	\mu_{U|V}([\T_{U|V,\gamma}(\vv)]^c|\vv)
	&\leq
	2^{-n[\gamma-\lambda_{\U\V}]},
 \end{align*}
 where $\lambda_{\U}$ and $\lambda_{\U\V}$ are defined in (\ref{eq:lambda}).
\end{lem}
\begin{proof}
The lemma can be proved from \cite[Lemma 2.2]{CK} and 
\cite[Lemma 2.6]{CK}.
\end{proof}

\begin{lem}[{\cite[Theorem 2.9]{UYE}}]
 \label{lem:typical-number}
 For any $\gamma>0$, $\gamma'>0$, and $\vv\in\T_{V,\gamma}$,
 \begin{align*}
	\left|
	\frac 1{n}\log_2 |\T_{U,\gamma}| - H(U)
	\right|
	&\leq
	\eta_{\U}(\gamma)
	\\
	\left|
	\frac 1{n}\log_2 |\T_{U|V,\gamma'}(\vv)| - H(U|V)
	\right|
	&\leq
	\eta_{\U|\V}(\gamma'|\gamma),
	\end{align*}
 where $\eta_{\U}(\gamma)$ and $\eta_{\U|\V}(\gamma'|\gamma)$
 are defined in (\ref{eq:def-eta}) and (\ref{eq:def-etac}), respectively.
\end{lem}
\begin{proof}
The lemma can be proved in the same way as the proof of \cite[Lemma
 2.13]{CK}.
\end{proof}

\section*{Acknowledgements}
This paper was written while one of authors J.\ M.\ was a visiting
researcher at ETH, Z\"urich.
He wishes to thank Prof.\ Maurer for arranging for his stay.
Constructive comments, suggestions,
and refferences by anonymous reviewers of IEEE Transactions on
Information Theory
have significantly improved the presentation of our results.

\end{document}